\documentclass[conference]{IEEEtran}
\IEEEoverridecommandlockouts

\usepackage{booktabs} 
\usepackage{times}
\usepackage{color}
\usepackage{balance}
\usepackage{url} 

\usepackage{tabularx}

\usepackage{siunitx} 

\usepackage{graphicx}
\usepackage{epstopdf} 
\usepackage{epsfig,tabularx,subfigure,multirow}
\usepackage{enumitem}
\usepackage{caption}
\usepackage{graphicx}
\usepackage{xcolor}
\usepackage{algorithmicx}
\usepackage{algpseudocode}
\usepackage{amsfonts}

\usepackage{hyperref}

\usepackage{makecell}
\usepackage{amsthm,amsmath}  

\newcolumntype{L}[1]{>{\raggedright\arraybackslash}p{#1}}
\newcolumntype{C}[1]{>{\centering\arraybackslash}p{#1}}
\newcolumntype{R}[1]{>{\raggedleft\arraybackslash}p{#1}}

\usepackage[linesnumbered,ruled,vlined]{algorithm2e}
\SetKwRepeat{Do}{do}{while}
 
\SetCommentSty{mycommfont}

\long\def\comment#1{}

\newcommand{\nop}[1]{}

\newtheorem{theorem}{\bf Theorem}[section]
\newtheorem{lemma}{\bf Lemma}[section]
\newtheorem{example}{\bf Example}

\theoremstyle{definition}
\newtheorem{definition}{\bf Definition}

\newcommand{\revision}[1]{\color{black}{#1} \color{black}}

\begin{document}
\title{StructRide: A Framework to Exploit the Structure Information of Shareability Graph in Ridesharing}

\author{
	{Jiexi Zhan{\small$~^{1}$}, Yu Chen{\small$~^{1}$}, Peng Cheng{\small$~^{1}$}, Lei Chen{\small$~^{2, 3}$}, Wangze Ni{\small$~^{4*}$}, Xuemin Lin{\small$~^{5}$}}\\
	\fontsize{10}{10}\itshape
	$~^{1}$East China Normal University, Shanghai, China;
	$~^{2}$HKUST (GZ), Guangzhou, China;
	$~^{3}$HKUST, Hong Kong SAR, China;\\
	$~^{4}$Zhejiang University, Hangzhou, China; $~^{5}$Shanghai Jiaotong University, Shanghai, China\\
	\fontsize{9}{9}\upshape
	\{jxzhan, yu.chen\}@stu.ecnu.edu.cn; pcheng@sei.ecnu.edu.cn; leichen@cse.ust.hk; niwangze@zju.edu.cn; xuemin.lin@gmail.com
	\thanks{*Wangze Ni is also with The State Key Laboratory of Blockchain and Data Security; Hangzhou High-Tech Zone (Binjiang) Institute of Blockchain and Data Security.}
}

\maketitle

\begin{abstract}
	Ridesharing services play an essential role in modern transportation, which significantly reduces traffic congestion and exhaust pollution. 
	In the ridesharing problem, improving the sharing rate between riders can not only save the travel cost of drivers but also utilize vehicle resources more efficiently. 
	The existing online-based and batch-based methods for the ridesharing problem lack the analysis of the sharing relationship among riders, leading to a compromise between efficiency and accuracy.
	In addition, the graph is a powerful tool to analyze the structure information between nodes. 
	Therefore, in this paper, we propose a framework, namely StructRide, to utilize the structure information to improve the results for ridesharing problems. 
	Specifically, we extract the sharing relationships between riders to construct a shareability graph. 
	Then, we define a novel measurement, namely shareability loss, for vehicles to select groups of requests such that the unselected requests still have high probabilities of sharing with other requests. 
	Our SARD algorithm can efficiently solve dynamic ridesharing problems to achieve dramatically improved results. 
	Through extensive experiments, we demonstrate the efficiency and effectiveness of our SARD algorithm on two real datasets. 
	Our SARD can run up to 72.68 times faster and serve up to 50\% more requests than the state-of-the-art algorithms.
\end{abstract}

\begin{IEEEkeywords}
		Ridesharing, Shareability Graph, Algorithms
\end{IEEEkeywords}

\section{Introduction}
\label{sec:introduction}
Recently, ridesharing has become a popular public transportation choice, which greatly reduces energy and relieves traffic pressure. 
In ridesharing, a driver can serve different riders simultaneously once riders can share parts of their trips and there are enough seats for them.
The ridesharing service providers (e.g., Uber~\cite{Uber}, Didi~\cite{DiDi}) are constantly pursuing better service quality, such as higher service rate~\cite{Online-ridesharing,flexible-realtime,luo2020}, lower total travel distance~\cite{T-share,Kinetic_tree,luo2020}, or higher total revenue~\cite{price-aware,online-truthful}.
For ridesharing platforms, \textit{vehicle-request matching} and \textit{route planning} are two critical issues to address.
With a set of vehicles and requests, vehicle-request matching filters out a set of valid candidate requests for each vehicle, while route planning targets on designing a route schedule for a vehicle to pick up and drop off the assigned requests. 
If a vehicle can serve two requests on the same trip, we call that they have a shareable relationship.

The structure information of graphs reveal key properties (e.g., $k$-core~\cite{bestk_core,radiu_kcore}, $k$-truss~\cite{truss_query,truss_decomp}, clique~\cite{spatial_clique,clique_listing}), aiding in specific applications (e.g., community discovery~\cite{community_search,colocated_community,core_cluster}, anomaly detection~\cite{anomalies,anomaly_detection} and influential analysis~\cite{influential_spreader,influential_community}).
In ridesharing problems, the shareable relationships form a \textit{shareability graph} \cite{santi2014quantifying, ma2017demand, chen2022p, lin2016model, pnas, kucharski2020exact}, where each node represents a request and each edge indicates that the connected requests can share part of their trips.
However, to the best of our knowledge, no existing works have targeted improving service quality by utilizing the \textit{structural properties} of the shareability graph in ridesharing problems.

\begin{table}[t!]
	\centering
	{\small
		\caption{ Requests Release Detail.} 
		\label{example_table}
		\begin{tabular}{c|cccc}
			\textbf{request} & \textbf{source} & \textbf{destination} & \makecell{\textbf{release} \\ \textbf{time}} & \textbf{deadline} \\ \hline			\hline
			$r_1$   & a      & d           & 0            & 30       \\ 
			$r_2$   & c      & f           & 1            & 19       \\ 
			$r_3$   & b      & e           & 2            & 21       \\ 
			$r_4$   & c      & g           & 3            & 21        \\ \hline
		\end{tabular}
	}
\end{table}

\begin{figure}[t!]
	\centering
	\subfigure[ Road Network]{
		\scalebox{0.33}[0.33]{\includegraphics{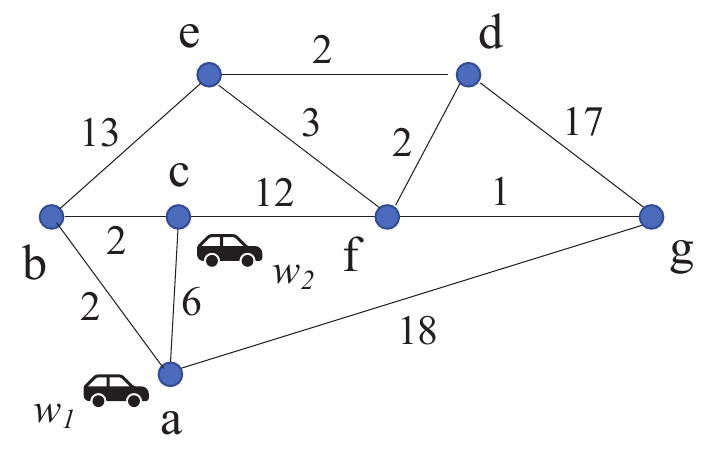}}
		\label{subfig:roadnet_example}    
	}
	\subfigure[Shareability Graph]{
		\scalebox{0.6}[0.6]{\includegraphics{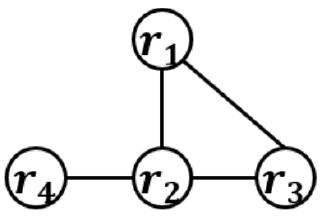}}
		\label{subfig:shareablenet_example}    
	}
	\caption{A Motivation Example }
	\label{fig:motivation_example}
\end{figure}
In real platforms, vehicles and requests often arrive dynamically.
Existing works on ridesharing can be summarized into two modes: the \textit{online mode}~\cite{liu2020,chen2018,cheng2017utility,T-share,Kinetic_tree,cordeau2003tabu} and \textit{batch mode}~\cite{pnas,simple_better,zheng2018order,group_aaai}.
The state-of-the-art operator for the online mode is \textit{insertion}~\cite{cordeau2003tabu,cheng2017utility,insertion,zheng2018order,TongZZCYX18}, which inserts the requests' sources and destinations into proper positions of the vehicle's route without reordering its current schedule, minimizing the increase in the total travel cost.
Batch-based methods~\cite{pnas,simple_better,cargo,ijcai_simulator,zheng2018order} first package the incoming requests into groups, then assign a vehicle with a properly designed serving-schedule for each group of requests.

We illustrate our motivation with Example \ref{exp:problem_instance}.
The existing approaches have their shortcomings. 
Insertion operator is fast but only achieve the local optimal schedule for a request. 
The batch-based methods can achieve better results than online methods for a period of a relatively long time (e.g., one day), but their time complexities are much higher than insertion (e.g., RTV~\cite{pnas} requires costly bipartite matching). 
Can we have a more efficient batch-based approach to achieve better results for a relatively long period?
In this paper, we exploit the graph structures in the shareability graph of requests to achieve better results for ridesharing services.

\begin{example}
	\label{exp:problem_instance}
	Suppose there are two vehicles, $w_1$ and $w_2$, and four requests, $r_1$ to $r_4$, on a ridesharing platform. 
	The sources, destinations, release times, and deadlines of the requests are listed in Table~\ref{example_table}.
	As shown in Figure~\ref{subfig:roadnet_example}, the vehicles are initially located at points a and c on the road network with seven vertices, each with a capacity of three. 
	In the online-based solution~\cite{liu2020,chen2018,cheng2017utility,T-share,Kinetic_tree,cordeau2003tabu}, $r_2$ will be inserted into the existing schedule of $r_1$ of $w_1$ with a scheduled route $\langle a,c,f,d\rangle$ at time 1. 
	This results in $r_4$ not being served as it cannot be inserted into the planned route $\langle c,b,e\rangle$ of $w_2$, which serves $r_3$.
	Although it can provide better scheduling, the existing batch-based methods~\cite{pnas,cargo,ijcai_simulator,zheng2018order,luo2020} list all possible request combinations or vehicle-request pair, and run the time-consuming global matching.
	
	Instead of considering all request combinations, we can heuristically achieve good dispatch result by investigating the shareability graph shown in Figure~\ref{subfig:shareablenet_example}. 
	We construct the graph by connecting the requests that can share a vehicle, so the request with a larger degree has more sharing opportunities. 
	If we give higher priority to grouping the requests with lower degree, $r_4$ will first share $w_2$ with $r_2$, then $r_3$ will share $w_1$ with $r_1$. 
	All the requests can be served with scheduled routes $\langle a,b,e,d\rangle$ and $\langle c,f,g\rangle$ for $w_1$ and $w_2$, respectively.
\end{example}

In this paper, we propose StructRide, a well-tailored batch-based framework to exploit the structure information of requests' shareability graph in ridesharing systems. 
The achieved results are significantly improved (e.g., higher service rates and lower travel distances) under just a slight delay of response time. 
\textit{To the best of our knowledge, we are the first to integrate
	traditional allocation algorithms with graph analyses of the shareability graph in ridesharing problems, which is also the key insight of our StructRide framework.}
In StructRide, we first propose a fast shareability graph builder that efficiently extracts shareable relationships between requests through spatial indexes and an angle-based pruning strategy to only maintain the feasible shareable relationships.
To exploit the structure information of the shareability graph, we devise the \textit{structure-aware ridesharing dispatch} (SARD) algorithm to take the cohesiveness of requests in the shareability graph (evaluated through the graph structures) into consideration and revise request priorities through analyzing its shareability. 
SARD's two-phase ``proposal-acceptance'' strategy avoids invalid group enumerations for any specific vehicle.
For vehicle scheduling, we adjust the insertion order of riders' requests based on shareability, maintaining an optimal schedule with linear insertion complexity. 
Thus, the batch-based method can be more efficient in large-scale ridesharing problems. 
To summarize, we make the following contributions:

\begin{itemize}[leftmargin=*]
	\item We introduce our StructRide framework, which handles incoming requests in batch mode and adjusts the matching priority of requests in each batch in Section~\ref{subsec:systemFramework}.
	\item We utilize an existing graph structure, namely \textit{shareability graph}, for intuitively analyzing the sharing relationships among requests and design an efficient algorithm to generate the shareability graph in each batch in Section~\ref{sec:shareable_network}.
	\item We devise a heuristic algorithm, namely \textit{SARD}, for priority scheduling in each batch under the guidance of the shareability graph with theoretical analyses in Sections~\ref{sec:sard_solution}.
	\item We conduct extensive experiments on two real datasets to show the efficiency and effectiveness of StructRide framework in Section~\ref{sec:experimental}.
\end{itemize}

\section{Problem Definition}
\label{sec:problemDefinition}

We use a graph $\langle V, E \rangle$  to indicate the road network, where $V$ is a set of vertices and $E$ is a set of edges between vertices. 
Each edge $(u,v)$ is associated with a weight $cost(u,v)$ to represent the travel cost from $u$ to $v$. 
In this paper, travel cost means the minimum travel time cost from $u$ to $v$.

\subsection{Definitions}
\label{subsec:definition}

\begin{definition}[Request]
	Let $r_i=\langle s_i,e_i,n_i,t_i,d_i \rangle$ {denote} a ridesharing request with $n_i$ riders from source $s_i$ to destination $e_i$, which is released at time $t_i$ and requires reaching the destination before the delivery deadline $d_i$.
\end{definition}
To guarantee the quality of ridesharing services, the existing studies~\cite{lru_cache,cheng2017utility,zheng2018order} commonly use the detour ratio, $\epsilon$, to avoid reaching the destination of each ride too late. 
In particular, riders join ridesharing services to benefit from a discount on the travel fee, but they need to tolerate some detours. 
In this paper, the delivery deadline $d_i$ of request  $r_i$ is calculated by adding a detour tolerance to the shortest travel time (i.e. $d_i=t_i + \gamma*cost(s_i,e_i)$, $\gamma>1$).

\begin{definition}[Schedule]
	\label{def:vehicle_route}
	Given a vehicle $w_j$ with $m$ allocated requests $R_j=\left\{ r_1,\dots,r_m \right\}$.
	Let $S_j=\langle o_1,\dots,o_{2m} \rangle$ define a schedule for $w_j$, where the location $o_x\in S$ is the source location or destination of a request $r_i\in R_j$.
\end{definition}

We call the location $o_i$ in the schedule a \textit{way-point}.
A route $S_j$ is \textit{feasible} if and only if it meets these four constraints:
\begin{itemize}[leftmargin=*]
	\item \textbf{Coverage Constraint}. For any request $r_i\in R_j$, the source $s_i$ and destination $e_i$ should be included in $S_j$.
	\item \textbf{Order Constraint}.
	For any $r_i\in R_j$, the source location $s_i$ appears before the destination $e_i$ in the route $S_j$;
	\item \textbf{Capacity Constraint}.
	The total number of assigned riders must not exceed the maximum capacity $c_j$ of vehicle $v_j$.
	\item \textbf{Deadline Constraint}.
	For any $o_x \in S_j$, the total driving time before $o_x$ must {satisfy the} inequality~\ref{eq:validCheck}.
\end{itemize}

\vspace{-2ex}
{\scriptsize\begin{equation}
		\label{eq:validCheck}
		\sum_{k=1}^x cost(o_{k-1},o_k)\le ddl(o_x).
\end{equation}}

\noindent$ddl(o_k)$ is the deadline $d_i$ when $o_k$ is the destination of $r_i$, while $ddl(o_k)=d_i - cost(s_i,e_i)$ when $o_k$ is the source of $r_i$.
For simplicity, we denote $cost(s_i,e_i)$ as $cost(r_i)$.

\begin{definition}[Buffer Time~\cite{cheng2017utility,TongZZCYX18}]
	\label{def:bufferTime}
	Given a valid vehicle schedule $S_i=\langle o_1,o_2,$ $\dots,o_{2m} \rangle$, let $buf(o_x)$ be the maximum detour time at the way-point $o_x$ without violating the deadline constraints of its consequent way-points.
	{\scriptsize\begin{equation}
			\label{eq:bufferTime}
			buf(o_x)= \min\left\{buf(o_{x+1}),ddl(o_{x+1})-arrive(o_{x+1})\right\},
	\end{equation}}\vspace{-2ex}
	
	\noindent where $arrive(o_{x+1})$ is the earliest arriving time of way-point $o_{x+1}$ without any detour at previous way-points of $S_i$.
\end{definition}

\begin{definition}[Batched Dynamic Ridesharing Problem (BDRP) \cite{TongZZCYX18}]
	Given a set $R$ of $n$ dynamically arriving requests and a set $W$ of $m$ vehicles, the dynamic ridesharing problem is to plan a feasible schedule for each vehicle $w_i \in W$ to minimize a specific utility function in batch mode.
	
	In BDRP, incoming requests in time period $T$ are handled as a batch $\mathcal{P}$. Requests are partitioned into groups {\small$\mathbb{G}=\{ G_1,G_2,\dots,G_z \}$}, satisfying following conditions: {\small$\forall G_x, G_y \in \mathbb{G},  G_x\cap G_y=\emptyset$} and {\small$\mathcal{P}=\bigcup_{a=1..z}{G_a}$}.
	We then assign request groups to valid vehicles to minimize the utility function.
	In this work, we refer to the unified cost function $UC$ in~\cite{TongZZCYX18} and define the utility function $U$:
	{\small\begin{equation}
			\label{eq:global_utility}
			U(W,\mathcal{P})=\alpha\sum_{w_i\in W}\mu(w_i,G_{w_i})+\sum_{G_i\in G^-}{p_i}
	\end{equation}}
	
	{\small\begin{equation}
			\label{eq:single_utility}
			\mu(w_i,G_{w_i})= \sum_{x=1}^{|S_{w_i}| -1}cost(o_x,o_{x+1}),
	\end{equation}}
	
	\noindent where $S_{w_i}$ is the planned schedule for vehicle $w_i$ and $\mu(w_i,G_{w_i})$ depicts the total travel cost of all schedules.
	$G^-$ is composed of unassigned request groups in each batch, with a penalty $p_i$.
	\label{def:problem_definition}
\end{definition}

By setting special $\alpha$ and $p_i$, the utility function $U$ can support common optimization objectives in ridesharing problems, such as minimizing total travel cost, maximizing service rate, and maximizing total revenue \cite{TongZZCYX18}.
In this paper, we fix $\alpha$ to 1 and define $p_i=p_r\sum_{r\in G_i}cost(r.s,r.e)$ to indicate the profit loss caused by unassigned requests with a parameter $p_r$. 

\noindent\textbf{Discussion of Hardness.}
The existing works~\cite{TongZZCYX18,cheng2017utility,demand_aware,Kinetic_tree} have proved that dynamic ridesharing problem is NP-hard, thus intractable. Moreover, It is proved that there is no polynomial-time algorithm with a constant competitive ratio for dynamic ridesharing problem~\cite{TongZZCYX18}. 
Thus, we use experimental results to show the effectiveness of our approaches.

\revision{We summarize the key notations used in this paper in Table~\ref{tab:notations}.}
\begin{table}[t!]
	\begin{center}
		{\scriptsize 
			\caption{ Symbols and Notations.} 
			\label{tab:notations}
			\begin{tabular}{|c|c|}
				\hline{\bf Symbol} & {\bf Description} \\ \hline 
				$R=\{r_i\}$  & Set of requests.\\
				$W=\{w_j\}$  & Set of vehicles. \\
				$s_i,e_i$ & Source and destination of request $r_i$.\\
				$t_i,d_i$ & Release time and deadline of request $r_i$.\\
				$n_i$ & Number of riders in request $r_i$.\\
				$c_j$ & Capacity of vehicle $w_j$.\\
				$S_j$& Schedule of vehicle $w_j$. \\
				$buf(o_k)$ & Maximum allowed detour time of way-point $o_k$.\\
				$ddl(o_k)$ & Deadline of way-point $o_k$.\\
				$cost(o_k,o_l)$ & Travel cost between way-points $o_k$ and $o_l$.\\
				$p_r$ & Penalty coefficient for unassigned requests.\\
				$\gamma$ & Deadline parameter.\\
				$\mathcal{P}$ & Batch of requests.\\
				$\Delta$ & Batch time.\\
				$\theta$ & Angle between two requests.\\	
				$\delta$ & Angle threshold for pruning.\\
				$SG$ & Sharability graph.\\
				$SLoss(G_i)$ & Shareability loss of group $G_i$.\\
				\hline
			\end{tabular}
		}
	\end{center}
\end{table}

\subsection{An Overview of StructRide Framework}
\label{subsec:systemFramework}
To efficiently and effectively solve BDRP, we propose a batch-based framework, namely StructRide, that exploits the structure information of the shareability graph of requests to improve the performance of service rates and unified costs.
We first introduce the main parts of our StructRide framework, as shown in Figure~\ref{fig:system_framework}.

\noindent\textbf{Index Structure.}
In ridesharing, since vehicles keep moving over time, the used index must update efficiently.
The grid index allows querying and updating moved vehicles in constant time. 
In StructRide, we partition the road network into $n\times n$ square cells.
With the grid index, we can retrieve all available vehicles efficiently through a range query for a given location $p$ and radius $r$ in constant time.

\begin{figure}[t!]
	\centering
	\scalebox{0.4}[0.4]{\includegraphics{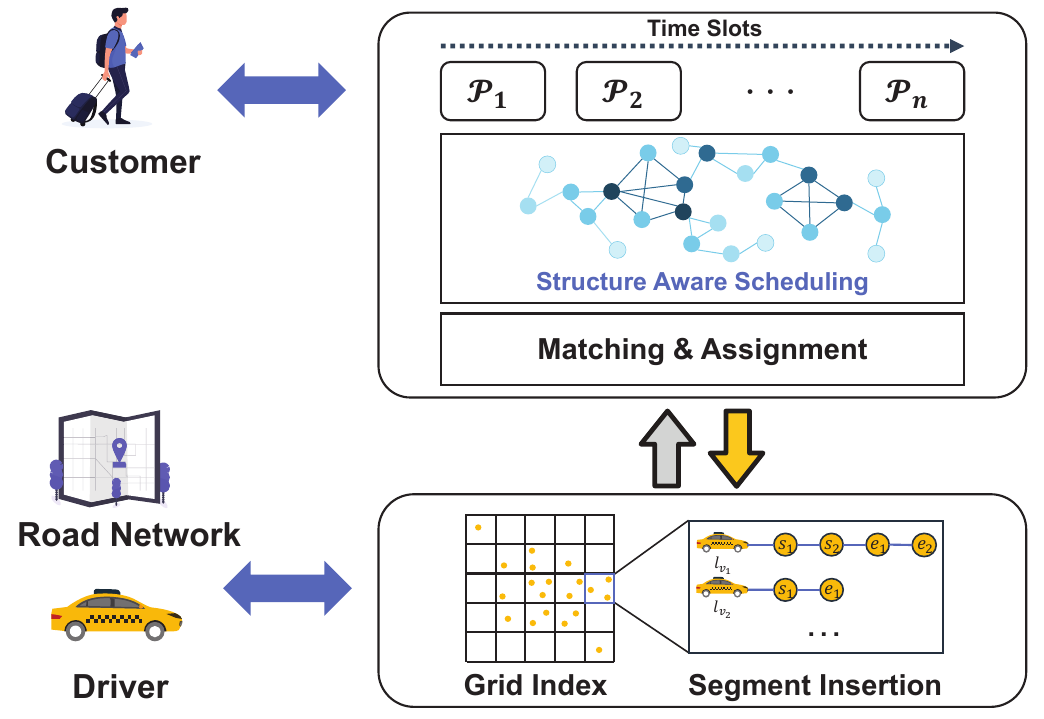}}
	\caption{ Overview of System Framework \textit{StructRide}.}
	\label{fig:system_framework}
\end{figure}

\noindent\textbf{Schedule Maintenance.}
StructRide arranges new request pickup and drop-off locations into available vehicle schedules.
Two wildly solutions are used: \textit{kinetic tree insertion}~\cite{Kinetic_tree} and \textit{linear insertion}~\cite{TongZZCYX18,insertion}.
Kinetic tree insertion maintains all feasible schedules for a vehicle in a tree structure, ensuring optimality.
Linear insertion maintains only the current optimal schedule, inserting new points without \textit{reordering}. 
While this may miss the optimal schedule, Ma \textit{et al.}~\cite{city_scale} found minimal impact on travel cost reduction, with increased time and space usage for reordering.
Thus, we use linear insertion for scheduling.

\noindent\textbf{Structure-Aware Assignment.}
We process the incoming requests into batches by their release timestamps.
For each batch, we utilize the shareability graph in Section~\ref{sec:shareable_network} to analyze the sharing relationship between requests.
In the \textit{Matching and Assignment} phase, we propose a two-phase algorithm SARD in Section~\ref{sec:sard_solution}. 
Requests will propose to candidate vehicles in descending order of additional travel cost (i.e., propose to vehicles needing more additional travel costs first). 
Vehicles then select structure-friendly request groups to accept, using our novel measurement, namely \textit{shareability loss}, to evaluate how merging a group of requests into a supernode affect the sharing probability of the rest nodes in the shareability graph.

\section{Shareability Graph Construction}
\label{sec:shareable_network}
Although the requests in the same batch have similar release times, the sharing probabilities can vary dramatically.
A request with higher shareability has more opportunities to share vehicles with other requests.
Thus, an intuitive motivation for optimizing the grouping process is to give higher priority to a request with lower shareability. 
Existing works \cite{santi2014quantifying, ma2017demand, chen2022p, lin2016model, pnas, kucharski2020exact} utilize an effective structure, the \textit{shareability graph}, to model and manage the sharing relationships between requests. 
In this section, we propose an efficient construction algorithm for it.

\subsection{Observations on Shareability Graph}
\label{subsec:shareable_network}

\begin{definition}[Shareability Graph \cite{santi2014quantifying, ma2017demand, chen2022p, lin2016model, pnas, kucharski2020exact}]
	\label{def:sn}
	Given a set of $m$ requests $R$, let $SG=\langle R,E \rangle$ represent a \textit{shareability graph} with nodes $R$ and edges $E$, where each node in $R$ corresponds to a request.
	An edge $(r_a,r_b)\in E$ signifies that $r_a$ and $r_b$ are shareable, meaning that there is at least one \textit{feasible} schedule for them to be served by a vehicle.
\end{definition}

For instance, in the \textit{shareability graph} depicted in Figure~\ref{subfig:shareablenet_example}, the edge $e=(r_1,r_2)$ between $r_1$ and $r_2$ indicates that there is a feasible route to serve them in a trip (e.g., a schedule of $\langle s_{1}, s_{2}, e_{2}, e_{1}\rangle$).
Therefore, we can easily identify the candidate shareable requests for each request by checking the neighbors in the shareability graph.
Additionally, we have the following two observations:

\textbf{Observation 1.} \textit{The degree of each request in the shareability graph reveals its shareability and importance in the corresponding batch.} 
We can intuitively evaluate the sharing opportunities of requests by comparing their degrees in the shareability graph.
We call the degree of a request as its \textit{shareability}.
A request with a smaller shareability is often more urgent to group with an appropriate shareable request.

\textbf{Observation 2.} \textit{The subgraph induced by the corresponding nodes of potential shareable requests is full connected.}
In the shareability graph, if there exists a valid sharing schedule consisting of $k$ requests, there must exist a fully connected subgraph, also called \textit{k-clique} \cite{maximum_clique,clique-community}, among the corresponding nodes. The reason is that, for a valid sharing schedule, any its  two requests must can share with each other, thus a connecting edge between the corresponding nodes will exist in the shareability graph.
With this observation, we can prune some infeasible combinations in the grouping phase of each batch.

The shareability graph's structure is intuitive, and some similar structures have been implemented in other works. 
These designs focus on different approaches to the vehicle-request matching problem.
Wang et al.~\cite{wang_infocom} developed a tree cover problem to minimize vehicle use for urban demands.
Alonso-Mora et al.~\cite{pnas} utilize linear programming to optimally assign vehicles to shareable customer groups.
Zhang et al.~\cite{zhang_tits} approach passenger matching as a monopartite matching problem, solved by the Irving-Tan algorithm.
\textit{However, existing studies ignored the heuristic insights brought by the local structure of the shareability graph to guide the vehicle-request matching.}
Furthermore, they construct the shareability graphs through the brute force enumeration of requests pairs without careful tailoring.
Therefore, we propose an efficient shareability graph construction method.

\subsection{Angle Pruning Strategy}
\label{subsec:sn_gen}
To build the shareability graph for each batch $\mathcal{P}_k$, the basic idea is to enumerate all pairs of requests $(r_a,r_b)$ in $\mathcal{P}_k$ and check for a valid sharing schedule using a linear insertion method.
The efficiency of this construction algorithm is primarily dominated by the number of shortest path queries needed. 
To avoid enumerating all pairs of requests, we suggest an angle pruning strategy, illustrated in Figure \ref{fig:pruning_example}, for more efficient shareability graph construction. 
Note that, when we prune the candidate requests for a given request $r_a$, to avoid duplicated consideration of schedules, we only consider the schedules with $s_a$ (the source of $r_a$) as the first way-point.

\begin{figure}[t!]
	\centering
	\scalebox{0.55}[0.55]{\includegraphics{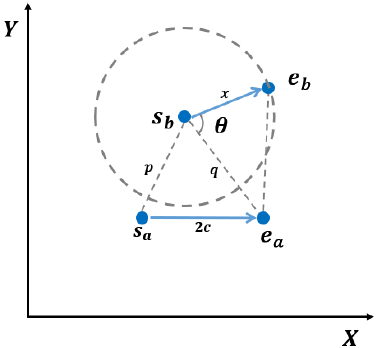}}
	\label{subfig:angel_example}    
	\caption{An Illustration Example for Angle Pruning}
	\label{fig:pruning_example}
\end{figure}

Our angle pruning strategy is based on the travel directions of the requests. 
We notice that the requests with similar travel directions are more likely to share trips. 
For instance, a northward request is unlikely to share a trip with a southward request because the driver would need to turn around and spend considerable time driving in the opposite direction. 
We define a threshold $\delta$ to prune request pairs that have similar sources but divergent directions with Theorem~\ref{theo:directional_filter}. 
Let $\overrightarrow{a b}$ be the vector pointing from $a$ to $b$.

\begin{theorem}
	\label{theo:directional_filter}
	For a request $r_a$ and its candidate sharing request $r_b$, the expected probability of sharing a trip will increase when the angle $\theta$ between $\overrightarrow{s_b e_a}$ and $\overrightarrow{s_b e_b}$ decreases.
\end{theorem}
\begin{proof}
	Firstly, we denote the maximum detour time for request $r_a$ as $(\gamma -1)\times cost(s_a, e_a)$.
	For the two possible schedules of $r_a$ and its candidate request $r_b$ ($\langle s_a, s_b, e_a, e_b\rangle$ or $\langle s_a, s_b, e_b, e_a\rangle$), one of the following two conditions needs to be satisfied due to the deadline constraint:
	(a) {\small$p+x+cost(e_a, e_b)\leq 2\gamma c$};
	(b) {\small$p+q+cost(e_a, e_b)\leq \gamma x$}, where $p, q, x$ and $2c$ indicate $cost(s_a, s_b)$, $cost(s_b, e_a)$, $cost(s_b, e_b)$ and $cost(s_a, e_a)$,  respectively. 
	
	When request $r_a$, source $s_b$ and $cost(s_b, e_b)$ are fixed, the total travel cost of the two possible schedules only depends on $cost(e_a, e_b)$.
	Then, $cost(e_a, e_b)$ can be represented through $x, n$ and $\cos(\theta)$.
	In (a), we have {\scriptsize$x\leq\frac{(2\gamma c-p)^2-q^2}{4\gamma c-2p-2q\cos(\theta)}$
		$\leq 1/(\frac{cos^2(\theta/2)}{\gamma c}+\frac{sin^2(\theta/2)}{(\gamma-1)c})$} (noted as $g(c)$), because {\small$p+q\leq 2\gamma c$} and {\small$p-q\leq 2(\gamma -1)c$} (due to the detour tolerance).
	In (b), we have {\scriptsize$\gamma x \geq p+q+\sqrt{(x-q*cos(\theta))^2}$}, which can be rewritten as {\scriptsize$x\geq \frac{2c(1-cos(\theta))}{\gamma-1}$} (marked as $h(c)$) by $p+q\geq 2c$.
	In both cases, the feasible range of $x$ gradually decreases as the angle reduces. In other words, for a given candidate request $r_b$, when the angle $\theta$ between $\overrightarrow{s_b e_a}$ and $\overrightarrow{s_b e_b}$ decreases, the travel cost $x=cost(s_b, e_b)$ is more likely to satisfy the deadline constraint.
\end{proof}

By investigating the real datasets of Chengdu and New York City in Section~\ref{sec:experimental}, we find that the distances of the requests almost follow the log-normal distribution whose probability density function is {\scriptsize$f(x;\mu;\sigma)=\frac{1}{\sqrt{2\pi}x\sigma}e^{-\frac{(\ln x-\mu)^2}{2\sigma^2}}$}, where $\mu$ and $\sigma$ are two parameters of log-normal distribution.
Therefore, the expected probability of $r_a$ sharing with a candidate request $r_b$ at an angle $\theta$ greater than a given threshold $\delta$ can be evaluated as follows:

{\small
	\begin{equation}
		\mathbb{E(\theta\geq\delta)}=\int_0^{+\infty}f(x)\left(\int_0^{g(\frac{x}{2})}f(y)dy+\int_{h(\frac{x}{2})}^{+\infty}f(y)dy\right)dx\notag
	\end{equation}
}

For instance, we fit the probability density function of the requests' distances of Chengdu and NYC datasets with a log-normal distribution. 
The probability expectations $\mathbb{E}(\theta \geq \frac{\pi}{2})$ are $40.98\%$ and $41.38\%$ respectively when $\gamma=1.5$.

\noindent\textbf{Discussion.} Our angle pruning strategy is an approximate pruning method, meaning it might occasionally remove some feasible shareable pairs. 
However, our experiments shows that this false-negative pruning will not harm the final results. 
In fact, it can sometimes even improve the final results, such as by achieving higher serving rates.
We think this is because the angle pruning strategy can help avoid considering some feasible but not very appropriate shareable pairs (i.e., the directions of the two requests diverge too much) in the following grouping and dispatching stage. 
This can somehow give the algorithm for the following stage a higher probability of achieving a better final result.

\setlength{\textfloatsep}{3pt}
\begin{algorithm}[t]
	{\small\DontPrintSemicolon
		\KwIn{The previous sharability graph $SG'$ and its request set $R_p$, a new request batch $\mathcal{P}$, an angle threshold $\delta$}
		\KwOut{The corresponding {shareability graph} $SG$}
		$SG=SG'$ \\
		\For{$r_a \in \mathcal{P}$} {
			$V=V\cup r_a, R_p=R_p\cup r_a$\\
			$C \leftarrow$ candidate requests filtered by spatial indexes, deadline and detour tolerance constraint from $R_p$\\
			\For{$r_b\in C$}{
				\If{$\arccos\left(\frac{\overrightarrow{s_b e_a}\cdot\overrightarrow{s_b e_b}}{|\overrightarrow{s_b e_a}||\overrightarrow{s_b e_b}|}\right)\in\left[-\frac{\delta}{2},\frac{\delta}{2}\right]$}{
					\If{ $r_a$ and $r_b$ are shareable}{
						$E=E\cup(r_a, r_b)$
					}
				}
			}
		}
		\Return{$SG=\langle V,E \rangle,R_p$}\;
		\caption{Dynamic \textit{Shareability Graph} Builder}
		\label{alg:sn_gen}}
\end{algorithm}

\subsection{Dynamic Shareability Graph Builder}
\label{subsec:sg_builder}
The details of the dynamic shareability graph builder algorithm are illustrated in Algorithm~\ref{alg:sn_gen}.

Firstly, we use the previous batch's sharability graph to initialize the current one (line 1).
We try to find new shareable relationships brought by each newly arrived request successively (lines 2-8).
For each new request $r_a$, we first add it to the graph nodes and the current request set (line 3). 
By leveraging spatial indexes (e.g., grid index), deadline constraints, and detour tolerance, we can quickly obtain a candidate request set $C$ with a similar source location to $r_a$ without a shortest path query.
Then, we further filter the requests $r_b\in C$ by the angle pruning rule (lines 5-8).
We construct the vector $\overrightarrow{s, e}$ by the source and destination of the request to represent the distance and direction. 
If the angle $\theta$ of $\overrightarrow{s_b e_a}$ and $\overrightarrow{s_b e_b}$ exceeds the given angle threshold $\delta$, $r_b$ will be pruned (line 6).
We'll add an edge $(r_a, r_b)$ if $r_a$ and $r_b$ are shareable  (lines 7-8).
Finally, we get a {shareability graph} $SG$ with $V$ and $E$ generated in the above steps (line 8).

\textbf{Complexity Analysis.} 
Assume the shortest path query takes $O(q)$ time.
For a request $r_a$, filtering can be done in constant time by searching \textit{Minimum Bounding Rectangle} of the detour tolerated area in grid index. 
In the worst case, the candidate request set size $|C|=|R_p|-1$ in line 3.
Besides, the angle calculation takes constant time, requiring only two insertions to test if $r_a$ and $r_b$ are shareable in $O(k)$.
Thus, the algorithm's final complexity is $O(q\cdot|R_p|\cdot|\mathcal{P}|)$.

\noindent\textbf{Discussion.}
The proposed angle pruning strategy is derived based on Euclidean space. 
In a realistic road network, facilities such as expressways make some excluded solutions actually feasible.
However, our experiments show very few cases are discarded, making the angle pruning strategy acceptable for heuristic pruning of the shareability network.

\section{Structure-Aware  Dispatching}
\label{sec:sard_solution}
With the shareability graph, we can intuitively analyze the shareability of each request and propose the structure-aware ridesharing dispatching (SARD) algorithm.
Specifically, we first discuss two main schedule maintenance methods. 
Then we introduce a bottom-up enumeration strategy for the different combinations of requests.

\subsection{Schedule Maintenance}
\label{subsec:route_maintain}
There are two state-of-art strategies for schedule maintenance in previous work: kinetic tree insertion~\cite{Kinetic_tree} and linear insertion~\cite{TongZZCYX18,insertion}.
The kinetic tree maintains all feasible schedules and checks all available way-points ordering exhaustively to insert a new request.
While the kinetic tree can achieve the exact optimal schedule for the vehicle, it needs to maintain up to $\frac{(2m)!}{2^m}$ schedules in the worst case ($m$ is the number of requests assigned to the vehicle).
In contrast, the schedule obtained by the linear insertion method is optimal only for the current schedule (i.e., local optimal). 
Linear insertion is optimal when the number of requests is 2. 
In the experimental study in Section \ref{sec:experimental}, we find that in NYC and Chengdu datasets, if we use linear insertion to handle requests according to their release time, the probability of achieving a global optimal schedule is up to $89\%$ and $85\%$, when inserting the third and fourth request, respectively.
To improve such a probability with the linear insertion method, we reorder the insertion sequence of the requests based on Observation 1 in Section~\ref{subsec:shareable_network}.
Specifically, we first select two requests with the lowest shareability (i.e., the degree of the node) in the shareability graph and generate an optimal sub-schedule.
Then, we insert the remaining requests into the sub-schedule one by one in ascending order of their shareability.
In this way, we improve the probability of achieving an optimal schedule using the linear insertion method to $91\%$ and $90\%$.

\subsection{Request Grouping}
\label{subsec:request_grouping}
In batch mode, we enumerate all feasible request combinations before the assignment phase.
Listing all combinations requires checking $\sum_{i=1}^{c} {n \choose i}$ groups, where $c$ is vehicle capacity.
After that, we verify the existence of a feasible route to serve these requests by linear insertion~\cite{TongZZCYX18} in $O(k^2)$ per group, where $k$ is the group size.
However, many groups are invalid for vehicles in practice.
To avoid unnecessary enumeration, Zeng \textit{et al.} proposed an index called \textit{additive tree}~\cite{simple_better}, which enumerates valid groups level by level through a tree-based structure. 
In the additive tree, each node represents a valid group of requests, extending its parent node's group by adding one more request.

Although the additive tree helps prune some invalid groups in the enumeration, it still maintains all feasible schedules for each tree node.
However, we will only pick the best one of these schedules for each group at the end.
With the schedule maintenance method proposed in Section~\ref{subsec:route_maintain}, we can heuristically reorder the insertion sequence to get an optimal schedule with higher probability. 
\textit{Note that, we still do not alter existing schedules, and just reorder the sequence to insert the new requests to the existing schedules one by one.}
In building the modified additive tree, we keep only one schedule for each node by inserting a new request to its parent's schedule. 

\begin{algorithm}[t]
	{\small	\DontPrintSemicolon
		\KwIn{A \textit{shareability graph} $SG$ of a batch instance and a set $R$ of $n$ requests, the capacity of vehicles $c$}
		\KwOut{Request groups set $\mathcal{RG}$.}
		initialize $\{{RG}_1, {RG}_2, \cdots, {RG}_c\}$ as empty sets\\
		\ForAll{$r_a\in R$}{
			${RG}_1$.insert($\{r_a\},\langle s_a, e_a\rangle$)
		}
		\For{$l\in[2..c]$}{
			\ForEach{pair $(G_x, G_y)$ in ${RG}_{l-1}$}{
				$G\leftarrow G_x\cup G_y$ \\ 
				\If{$|G|= l$ and $G$ satisfied Lemma~\ref{lemma:prune_rules}}{
					$r_b\leftarrow$ find maximum degree node in $G$ \\
					$S\leftarrow$ insert $r_b$ into the schedule $S'$ of $G \setminus\{r_b\}$ maintained in level ${RG}_{l-1}$ \\
					\If{ $S'$ is valid }{
						${RG}_l$.insert($G, S$)
					}
				}
			}
		}
		\Return{$\mathcal{RG}=\{{RG}_1, {RG}_2, \cdots, {RG}_c\}$}\;
		\caption{Request Grouping Algorithm}
		\label{alg:rg}
	}
\end{algorithm}

Algorithm~\ref{alg:rg} outlines the detailed construction steps.
Initially, $c$ empty sets $\{{RG}_1, {RG}_2, \cdots, {RG}_c\}$ are initialized to store $c$ levels of nodes for the modified additive tree (line 1).
Then, we construct the groups formed by individual request $r_a\in R$ whose schedule consists of its source and destination for level 1 of the modified additive tree (lines 2-3).
For the construction of the feasible groups of the remaining levels ${RG}_l$, we generate each node in level $l$ by traversing and merging pairs of parent nodes' groups ${RG}_{l-1}$ (lines 4-11).
During constructing the nodes in level $l$, only groups $G$ with $l$ requests that satisfy Lemma~\ref{lemma:prune_rules} are considered (line 7).
For each feasible group $G$, we search for the request $r_b\in G$ with the maximum shareability in $SG$ (line 8).
The new schedule $S$ of $G$ is generated by inserting $r_b$ into its parent node's schedule $S'$ using the linear insertion method (line 9).
If the generated schedule is valid for $G$, we store the new group with its maintained schedule $(G, S)$ into ${RG}_l$ (lines 10-11).

\begin{lemma}
	\label{lemma:prune_rules}
	For any valid group $G_x$: (a) $\forall r_a\in G_x$, the group $G_x \setminus\{r\}$ must be also valid; and (b) the nodes of $r_a\in G_x$ forms a clique in the shareability graph.
\end{lemma}
\begin{proof}
	Condition (a) is proved by Lemma 2 in~\cite{simple_better} and condition (b) is derived from  Observation 2 in Section \ref{subsec:shareable_network}.
\end{proof}

\begin{example}
	\label{exp:group_tree}
	Consider the example in Figure~\ref{fig:motivation_example}.
	We first initialize the groups composed of a single request with corresponding schedule.
	Next, we merge pairs of 1-size groups to create new schedules for child nodes.
	In the shareability graph $SG$ in Figure~\ref{subfig:shareablenet_example}, since $deg(r_2)>deg(r_3)$, we insert $r_2$ into the schedule of group $\{r_3\}$, forming $\{r_3,r_2\}$.
	The linear insertion method's schedule is optimal for groups of size 2.
	At the same time, we take the generated group $\{r_3,r_2\}$  as a child of group $\{r_3\}$.
	Since we cannot find a feasible schedule for $r_4,r_1$, all groups containing $\{r_1,r_4\}$ are pruned in subsequent steps.
	We then insert $r_2$ into the schedule of group  $(r_1,r_3)$ and get an approximate schedule $\langle a,b,c,f,e,d\rangle$ because $deg(r_2)>deg(r_1)=deg(r_3)$.
	As we cannot find any valid group $G$ which $|G|\geq 4$, we end the group building process and got the final result  in Figure~\ref{fig:additive_tree}.
\end{example}

\begin{figure}[t!]
	\centering
	\scalebox{0.34}[0.34]{\includegraphics{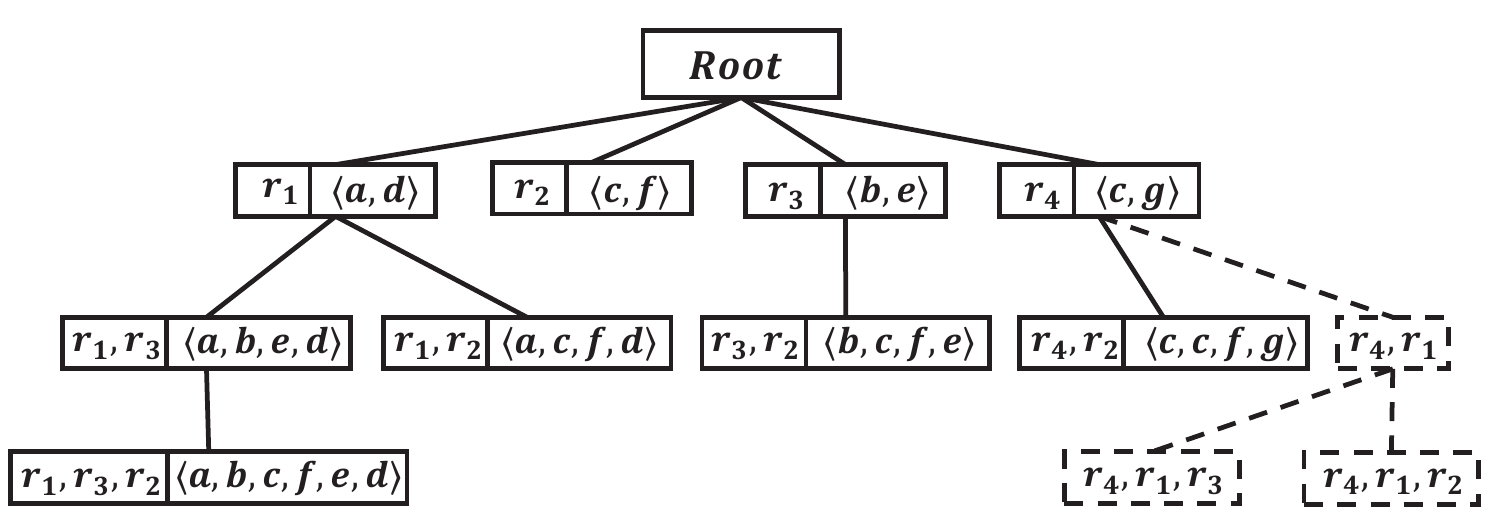}}\vspace{1ex}
	\caption{ Grouping Tree in Example~\ref{exp:group_tree}.}
	\label{fig:additive_tree}
\end{figure}

\noindent\textbf{Complexity Analysis.}
In the worst case, all the requests in a batch can be arbitrarily combined so that there are at most $\sum_{i=1}^{c}{n \choose i}$ nodes in total. 
Since the capacity constraint $c\ll n$ in practice, the number of combination groups can be noted as $O(n^c)$.
We only need $O(c)$ time for each new group to perform linear insert operation for a new schedule. 
We can finish the request grouping in $O(c\cdot n^c)$.

\subsection{SARD Algorithm}
\label{subsec:SARD}
With the grouping algorithm in Section~\ref{subsec:request_grouping}, we propose \textit{SARD}, a two-phase matching algorithm for matching between vehicles and requests in this section. 
The intuition of SARD is to greedily maintain the shareability or connectivity of nodes in the shareability graph, increasing the likelihood of requests sharing with others in the final assignment. 
In SARD, requests initially propose to \textit{worse} vehicles that result in higher travel costs, giving more initiative to vehicles on selecting groups of requests. 
Each vehicle $v_j$ then greedily accepts a group of proposed requests with the smallest \textit{shareability loss}, while rejected requests propose to better vehicles in later rounds. 
The proposal and acceptance phases are iteratively conducted until no request will propose. 
We start by defining shareability loss, then develop theoretical analysis to support our design of SARD.

To evaluate the effect of assigning a request group on the shareability of the remaining graph, we introduce a substitution operation to replace a $k$-clique $G_i$ in shareability graph $SG$ with a supernode $\hat{v}_i$. 
After we substitute a supernode $\hat{v}_i$ for a $k$-clique $G_i$, an edge connects another node $v_j \in SG \setminus G_i$ with $\hat{v}_i$ if and only if $v_j$ connects to every node of $G_i$ in the original graph. 
Then, we define shareability loss as:

\begin{definition}[Shareability Loss]
	Given a shareability graph {\scriptsize$SG=\langle V, E\rangle$}, the shareability loss {\scriptsize$SLoss(G_i)$} of substituting a super-node $\hat{v}_i$ for a $k$-clique group {\scriptsize$G_i\subseteq V$}  is evaluated with the following structure-aware loss function:
	
	{\scriptsize\begin{equation}
			\label{eq:loss_function}
			SLoss(G_i)=\max_{r\in G}\{|\bigcap_{v\in G-\{r\}}{N(v)}|+|N(r)|-|\bigcap_{v\in G}{N(v)}|-1\},
	\end{equation}}
	where {\scriptsize$N(v)$} is the set of neighbor nodes of $v$ in the shareability graph {\scriptsize$SG$}. {\scriptsize$SLoss(G)=deg(r)$} if {\scriptsize$|G|=|\{r\}|=1$}.
\end{definition}

\begin{example}
	With the shareability graph in Figure~\ref{subfig:shareablenet_example}, we illustrate the idea of shareability loss in Figure~\ref{fig:loss_example}, assuming $r_4$ is unavailable.
	In Figure~\ref{subfig:loss_13}, we aim to merge $r_1$ and $r_3$ into a supernode as follows.
	We first remove the edges incident to $r_1$ and $r_3$.
	Since $r_1$, $r_2$ and $r_3$ form a 3-clique in the original graph, there should be a shareable relation between $r_2$ and the supernode $\hat{v}_{13}=\{r_1,r_3\}$.
	Thus, we add a new edge between $r_2$ and $\hat{v}_{13}$.
	Overall, if we substitute $\hat{v}_{13}$ for $\{r_1,r_3\}$, the shareability loss is $SLoss(\{r_1,r_3\}) = 3-1=2$.
	In Figure~\ref{subfig:loss_12}, we try to merge $r_1$ and $r_2$ into a supernode $\hat{v}_{12} = \{r_1,r_2\}$.
	Similarly, four edges incident to $r_1$ and $r_2$ are removed, and a new edge between $\hat{v}_{12}$ and $r_3$ is built.
	Then, the shareability loss is $SLoss(\{r_1,r_2\})=4-1=3$ here.
	Therefore, substituting $\{r_1,r_3\}$ is more structure-friendly than substituting $\{r_1,r_2\}$.
\end{example}

\begin{figure}[t!]
	\centering
	\subfigure[$SLoss(\{1,3\})$]{\scalebox{0.34}[0.34]{\includegraphics{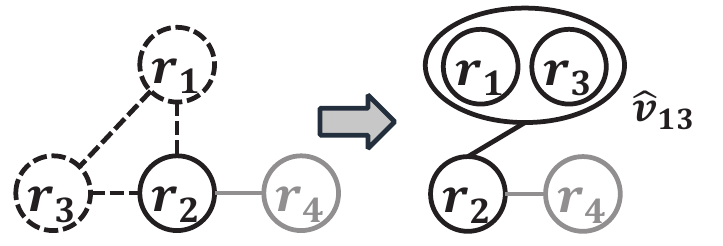}}
		\label{subfig:loss_13}}\hspace{-3ex}
	\subfigure[$SLoss(\{1,2\})$]{\scalebox{0.34}[0.34]{\includegraphics{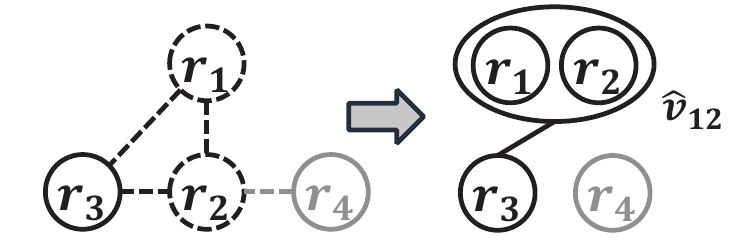}}
		\label{subfig:loss_12}}
	\caption{An Illustration Example of Shareability Loss}
	\label{fig:loss_example}
\end{figure}

Shareability loss can guide a vehicle to select a set of requests to serve from the potential requests (i.e., proposed to the vehicle in the proposal phase). 
Specifically, a vehicle $v_j$ should select a set of requests whose shareability loss is the minimum among all groups of its potential requests. 
In Theorem \ref{theo:eval_guarantee}, we prove that through serving a group of requests with the minimum shareability loss, the remaining requests can have a higher upper bound of sharing rate.

\begin{figure*}[t!]
	\centering
	\subfigure[Transformation of the Candidate Queue]{
		\scalebox{0.4}[0.4]{\includegraphics{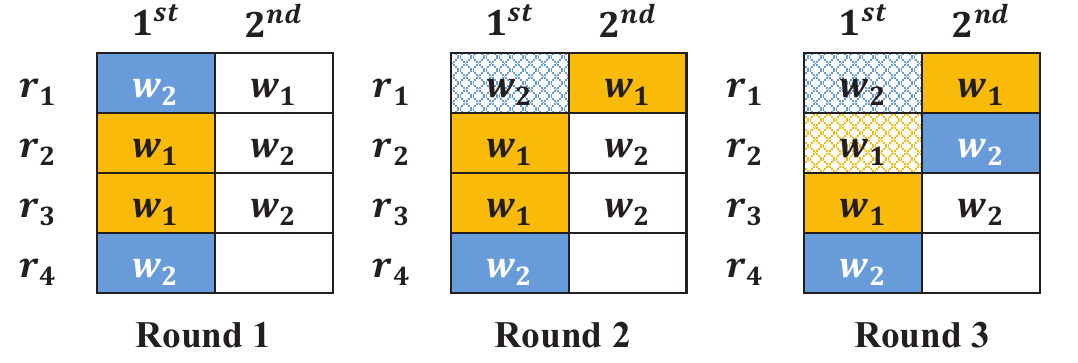}}
		\label{subfig:candidate_queue}}
	\subfigure[Grouping Trees for Vehicles]{
		\scalebox{0.3}[0.3]{\includegraphics{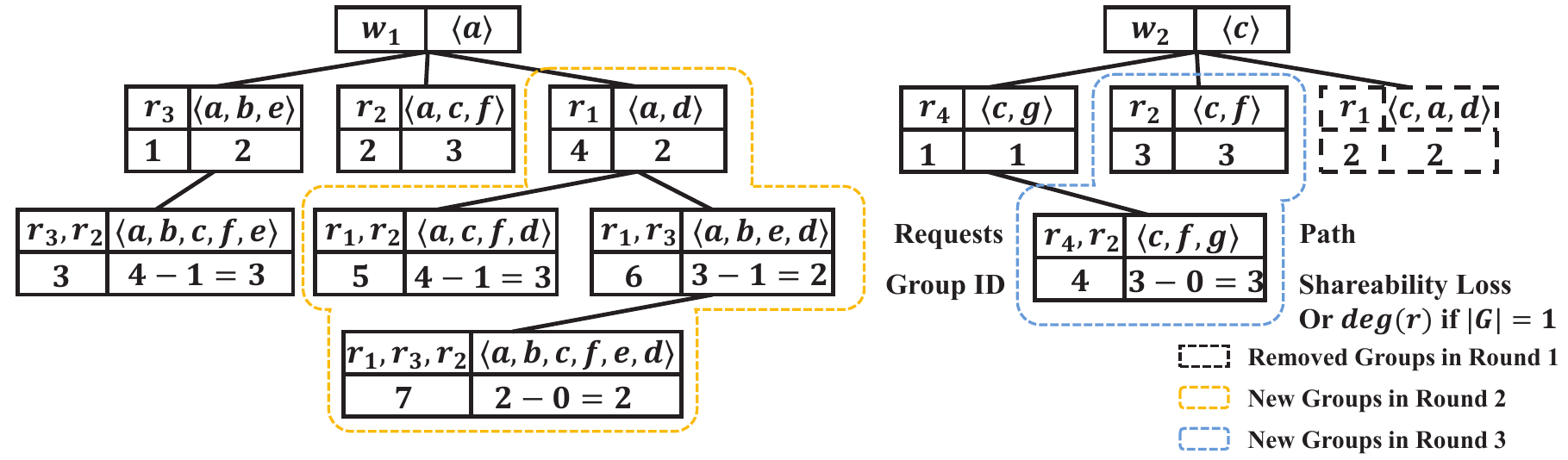}}
		\label{subfig:grouping_trees}}
	\caption{An Example for the \textit{SARD} Algorithm }
\end{figure*} 

\begin{theorem}
	\label{theo:eval_guarantee}
	Substituting a supernode for group $G$ with lower shareability loss will raise the sharing probabilities upper bound for the remaining nodes in the shareability graph.
\end{theorem}

\begin{proof}
	Observation 2 in Section \ref{subsec:shareable_network} shows that all valid groups during assignment form a $k$-clique in the shareability graph. 
	Therefore, we model maximizing sharing probability as a clique partition problem~\cite{cliqueCover}, aiming to find the minimum number of cliques to cover the graph, ensuring each node appears in precisely one clique.
	In the worst case, each node forms a $1$-clique, resulting in a sharing probability of 0 (no requests can share).
	Intuitively, the less clique we use to cover the shareability graph, the higher the sharing probability is. 
	Although Observation 2 provides a necessary but not sufficient condition for a clique to be shareable, cliques in the network still significantly boost the potential probability of requests being shareable. 
	Hence, in the following proof, we illustrate the help of the shareability loss on the sharing probability by minimizing the number of cliques.
	Moreover, since vehicle capacity limits group size to $k$, we consider the problem of partitioning the shareability graph into a minimum number of cliques no larger than $k$.
	In~\cite{cliquePartitionBound}, J. Bhasker \textit{et al.} presented an optimal upper bound (shown in equation~\ref{eq:clique_upper_bound}) for the clique partition problem evaluated with the number of nodes $n$ and edges $e$ for the graph.
	
	{\small\begin{equation}
			\label{eq:clique_upper_bound}
			\theta_{upper}=\lfloor \frac{1+\sqrt{4n^2-4n-8e+1}}{2} \rfloor
	\end{equation}}
	
	In addition, we observe that the degrees of most riders in the shareability graph are relatively small, consistent with a power-law distribution.
	For analytical tractability, we assume the node degree in the shareability graph as a random variable $\delta$ following the power-law distribution with exponent $\eta$, which complementary cumulative distribution is shown as $\Pr(\delta \geq x)=ax^{-\eta}$. 
	For a shareability graph $SG$ with $n$ nodes and the degree of its every node follows the power-low distribution with exponent $\eta$, Janson \textit{et al.}~\cite{largestCliqueBound} revealed that the size of the largest clique $\omega(SG(n,\eta))$ in the graph $SG$ is a constant with $\eta>2$.
	However, in the heavy-tail distribution, $\omega(SG(n,\eta))$ grows with $n^{1-\eta/2}$ (see equation~\ref{eq:clique_size_bound}).
	
	{\small 
		\begin{equation}
			\label{eq:clique_size_bound}
			\omega(SG(n,\eta))=\begin{cases}
				(c+o_p(1))n^{1-\eta/2}(\log n)^{-\eta/2}, &\text{if } 0 < \eta < 2 \\
				O_p(1), &\text{if } \eta = 2 \\
				2\text{ or }3\text{ w.h.p}, &\text{if } \eta > 2
			\end{cases}
	\end{equation}}

	For any shareability graph $SG$, we deal with the optimal partition $\{C_1\dots,C_\theta\}$ of general clique partition problem on $SG$ as follows: for each clique $C_i$, we divide $C_i$ into sub-cliques with size no larger than $k$. 
	After that, we obtain an upper bound $\theta'_{upper}$ after scaling of our problem with $n$ and $e$ by formula~\ref{eq:clique_upper_bound} and~\ref{eq:clique_size_bound}.
	
	{\small\begin{equation}
			\theta'_{upper}= \lfloor \frac{1+\sqrt{4n^2-4n-8e+1}}{2} \rfloor \cdot \lceil\frac{\omega(SG(n,\eta))}{k}\rceil
	\end{equation}}
	
	The higher the number of edges $e$ in the shareability graph $SG$ is, the lower the upper bound $\theta'_{upper}$ of the number of clique partitions is.
	In conclusion, substituting a supernode for the group of nodes with the lowest edge loss will keep the most edges left in the shareability graph $SG$, improving the remaining nodes' sharing probabilities in $SG$.
\end{proof}

To merge nodes whose degrees are 2 in the shareability graph at the beginning, we have the following theorem:
\begin{theorem}
	\label{theo:first_merge}
	Given a shareability graph $SG=(V, E)$, merging the node $v$, whose degree is 1, with its neighbor into a $2$-clique will not reduce the sharing rate in $SG$.
\end{theorem}
\begin{proof}
	Suppose the optimal group partitions of size no larger than $k$ on $SG$ are $\{G_1,\dots,G_n\}$, with $v_x$ in $G_a$ and $v_y$ in $G_b$.
	The clique size $|G_a|\leq 2$ because $deg(v_x)=1$.
	If $v_x$ and $v_y$ are in different partitions ($a\neq b$), then $|G_a|=1$ because $v_x$ is not connected to any node except $v_y$. 
	Thus, removing $v_y$ from $G_b$ and merging it into $G_a$ won't increase isolated groups of size no more than $1$.
	Otherwise, $v_x$ and $v_y$ are in the same partition ($a = b$) and form a $2$-clique, consistent with the theorem.
	Thus, the operation in the theorem maintains the sharing rate of nodes on $SG$. This completes the proof.
\end{proof}

\begin{algorithm}[t!]
	{\small
		\DontPrintSemicolon
		\KwIn{A set $R$ of $n$ requests with a batching time period $T$ and a set $W$ of $m$ vehicles}
		\KwOut{A vehicle set $W$ with updated schedules.}
		$R_p\leftarrow$ $\emptyset$\Comment{\textit{\scriptsize for unmatched and unexpired requests}}\; 
		\ForEach{batch $\mathcal{P}$ within time period $T$}{
			$SG,R_p\leftarrow$ build \textit{shareability graph} with Algorithm~\ref{alg:sn_gen} \;
			\ForEach{${r_a}\in R_p$}{
				$\mathcal{Q}_{r_a}\leftarrow$ initialize a priority queue by $\Delta$ utility\;
				insert candidate vehicles into $\mathcal{Q}_{r_a}$ \;
			}
			\While{$\exists r_a\in R_p, |\mathcal{Q}_{r_a}|\neq0$}{
				\ForEach(\Comment{\textit{\scriptsize Proposal Phase}}){${r_a}\in R_p$ \textbf{and} $|\mathcal{Q}_{r_a}|\neq0$}{
					$w_x\leftarrow\mathcal{Q}_{r_a}$.pop() the worst vehicle\;
					$R_{w_x}$.insert(${r_a}$)\;
				}
				\ForEach(\Comment{\textit{\scriptsize Acceptance Phase}}){${w_x}\in W$}{
					$G_{w_x}\leftarrow$ grouping $R_{w_x}$ by Algorithm~\ref{alg:rg} \;
					$G_{w_x}^*\leftarrow$ select a group with minimum shareability loss from $G_{w_x}$ \;
					push ${w_x}.ac\setminus G_{w_x}^*$ back to $R_p$ for next proposal \;
					\ForEach{$r_b\in G_{w_x}^*$}{
						remove $r_b$ from $R_{w_x}$ and add $r_b$  to ${w_x}.ac$\;
					}
				}
			}
			remove expired requests from $R_p$ and SG\;
			\Return{$W$}\;
		}
		\caption{\textit{SARD}}
		\label{alg:SARD}
	}
\end{algorithm}

With Theorems \ref{theo:eval_guarantee} and \ref{theo:first_merge}, we design our SARD, whose pseudo code is shown in Algorithm~\ref{alg:SARD}.
Firstly, we maintain a current working set $R_p$ (line 1) containing all available requests until the current timestamps.
When each batch arrives, we insert the requests in this batch into $R_p$ (line 3) and process them along with the available (unmatched and unexpired) requests in the previous round (lines 2-17).
Before starting the two-phase processing, we retrieve all candidate vehicles for each request $r_a$ and store them in a priority queue $\mathcal{Q}_{r_a}$ in descending order of the increased utility cost for serving $r_a$ (lines 4-6).
In the proposal phase (lines 8-10), the discarded requests in the previous round will propose to their current worst vehicle (i.e., the vehicle with the highest increase in travel cost for that request) (line 10).
After that, we enumerate the feasible groups of requests received by each vehicle $w_x$ by Algorithm~\ref{alg:sn_gen}. 
We evaluate the validation of group nodes based on $w_x$'s current schedule (line 12). 
Specifically, we replace the single request schedule in line 3 of Algorithm~\ref{alg:sn_gen} with the schedule generated by inserting the request into the worker's current schedule. 
This helps filter out the groups that vehicle $w_x$ cannot serve. 
With Theorem~\ref{theo:eval_guarantee}, we prioritize groups with minimal shareability loss and select such groups for vehicle $w_x$ (line 13). 
Then, we put the discarded requests back to the working pool $R_p$ (line 14), where $w_x.ac$ denotes the currently accepted requests of $w_x$.
At the end of each acceptance phase, we assign the selected group to the vehicle and put the rejected requests back into the working set $R_p$ for future proposals (lines 15-16).
We repeat the proposal and acceptance phase until no requests remain. 
Finally, we remove expired requests from $R_p$ and $SG$ that cannot be completed due to exceeding the maximum waiting time at the end of each batch.

\begin{example}
	Let's examine the batch of requests in Example~\ref{exp:problem_instance}. 
	Suppose two idle vehicles, $w_1$ and $w_2$, are located at $a$ and $c$. 
	We first create a candidate vehicle queue for $r_1\sim r_4$. 
	For example, the travel cost for $w_1$ and $w_2$ to serve $r_1$ is $cost(r_1)$ and $cost(r_1)+cost(c,a)$, respectively.
	Thus, the priority order in $r_1$'s candidate queue is $\langle w_2,w_1\rangle$.
	We can compute the candidate queue for the remaining requests, as shown in Figure~\ref{subfig:candidate_queue}.
	In the first round, $r_2$ and $r_3$ are added to $w_1$'s candidate pool, while the remaining requests go to $w_2$'s pool. 
	Then, $w_1$ and $w_2$ enumerate all the groups by Algorithm~\ref{alg:rg} with the requests in their candidate pool (see Figure~\ref{subfig:grouping_trees}, excluding the colored dashed area).
	For $w_1$, we take $\{r_3,r_2\}$ as the temporal assignment as it's the only group with at least 2 members.
	Since there are no groups with at least 2 members in the first round of $w_2$, we prefer the group $\{r_4\}$, which has a lower degree by default as it leads to less shareability loss.
	Since $r_1$ is not accepted by $w_2$, it will propose to $w_1$ according to its candidate list in the second round. 
	At this time, $w_1$'s grouping tree will update to the groups shown in the Figure~\ref{subfig:grouping_trees}.
	Despite the same loss for $G_6$ and $G_7$, $w_1$ updates its best group to $\{r_1, r_3\}$ over $\{r_1, r_3, r_2\}$ due to a higher sharing ratio in $G_6${\scriptsize$\frac{cost(P)}{\sum_{r \in SG} cost(r)}$}.
	Therefore, $r_2$ will be discarded and propose in the third round, and $w_2$ will take $\{r_4,r_2\}$ by Theorem~\ref{theo:first_merge}.
	Finally, $w_1$ and $w_2$ are assigned $\{r_1,r_3\}$ and $\{r_2,r_4\}$, respectively.
\end{example}

\textbf{Complexity Analysis.}
We need to perform an insertion operation in {\small$O(c)$} time to obtain the increased travel cost.
Thus, the time cost of constructing candidate queues for a batch of size $n$ is {\small$O(c\cdot n^2)$}. 
In the proposal phase, every request will propose to every car at most one time in the worst case. 
Thus, we need to build a grouping tree of the whole batch for each car in {\small$O(m\cdot n^c)$}. 
The time complexity of \textit{SARD} is {\small$O(c\cdot n^2+m\cdot n^c)$} in total.

\section{Experimental Study}
\label{sec:experimental}

\subsection{Experimental Settings}
\noindent\textbf{Data Sets.} The road networks of \textit{CHD} and \textit{NYC} are retrieved as directed weighted graphs. 
The nodes in the graph are intersections, and the edges' weight represents the required travel time on average. 
There are 214,440 nodes and 466,330 edges in CHD road network (112,572 nodes and 300,425 edges in NYC road network).
All algorithms utilize the hub labeling algorithm~\cite{hub_labeling} with LRU cache~\cite{lru_cache} for shortest path queries on the road network.

We conduct experiments on \revision{three} real datasets.
The first dataset consists of requests collected by Didi in Chengdu, China (noted as \textit{CHD}), published via its GAIA program~\cite{GAIA}.
The second one is the yellow and green taxi in New York, USA (noted as \textit{NYC})~\cite{nycdataset}, which has been used as benchmarks in ridesharing studies~\cite{TongZZCYX18,demand_aware}.
These datasets include each request's release time, source and destination coordinates, and the number of passengers.
\revision{The third dataset is from a delivery service in Shanghai, China (referred to as \textit{Cainiao}).
	Given space constraints, we put the content related to the third dataset in Appendix B.}
We focus on the days with the highest number of requests: November 18, 2016, for CHD (259K requests) and April 09, 2016, for NYC (250K requests). 
For requests with multiple riders, they are examined during the group enumeration in Algorithm~\ref{alg:rg}.
We set the deadline for request $r_i$ as {\small$d_i = t_i + \gamma \cdot cost(r_i)$}, a commonly used configuration in many existing studies~\cite{lru_cache,cheng2017utility,demand_aware}.
We set the maximum waiting time threshold for requests to 5 minutes\revision{ follows previous work~\cite{santi2014quantifying},} which means {\small$w_ i=\min(5min, d_i-cost(r_i)-t_i)$}.
\revision{In reality, most requests are processed within 5 minutes. 
	In the \textit{NYC} dataset, 97.69\% of requests are processed within 3 minutes, and 99.36\% of requests are processed within 5 minutes.}
Figure~\ref{fig:distribution_analysis} shows request source and destination distribution.
Sources and destinations are mapped to road network nodes in advance.
The initial positions of the vehicles are chosen randomly from the road network.

Table~\ref{tab:settings} displays the parameter settings, with default values in bold.
For simplicity, we set all vehicles to have the same capacity. 
	The case of vehicles with different capacities is discussed in Appendix C and the effect of the Angle pruning strategy is discussed in Appendix D.

\begin{table}[t!]
	\begin{center}
		{\scriptsize 
			\caption{ Experimental Settings.} 
			\label{tab:settings}
			\begin{tabular}{l|l}
				{\bf \qquad \qquad \quad Parameters} & {\bf \qquad \qquad \qquad Values} \\ \hline \hline
				the number, $n$, of requests  & 10K, 50K, 100K, 150K, 200K, \textbf{250K}\\
				the number, $m$, of vehicles  & 1K, 2K, 3K, 4K, \textbf{5K} \\
				the capacity of vehicles $c$ & 2, 3, \textbf{4}, 5, 6\\
				the deadline parameter $\gamma$ & 1.2, 1.3, \textbf{1.5}, 1.8, 2.0 \\
				the penalty coefficient $p_r$ & 2, 5, \textbf{10}, 20, 30 \\
				the batching time $\Delta$ (s) & 1, 3, \textbf{5}, 7, 9 \\
				\hline
			\end{tabular}
		}
	\end{center}
\end{table}

\noindent\textbf{Approaches and Measurements.}
We compare our {SARD} with the following algorithms:
\begin{itemize}[leftmargin=*]
	\item \textit{\textbf{pruneGDP~\cite{TongZZCYX18}.}} 
	The request is added to the vehicle's current schedule in order and the vehicle with the smallest increase in distance to serve is chosen.
	\item \textit{\textbf{RTV~\cite{pnas}.}}
	It is in batch mode and obtains an allocation scheme with the least increased distance and most serving vehicles by linear programming.
	\item \textit{\textbf{GAS~\cite{simple_better}.}} 
	It enumerates vehicle-request combinations in random batch order, selecting the best group per vehicle based on total request length as profit.
	\revision{\item \textit{\textbf{DARM+DPRS~\cite{Haliem2020ADM}.}} 
		It uses deep reinforcement learning to allocate idle vehicles to anticipated high-demand areas.
		\item \textit{\textbf{TicketAssign+~\cite{10570212}.}} 
		It proposes ticket locking on individual vehicles to enhance multiple-worker parallelism.}
\end{itemize}

We present the following metrics of all algorithms. 
\begin{itemize}[leftmargin=*]
	\item \textbf{Unified Cost.}
	The objective function of BDRP problem in Equation \ref{eq:global_utility}.
	\item \textbf{Service Rate.}
	The service rate is the ratio of assigned requests to the total number of requests.
	\item \textbf{Running Time.} We report the total running time of assigning all the tasks.
	The response time for a single request can be calculated by dividing the running time with the number of total requests (e.g., about 2 ms for each request in SARD).
\end{itemize}

\begin{figure}[t!]\centering
	\subfigcapskip=-1ex
	\subfigure[ Requests distribution (\textit{CHD})]{
		\includegraphics[scale=0.15]{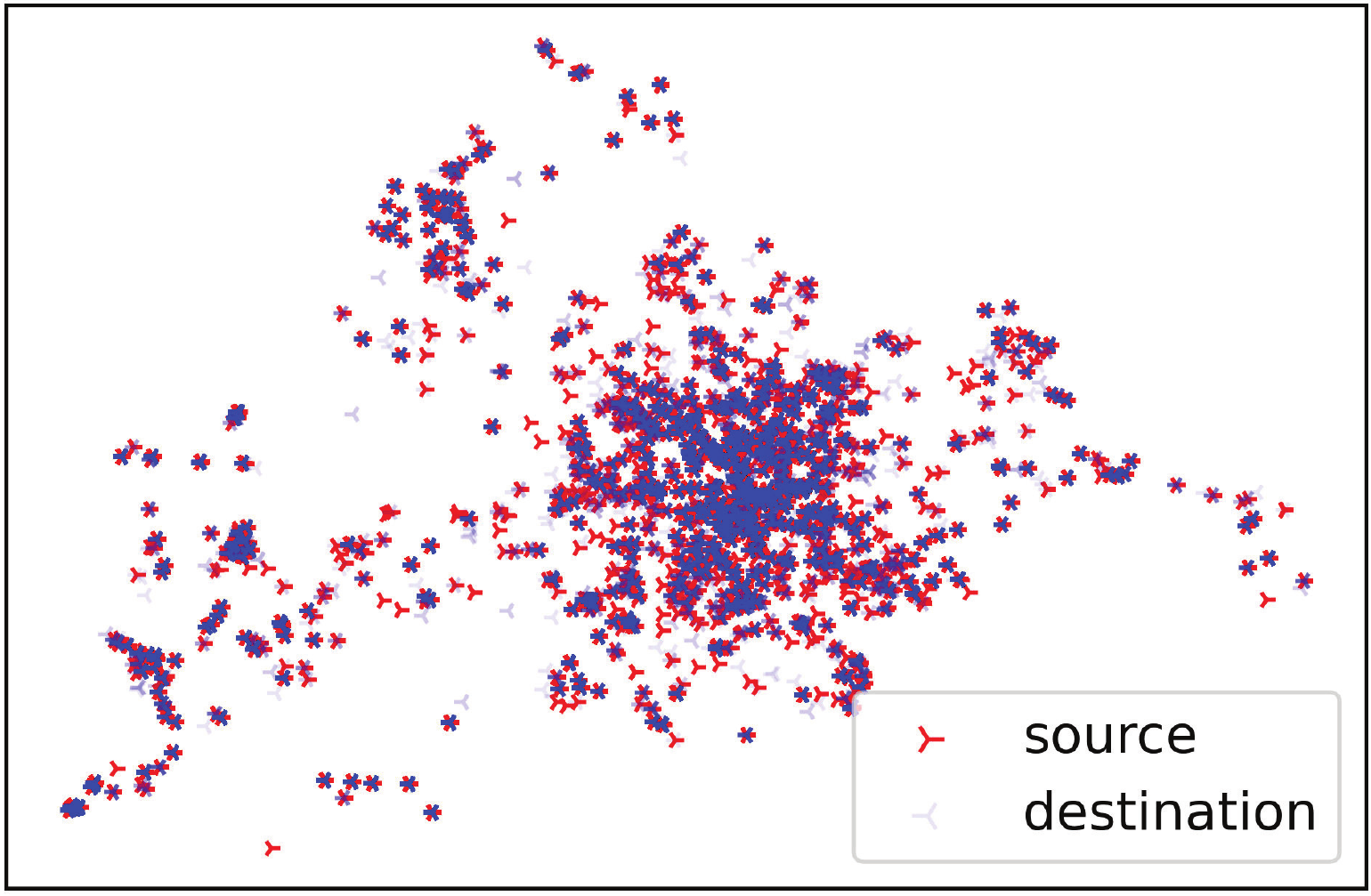}
		\label{subfig:distr_cd}}
	\subfigure[{ Requests distribution (\textit{NYC})}]{
		\includegraphics[scale=0.15]{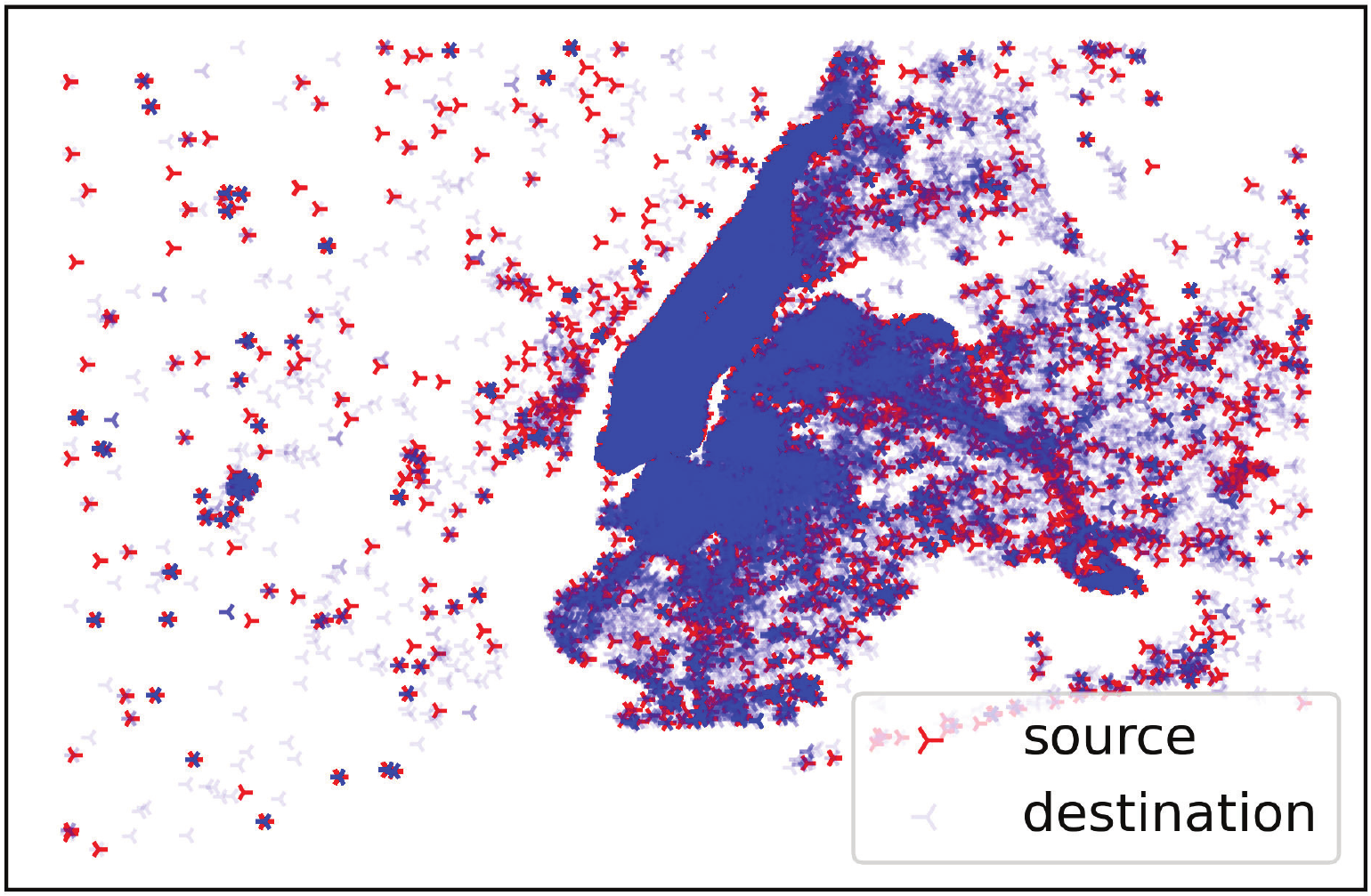}
		\label{subfig:distr_nyc}}
	\caption{The distribution of datasets.}
	\label{fig:distribution_analysis}
\end{figure}
\revision{We also provide memory consumption performance analysis of all algorithms in Appendix A~\cite{report}.}

We conduct experiments on a single server with Intel Xeon 4210R@2.40GHz CPU and 128 GB memory.
All the algorithms are implemented in C++ and optimized by -O3 parameter.
Besides, we use \textit{glpk}~\cite{GLPK} to solve the linear programming in the RTV algorithm, and all algorithms are implemented in a single thread \revision{except for {TicketAssign+}}.

\begin{figure*}[t!]\centering
	\begin{tabularx}{\textwidth}{XX}
		\begin{minipage}[c]{.5\textwidth}
			\subfigure{
				\scalebox{0.45}[0.45]{\includegraphics{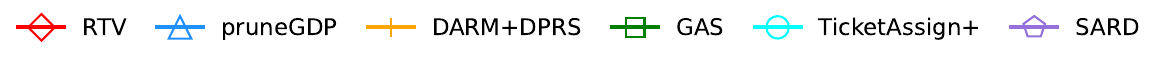}}}\hfill
			\addtocounter{subfigure}{-1}\\[-3ex]
			\subfigure[][{\scriptsize Unified Cost (\textit{CHD})}]{
				\raisebox{-1ex}{\scalebox{0.19}[0.17]{\includegraphics{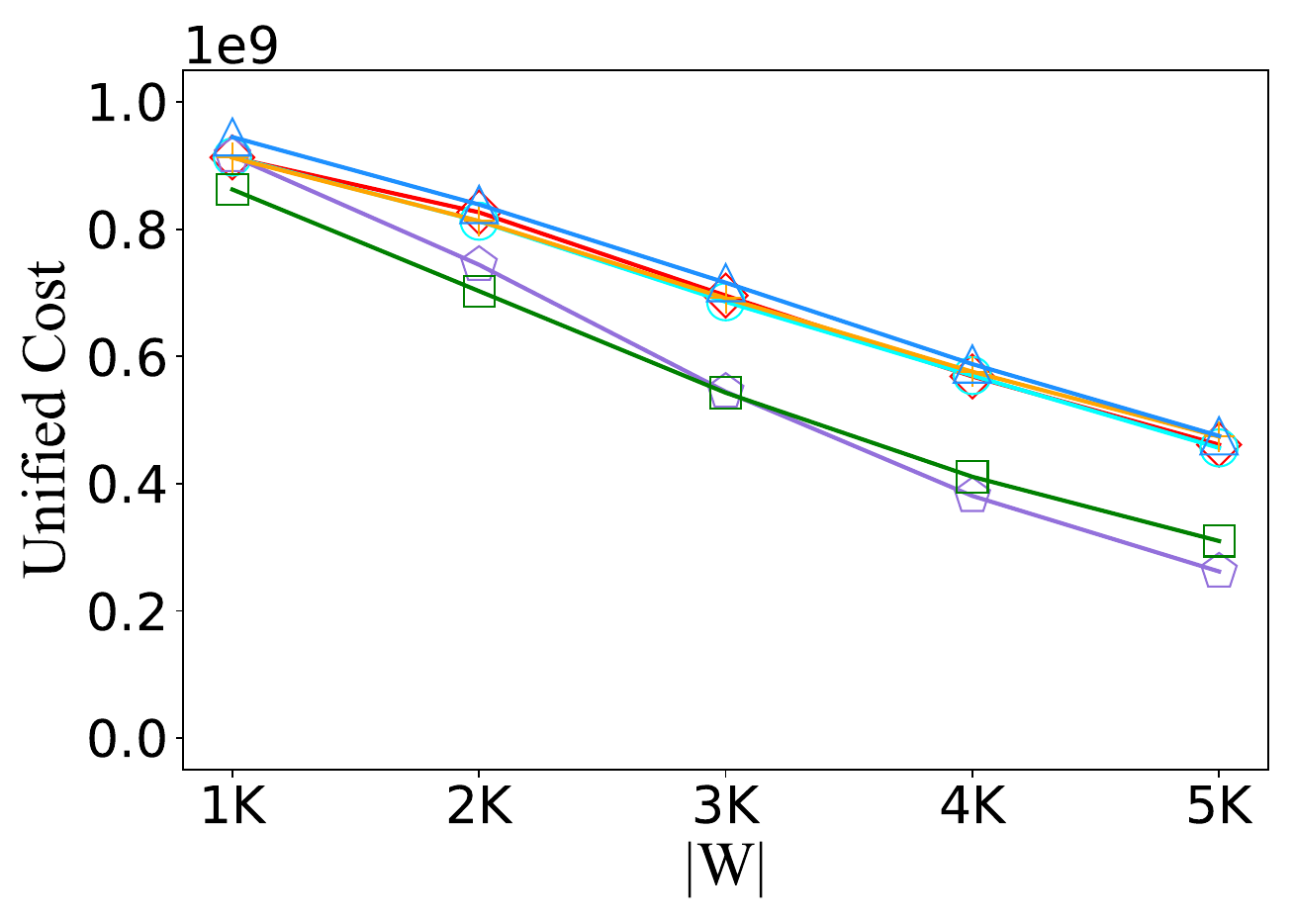}}}
				\label{subfig:uc_varing_w_cd}}\hspace{-2ex}
			\subfigure[][{\scriptsize Unified Cost (\textit{NYC})}]{
				\raisebox{-1ex}{\scalebox{0.19}[0.17]{\includegraphics{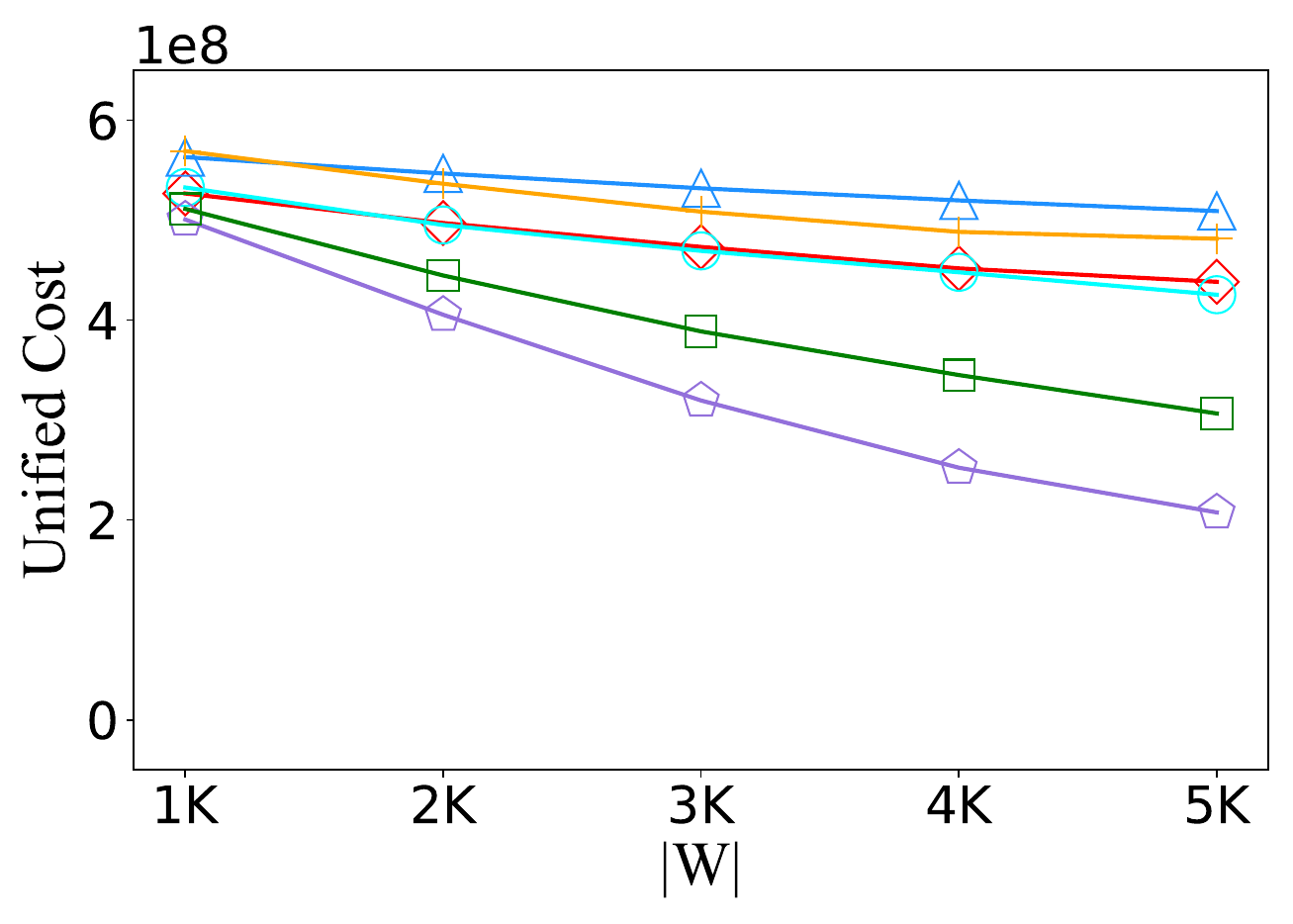}}}
				\label{subfig:uc_varing_w_nyc}}\\[-2ex]
			\subfigure[][{\scriptsize Service Rate (\textit{CHD})}]{
				\raisebox{-1ex}{\scalebox{0.19}[0.17]{\includegraphics{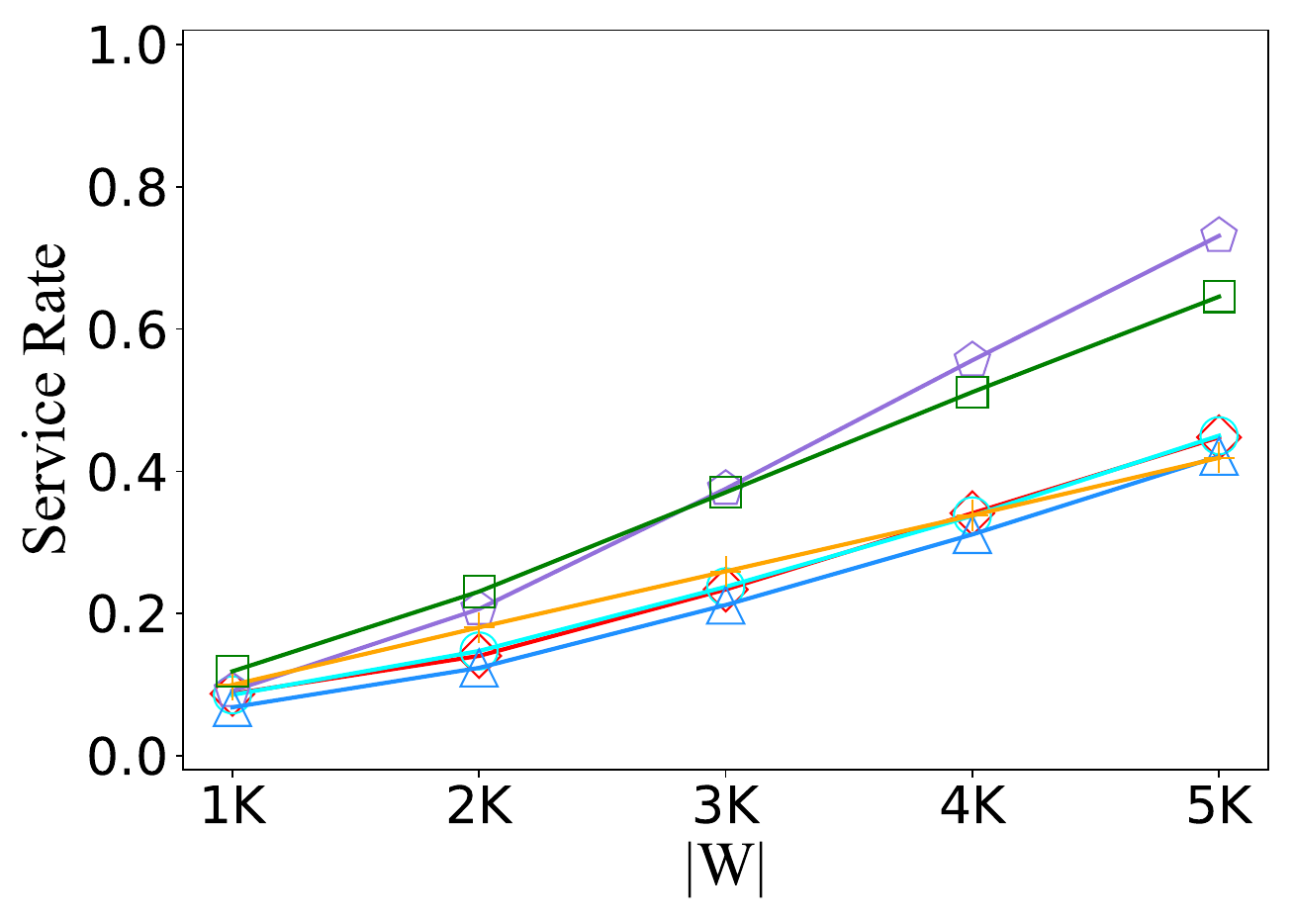}}}
				\label{subfig:sr_varing_w_cd}}\hspace{-2ex}
			\subfigure[][{\scriptsize Service Rate (\textit{NYC})}]{
				\raisebox{-1ex}{\scalebox{0.19}[0.17]{\includegraphics{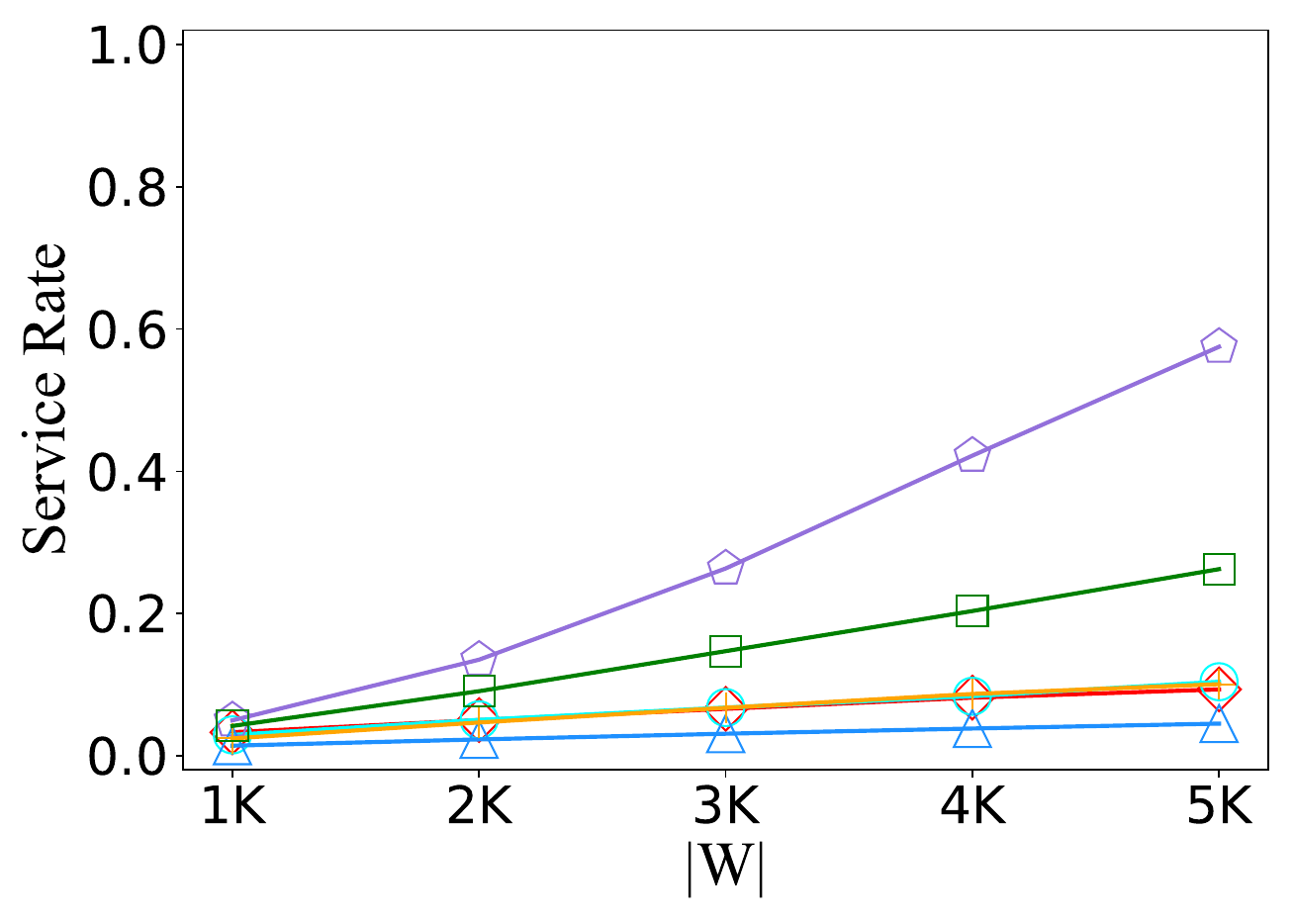}}}
				\label{subfig:sr_varing_w_nyc}}\\[-2ex]
			\subfigure[][{\scriptsize Running Time (\textit{CHD})}]{
				\raisebox{-1ex}{\scalebox{0.19}[0.17]{\includegraphics{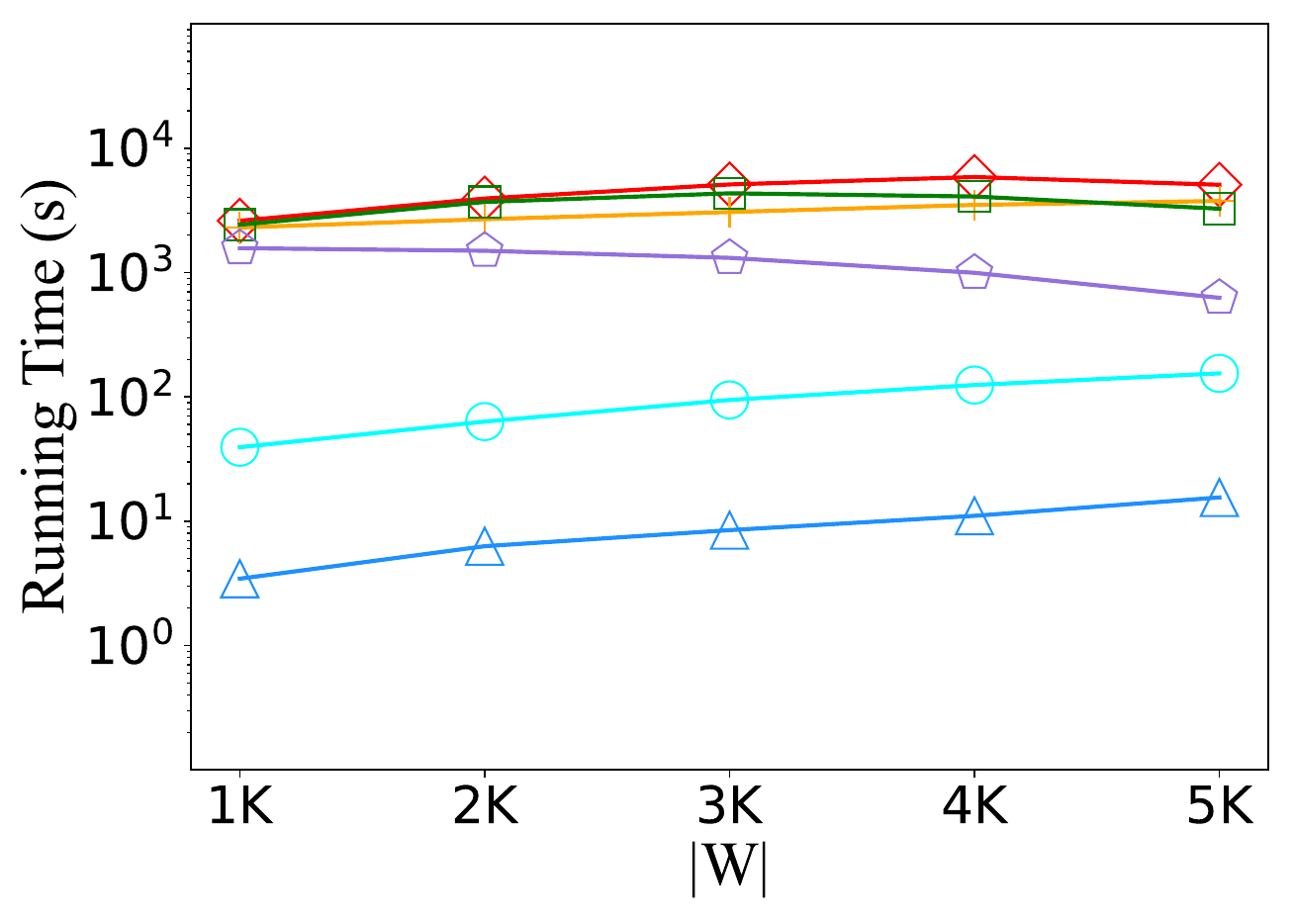}}}
				\label{subfig:tm_varing_w_cd}}\hspace{-2ex}
			\subfigure[][{\scriptsize Running Time (\textit{NYC})}]{
				\raisebox{-1ex}{\scalebox{0.19}[0.17]{\includegraphics{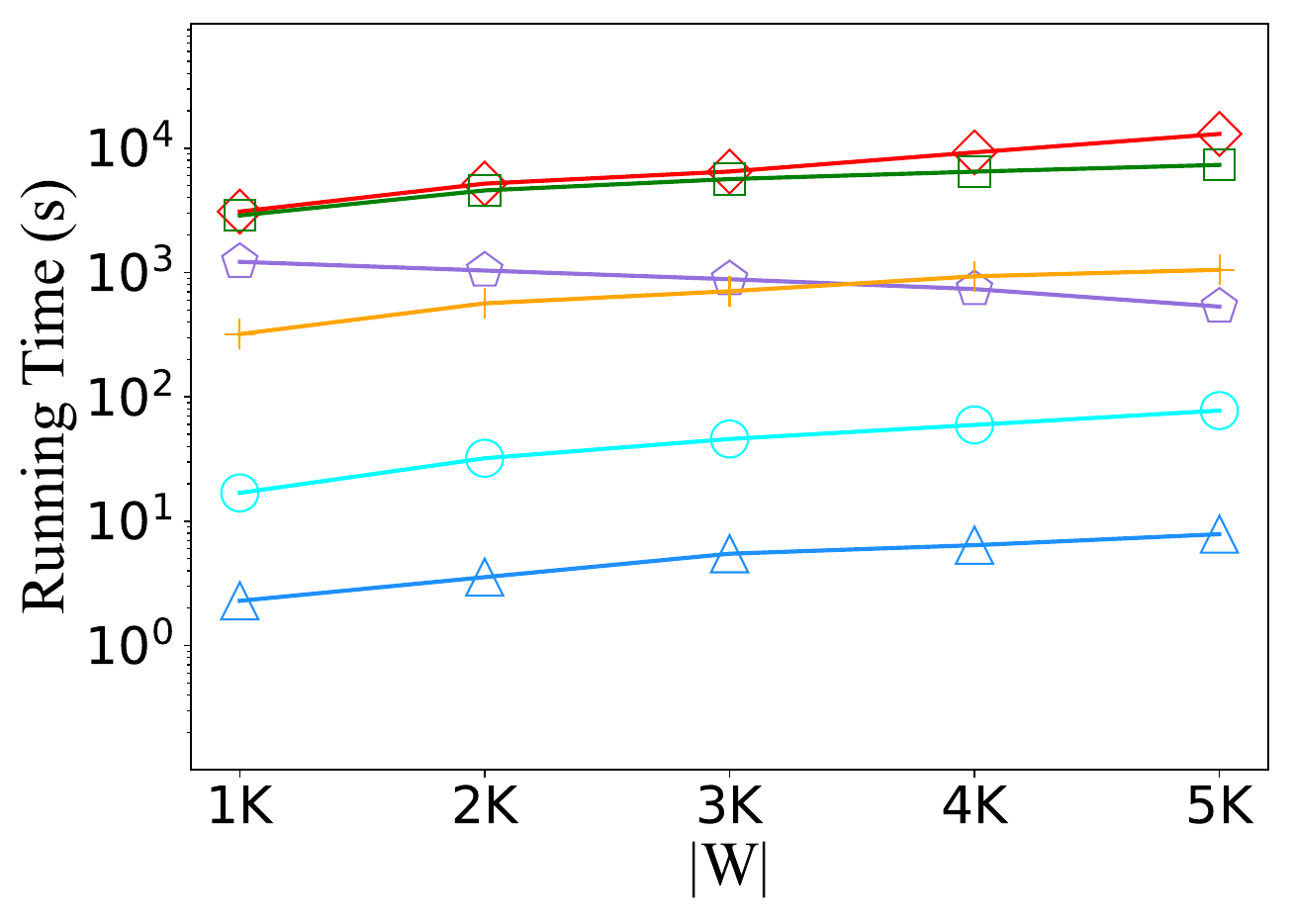}}}
				\label{subfig:tm_varing_w_nyc}}
			\caption{Performance of Varying $|W|$.}
			\label{fig:vary_worker}
		\end{minipage} &
		\begin{minipage}[c]{.5\textwidth}
			\subfigure{
				\scalebox{0.45}[0.45]{\includegraphics{legend-eps-converted-to.pdf}}}\hfill
			\addtocounter{subfigure}{-1}\\[-3ex]
			\subfigure[][{\scriptsize Unified Cost (\textit{CHD})}]{
				\raisebox{-1ex}{\scalebox{0.19}[0.17]{\includegraphics{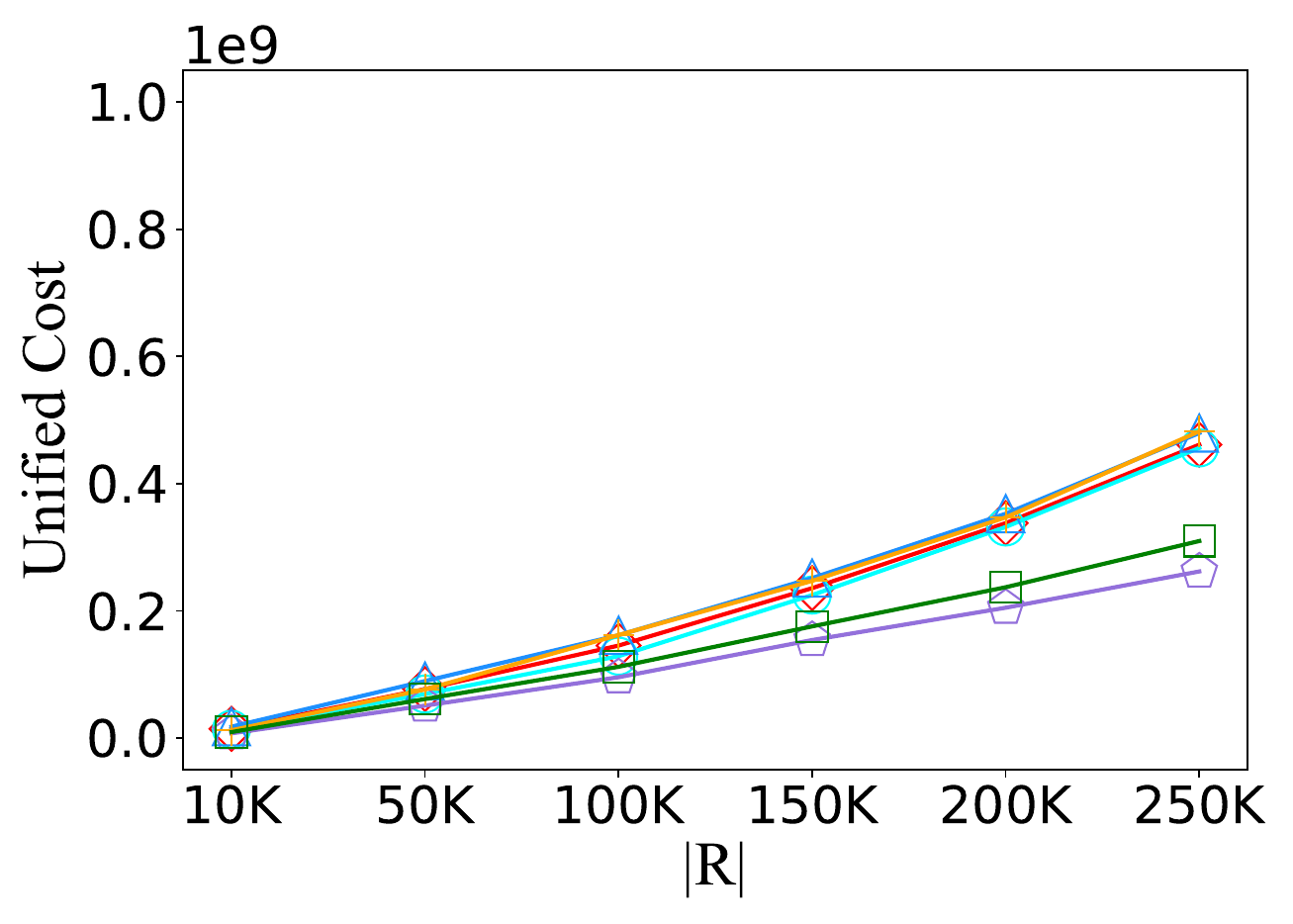}}}
				\label{subfig:uc_varing_req_cd}}\hspace{-2ex}
			\subfigure[][{\scriptsize Unified Cost (\textit{NYC})}]{
				\raisebox{-1ex}{\scalebox{0.19}[0.17]{\includegraphics{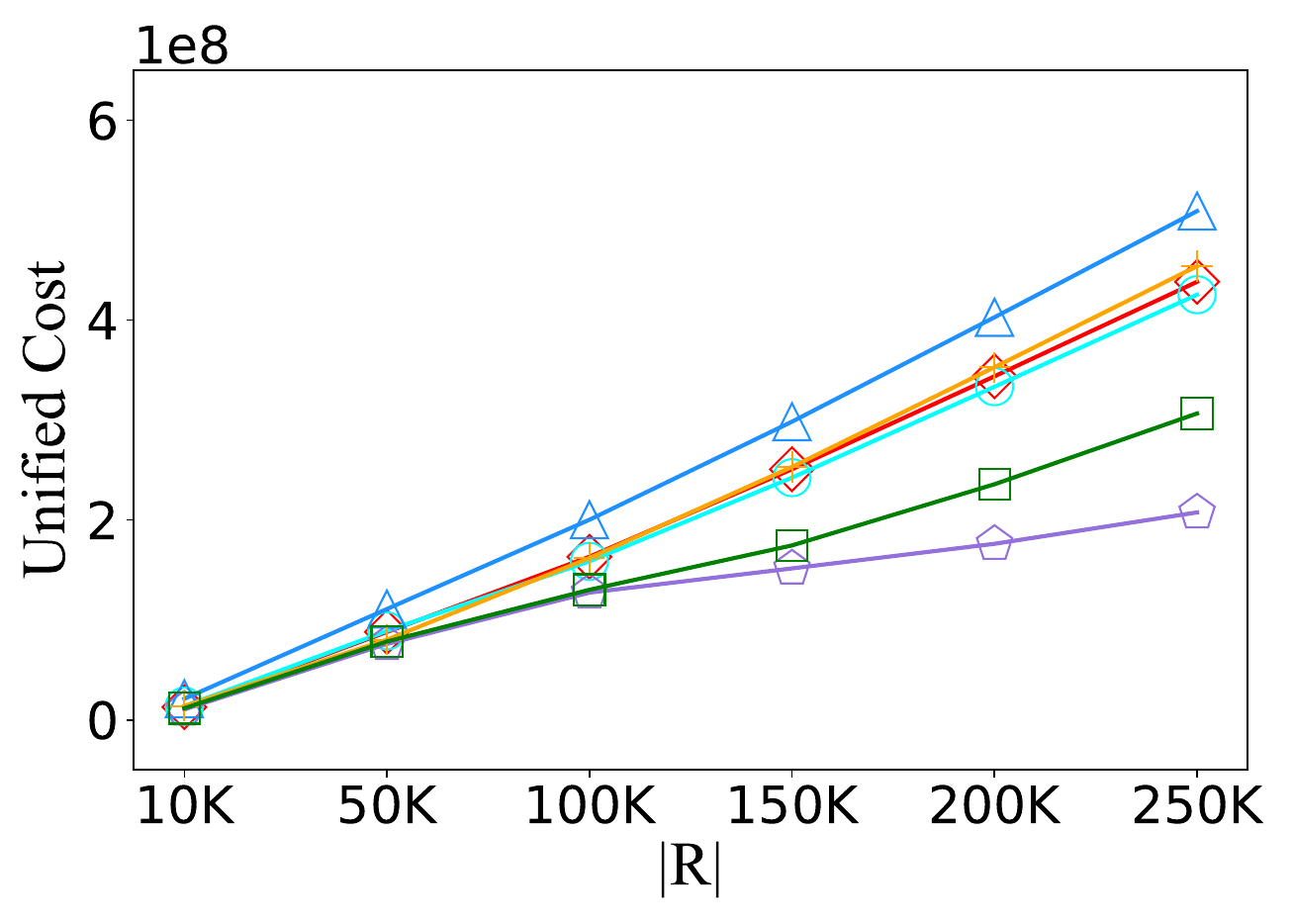}}}
				\label{subfig:uc_varing_req_nyc}}\\[-2ex]
			\subfigure[][{\scriptsize Service Rate (\textit{CHD})}]{
				\raisebox{-1ex}{\scalebox{0.19}[0.17]{\includegraphics{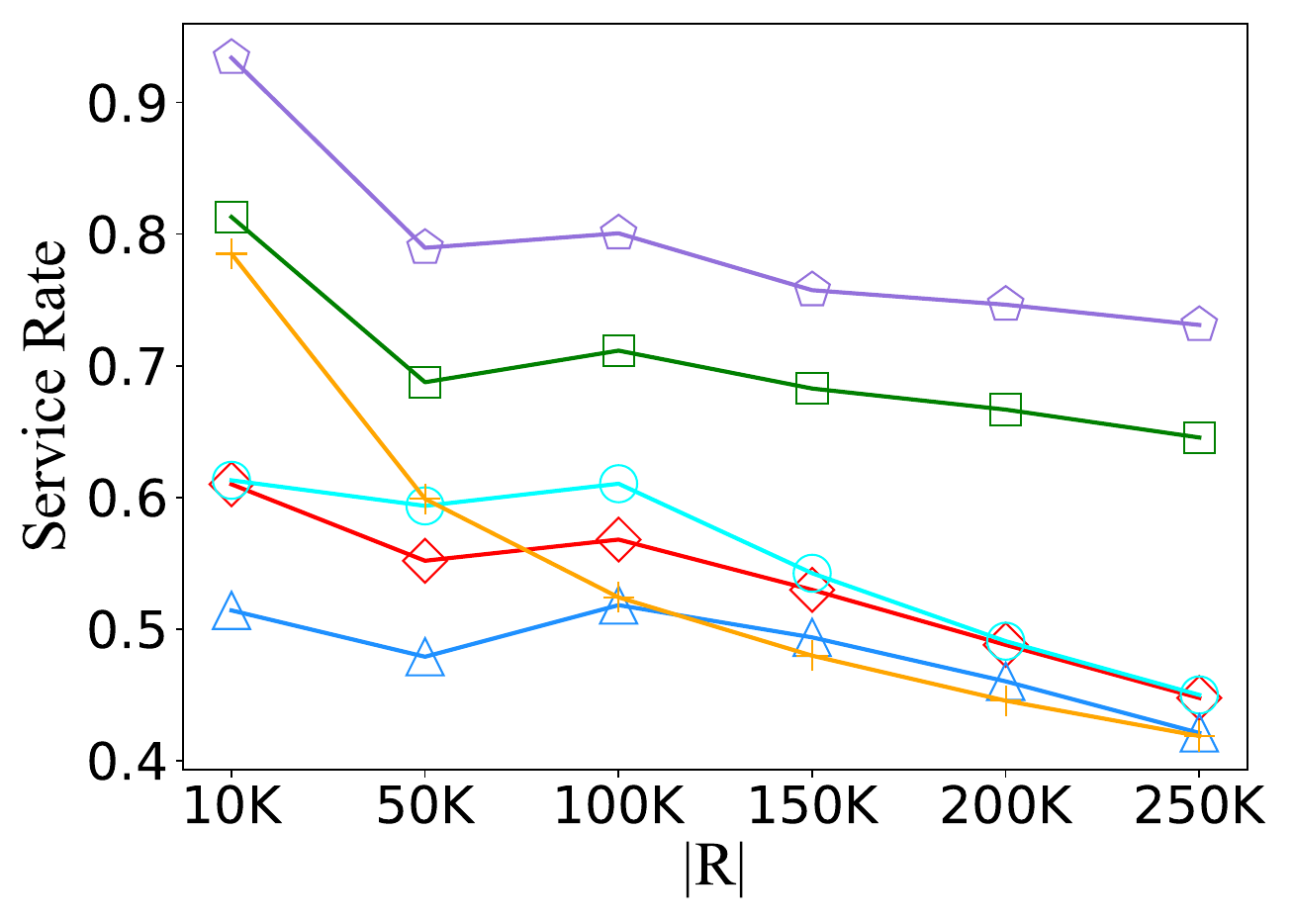}}}
				\label{subfig:sr_varing_req_cd}}\hspace{-2ex}
			\subfigure[][{\scriptsize Service Rate (\textit{NYC})}]{
				\raisebox{-1ex}{\scalebox{0.19}[0.17]{\includegraphics{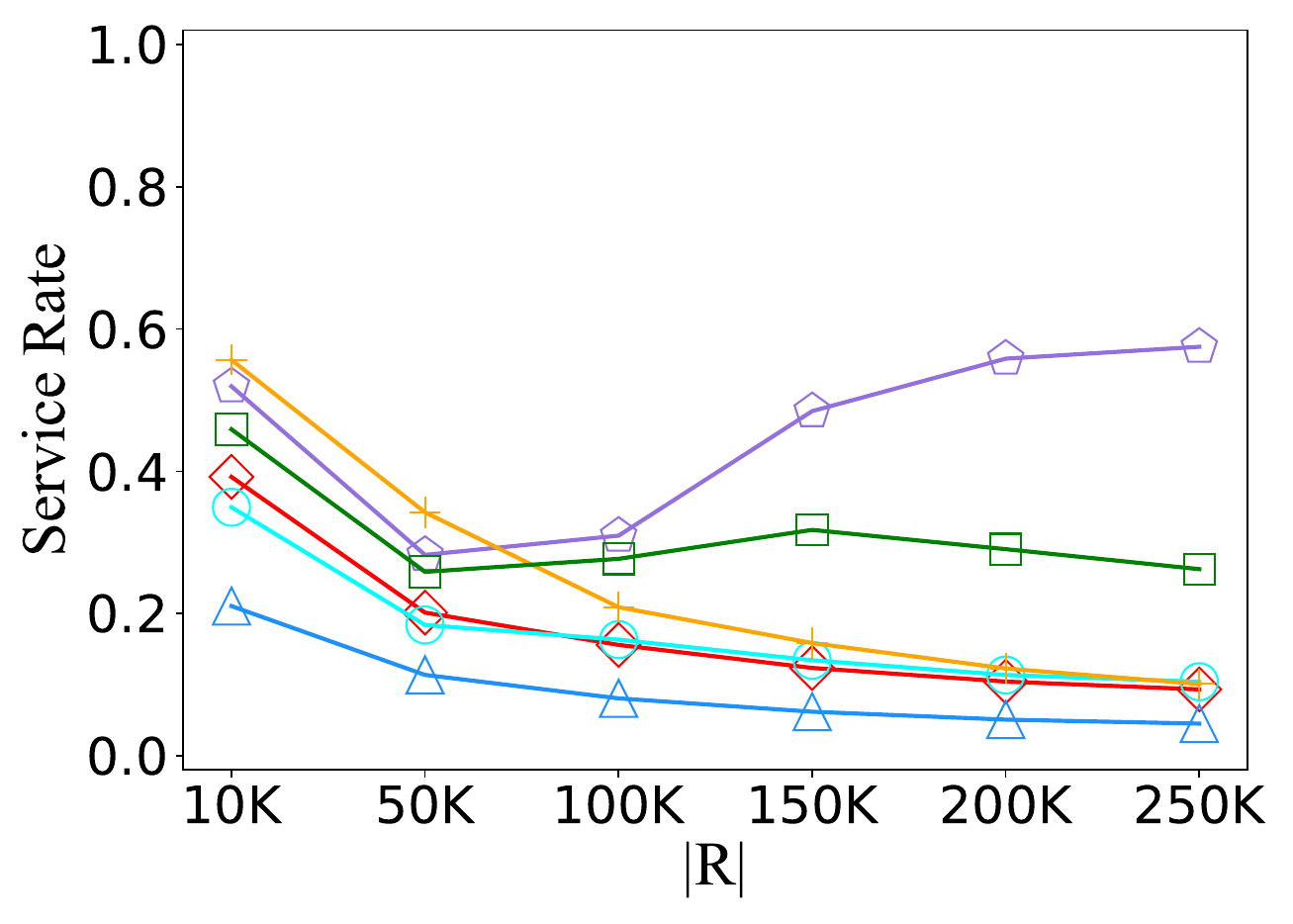}}}
				\label{subfig:sr_varing_req_nyc}}\\[-2ex]
			\subfigure[][{\scriptsize Running Time (\textit{CHD})}]{
				\raisebox{-1ex}{\scalebox{0.19}[0.17]{\includegraphics{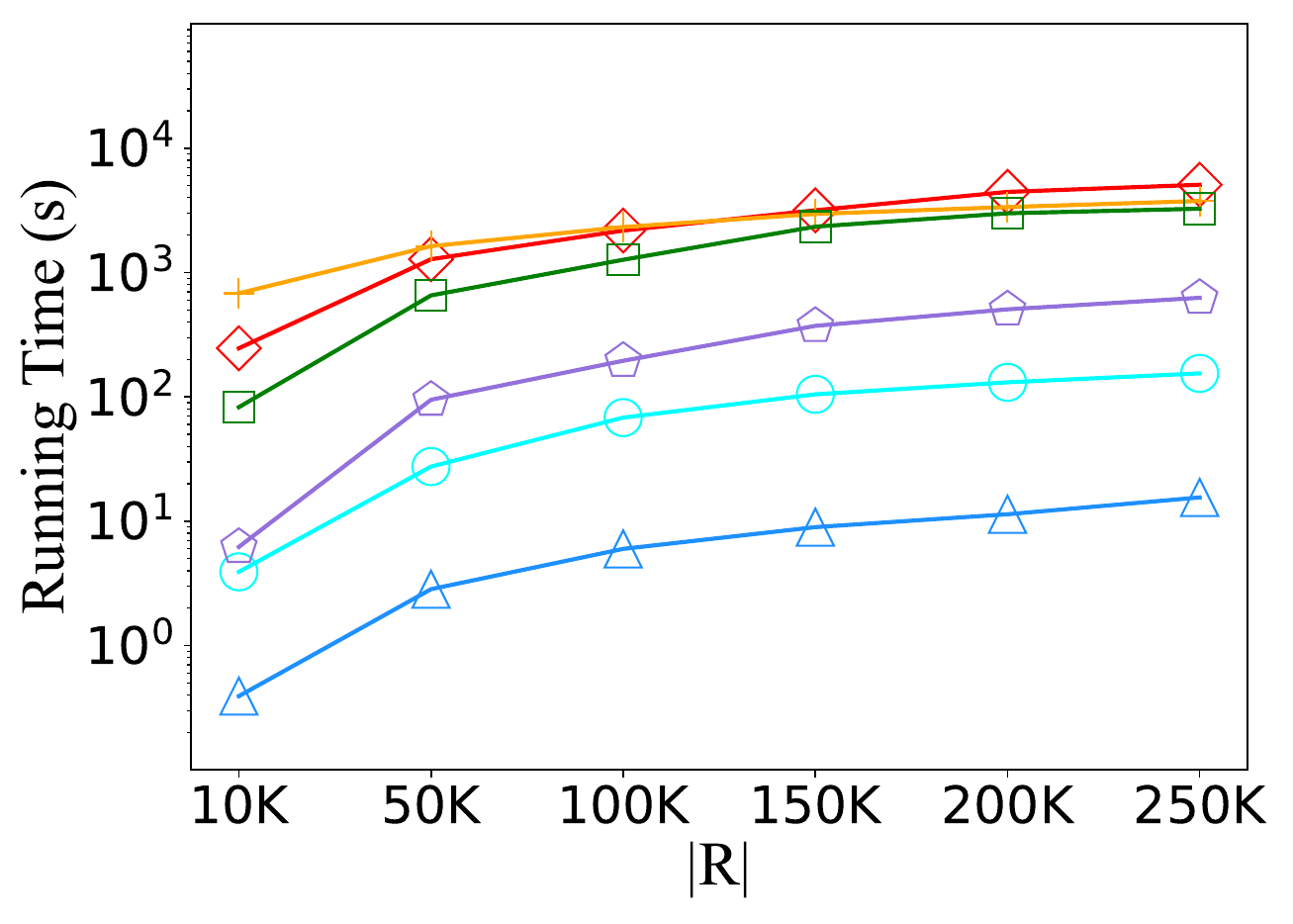}}}
				\label{subfig:tm_varing_req_cd}}\hspace{-2ex}
			\subfigure[][{\scriptsize Running Time (\textit{NYC})}]{
				\raisebox{-1ex}{\scalebox{0.19}[0.17]{\includegraphics{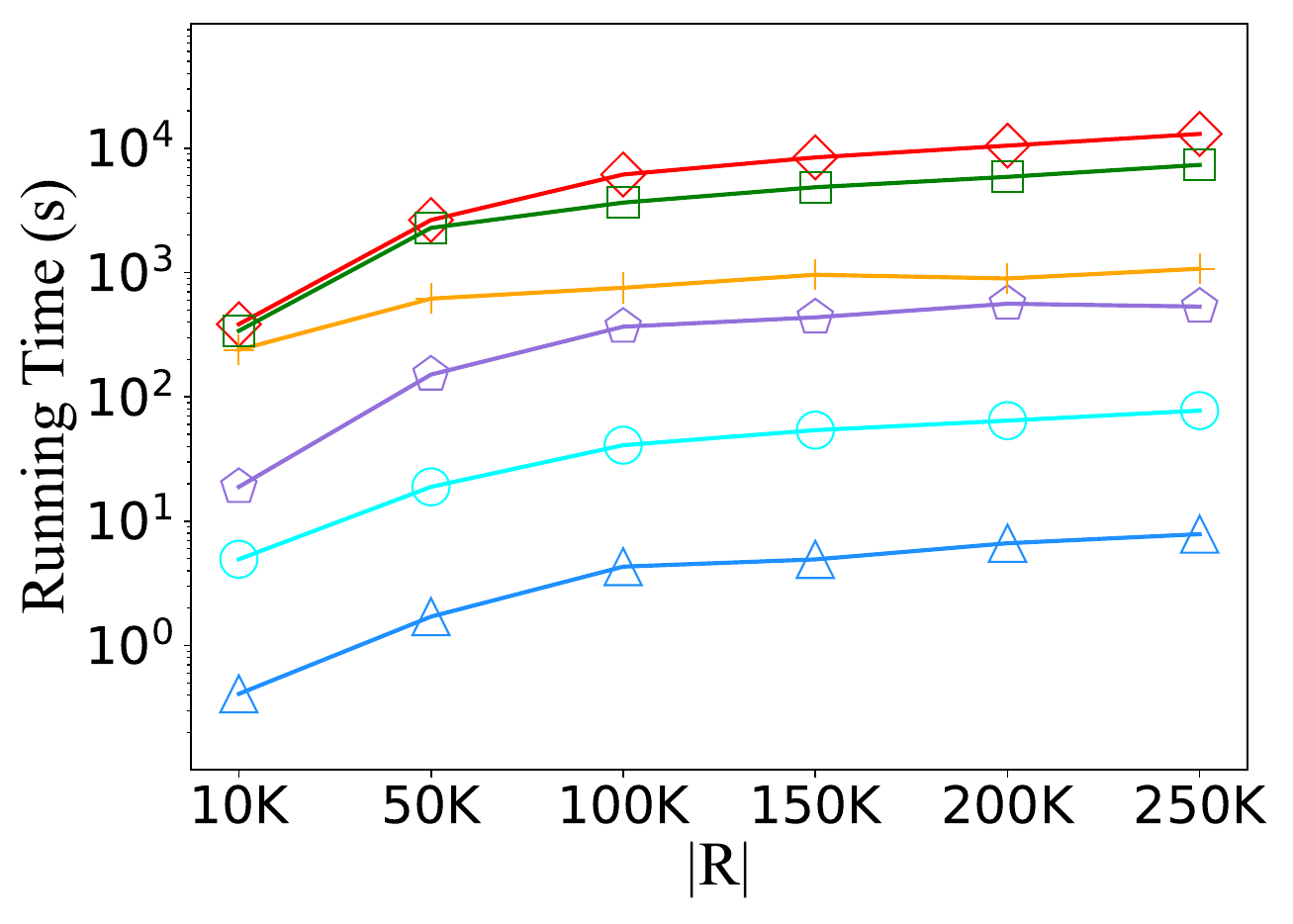}}}
				\label{subfig:tm_varing_req_nyc}}
			\caption{Performance of Varying $|R|$.}
			\label{fig:vary_req}
		\end{minipage}
	\end{tabularx}
\end{figure*}

\begin{figure*}[t!]\centering
	\begin{tabularx}{\textwidth}{XX}
		\begin{minipage}[t]{.5\textwidth}
			\subfigure{
				\scalebox{0.45}[0.45]{\includegraphics{legend-eps-converted-to.pdf}}}\hfill
			\addtocounter{subfigure}{-1}\\[-3ex]
			\subfigure[][{\scriptsize Unified Cost (\textit{CHD})}]{
				\raisebox{-1ex}{\scalebox{0.19}[0.17]{\includegraphics{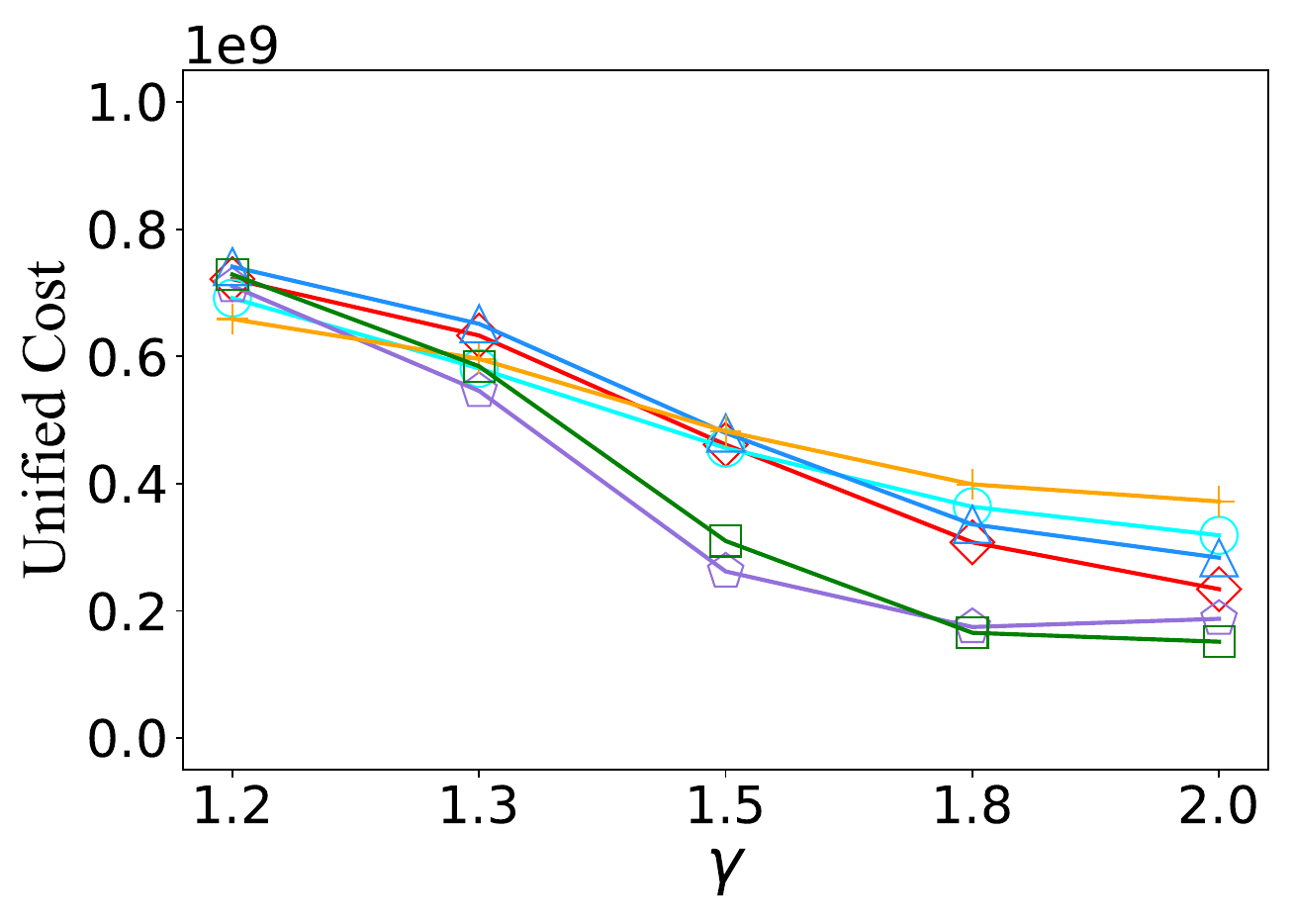}}}
				\label{subfig:uc_varing_ddl_cd}}\hspace{-2ex}
			\subfigure[][{\scriptsize Unified Cost (\textit{NYC})}]{
				\raisebox{-1ex}{\scalebox{0.19}[0.17]{\includegraphics{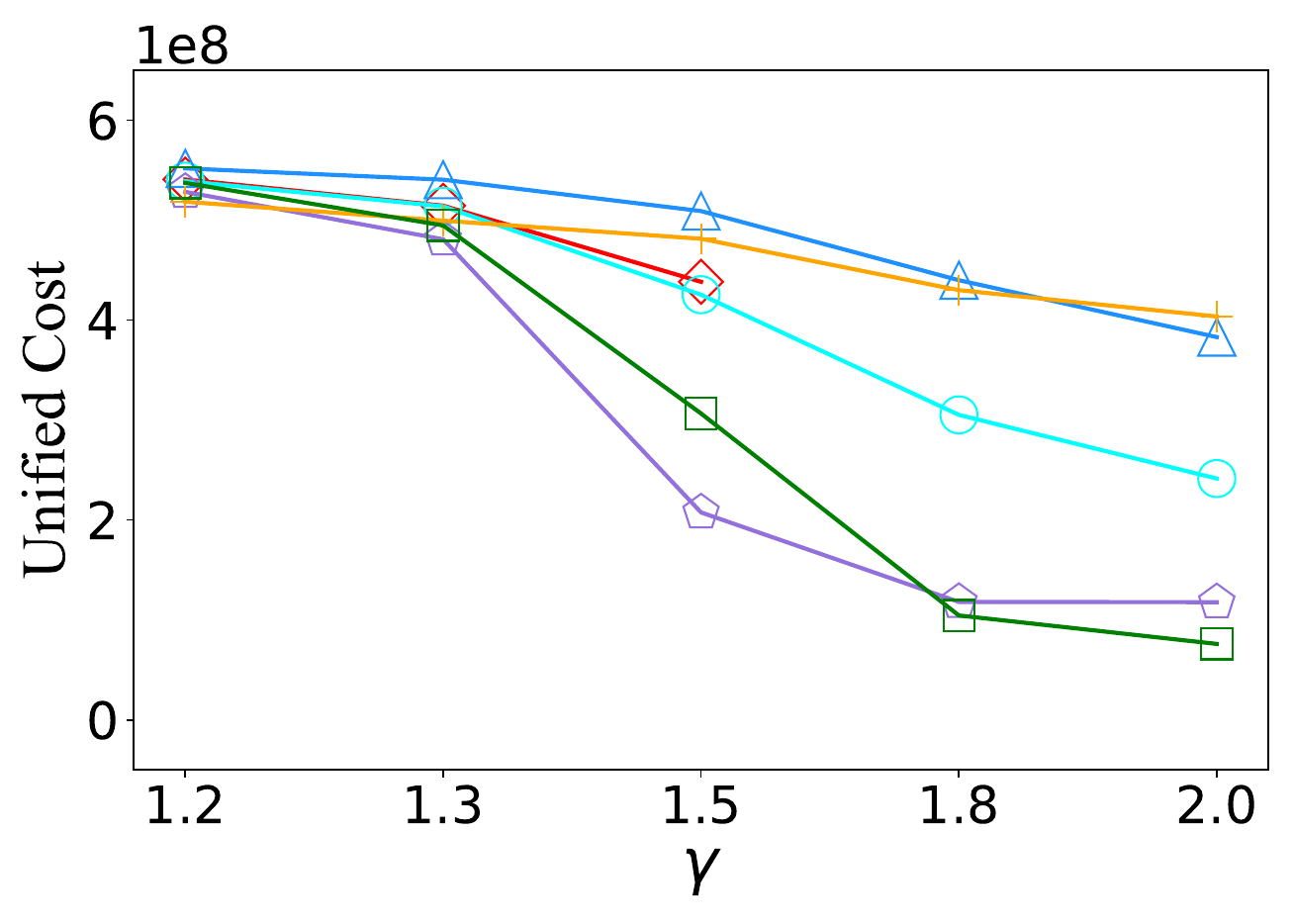}}}
				\label{subfig:uc_varing_ddl_nyc}}\\[-2ex]
			\subfigure[][{\scriptsize Service Rate (\textit{CHD})}]{
				\raisebox{-1ex}{\scalebox{0.19}[0.17]{\includegraphics{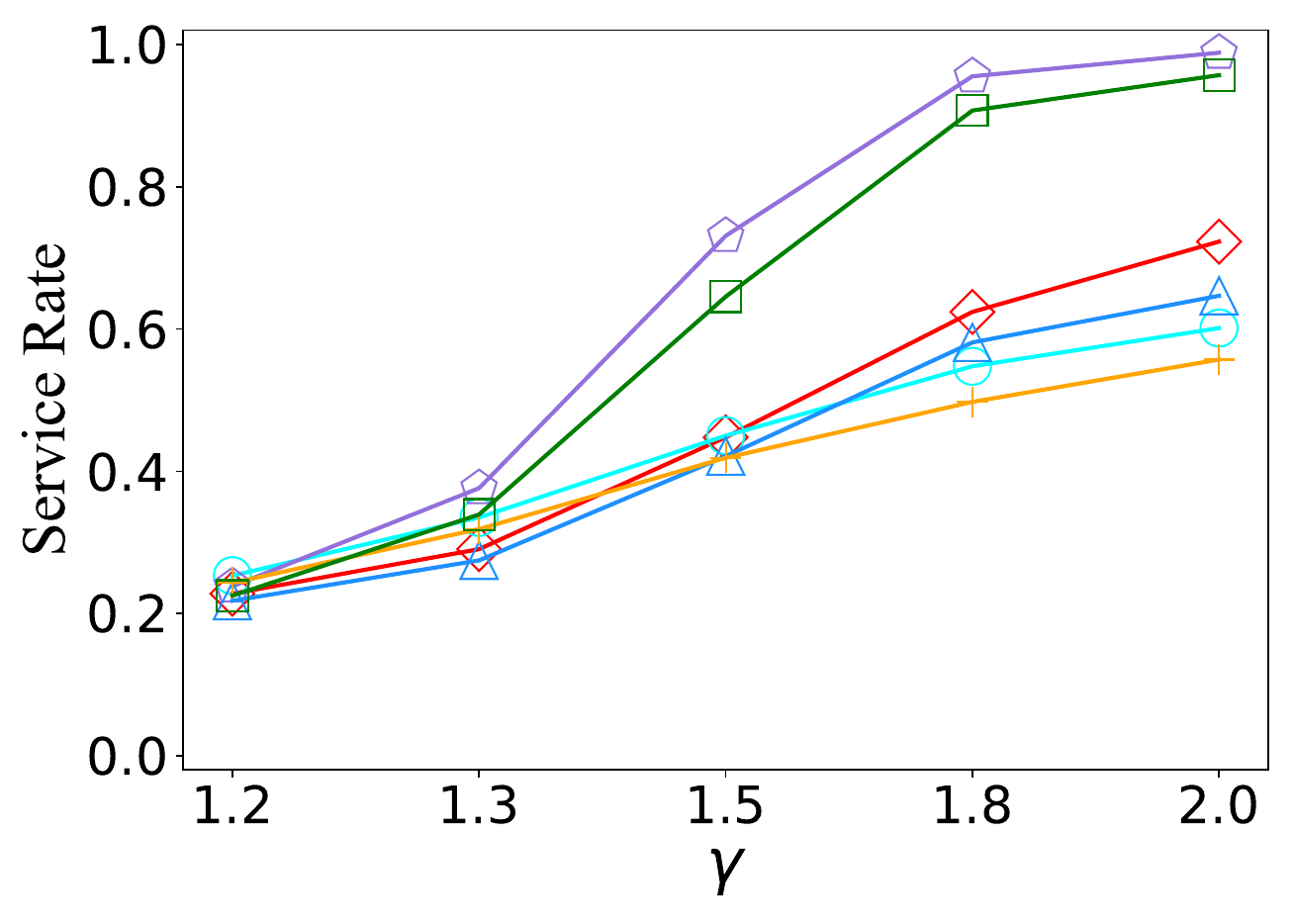}}}
				\label{subfig:sr_varing_ddl_cd}}\hspace{-2ex}
			\subfigure[][{\scriptsize Service Rate (\textit{NYC})}]{
				\raisebox{-1ex}{\scalebox{0.19}[0.17]{\includegraphics{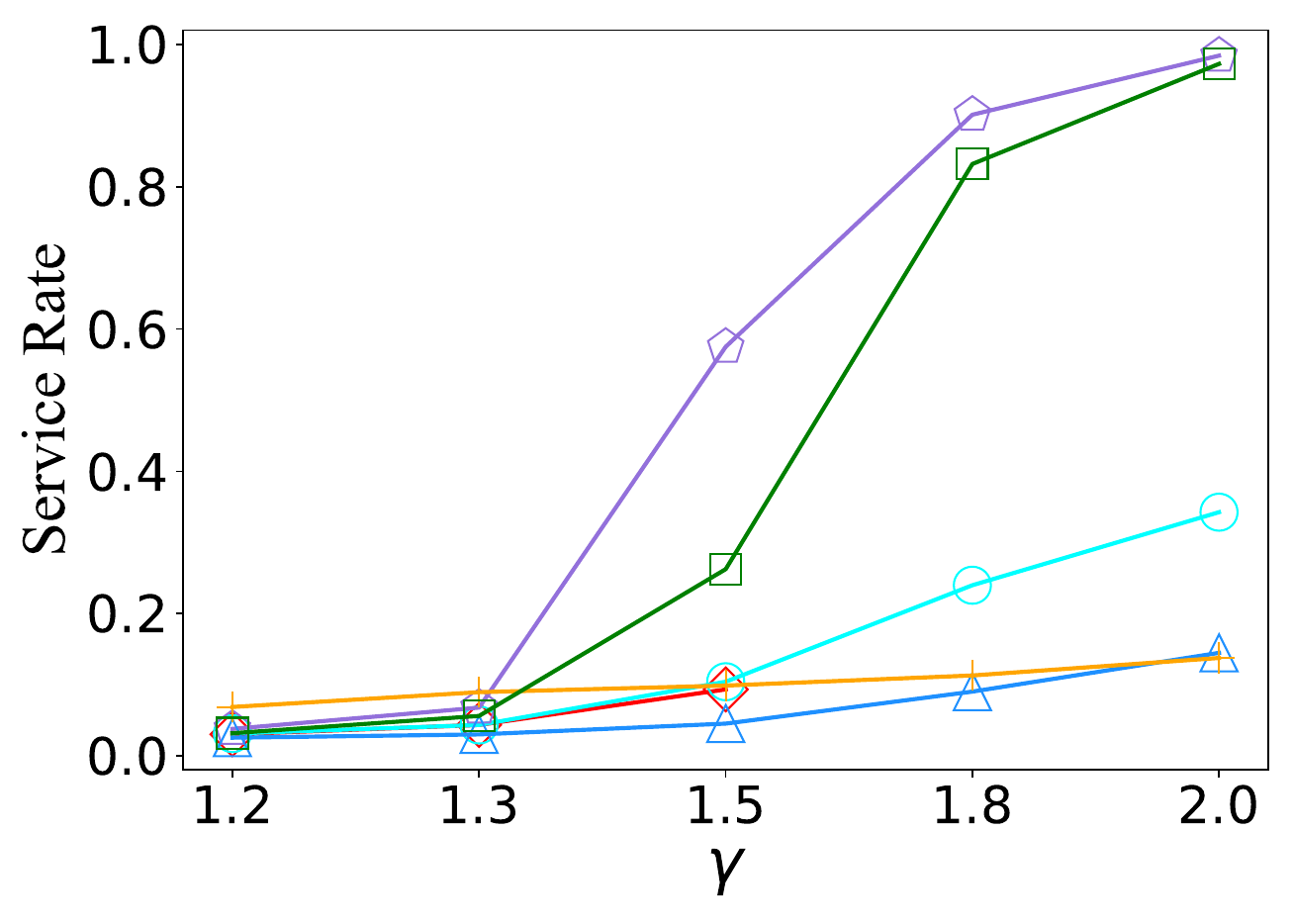}}}
				\label{subfig:sr_varing_ddl_nyc}}\\[-2ex]
			\subfigure[][{\scriptsize Running Time (\textit{CHD})}]{
				\raisebox{-1ex}{\scalebox{0.19}[0.17]{\includegraphics{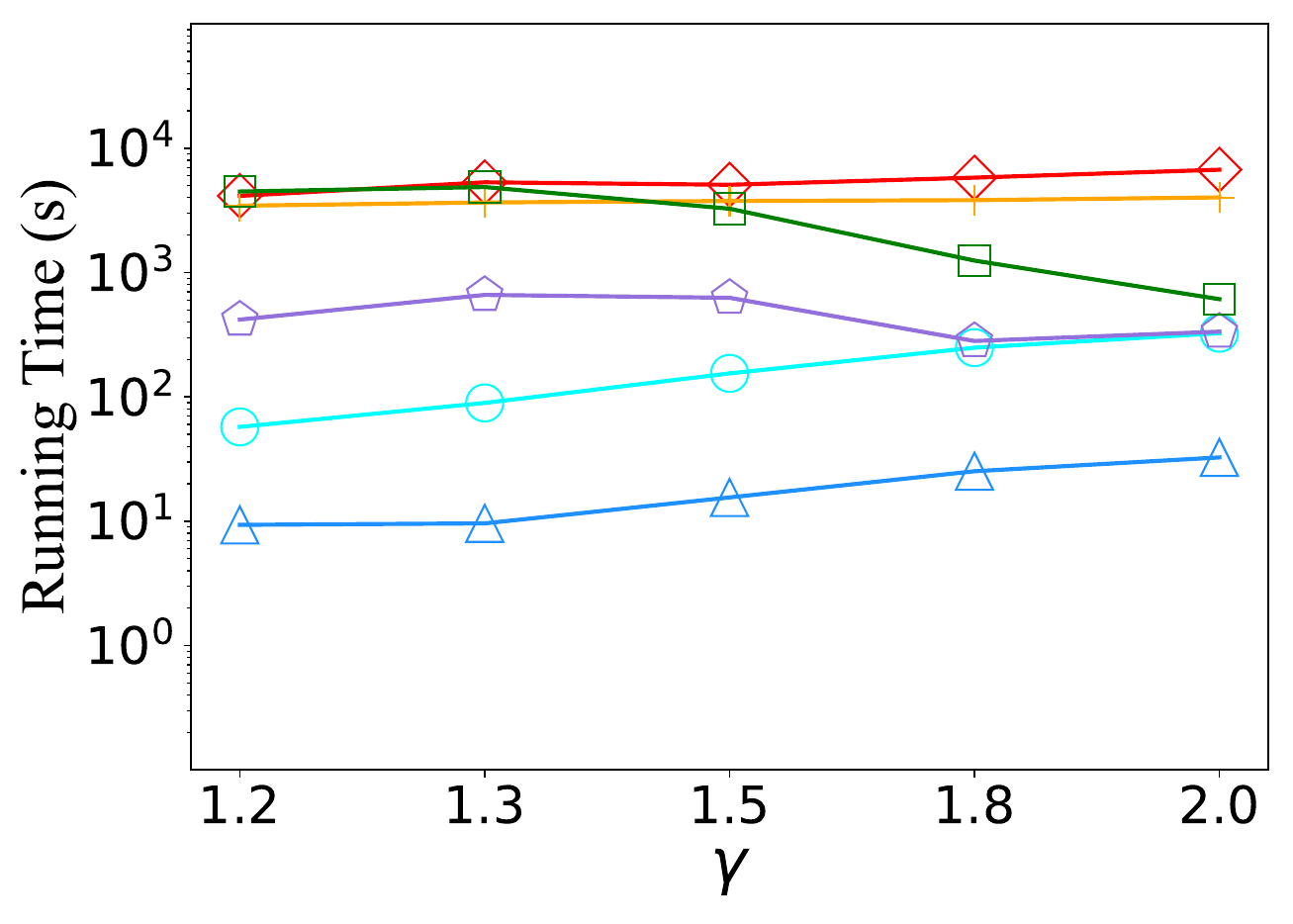}}}
				\label{subfig:tm_varing_ddl_cd}}\hspace{-2ex}
			\subfigure[][{\scriptsize Running Time (\textit{NYC})}]{
				\raisebox{-1ex}{\scalebox{0.19}[0.17]{\includegraphics{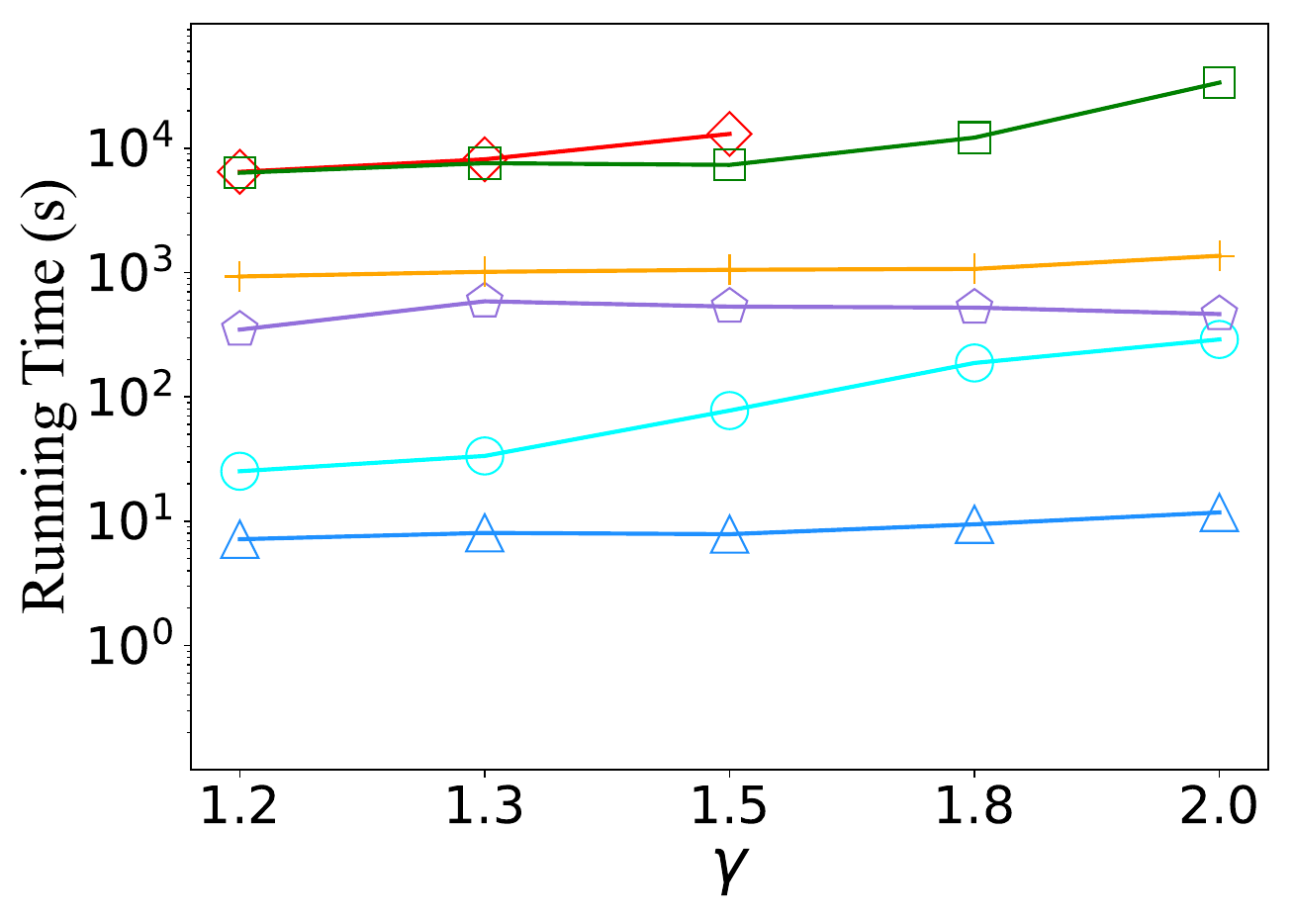}}}
				\label{subfig:tm_varing_ddl_nyc}}
			\caption{ Performance of Varying $\gamma$.}
			\label{fig:vary_ddl}
		\end{minipage} &
		\begin{minipage}[t]{.5\textwidth}
			\subfigure{
				\scalebox{0.45}[0.45]{\includegraphics{legend-eps-converted-to.pdf}}}\hfill
			\addtocounter{subfigure}{-1}\\[-3ex]
			\subfigure[][{\scriptsize Unified Cost (\textit{CHD})}]{
				\raisebox{-1ex}{\scalebox{0.19}[0.17]{\includegraphics{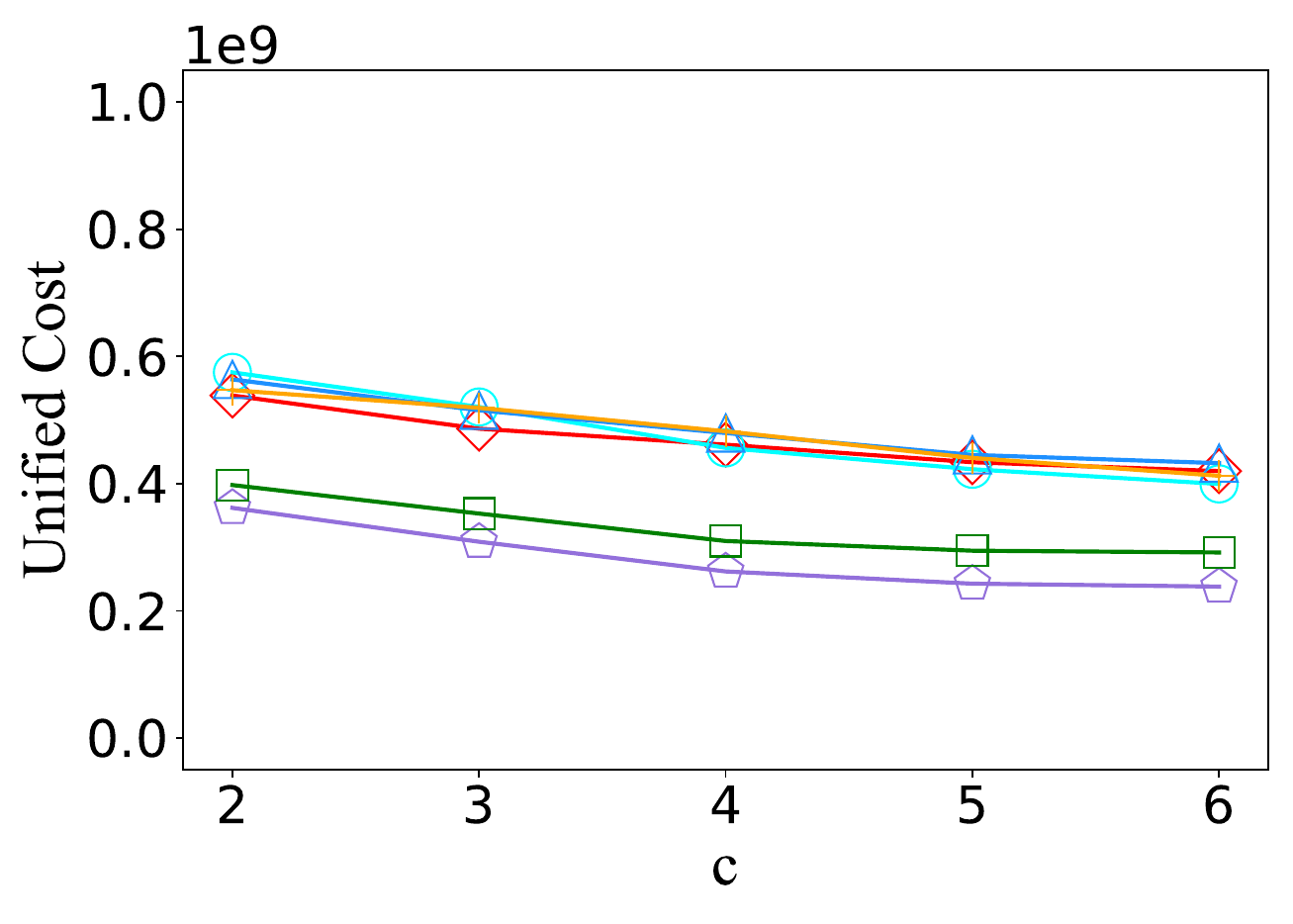}}}
				\label{subfig:uc_varing_cap_cd}}\hspace{-2ex}
			\subfigure[][{\scriptsize Unified Cost (\textit{NYC})}]{
				\raisebox{-1ex}{\scalebox{0.19}[0.17]{\includegraphics{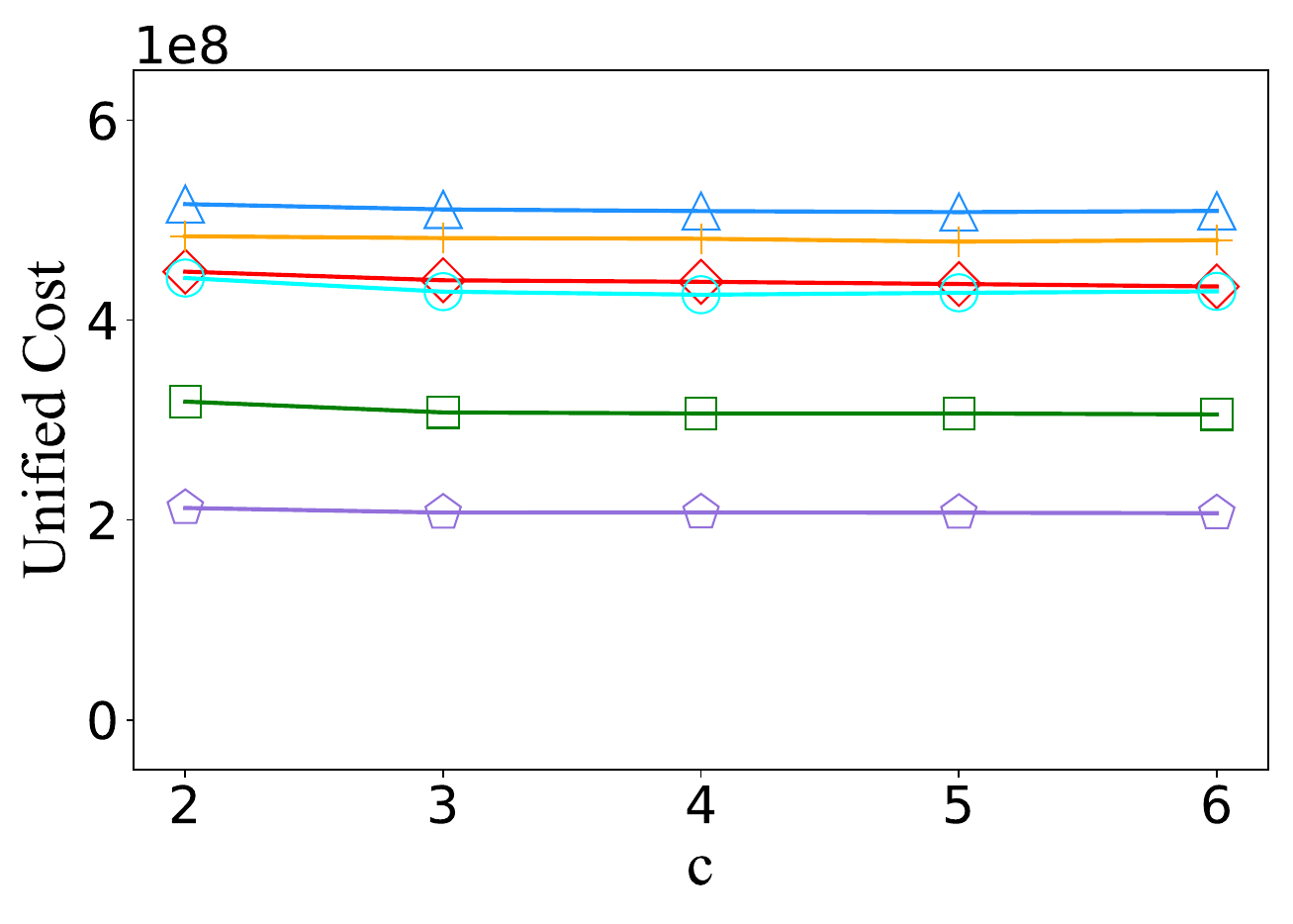}}}
				\label{subfig:uc_varing_cap_nyc}}\\[-2ex]
			\subfigure[][{\scriptsize Service Rate (\textit{CHD})}]{
				\raisebox{-1ex}{\scalebox{0.19}[0.17]{\includegraphics{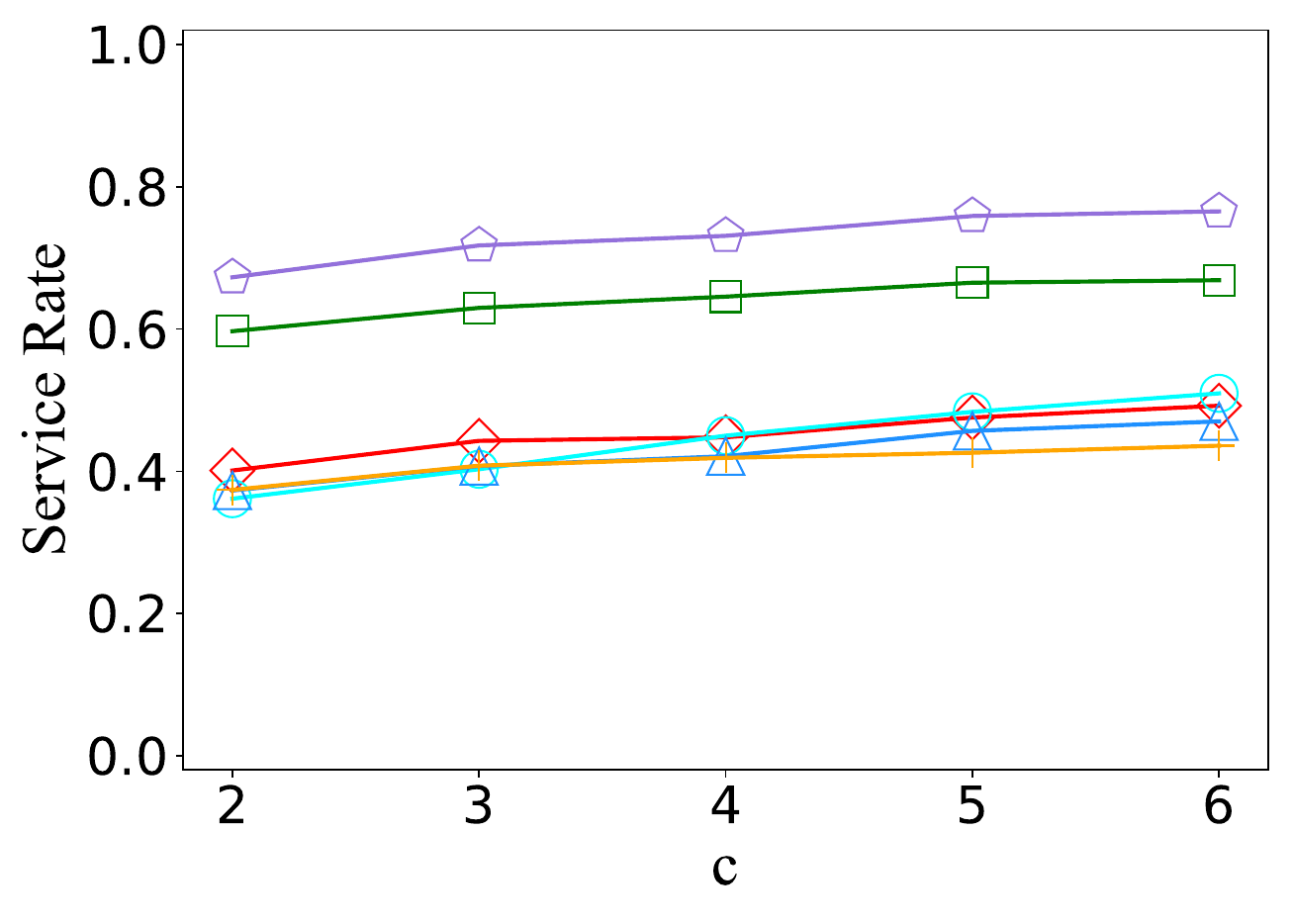}}}
				\label{subfig:sr_varing_cap_cd}}\hspace{-2ex}
			\subfigure[][{\scriptsize Service Rate (\textit{NYC})}]{
				\raisebox{-1ex}{\scalebox{0.19}[0.17]{\includegraphics{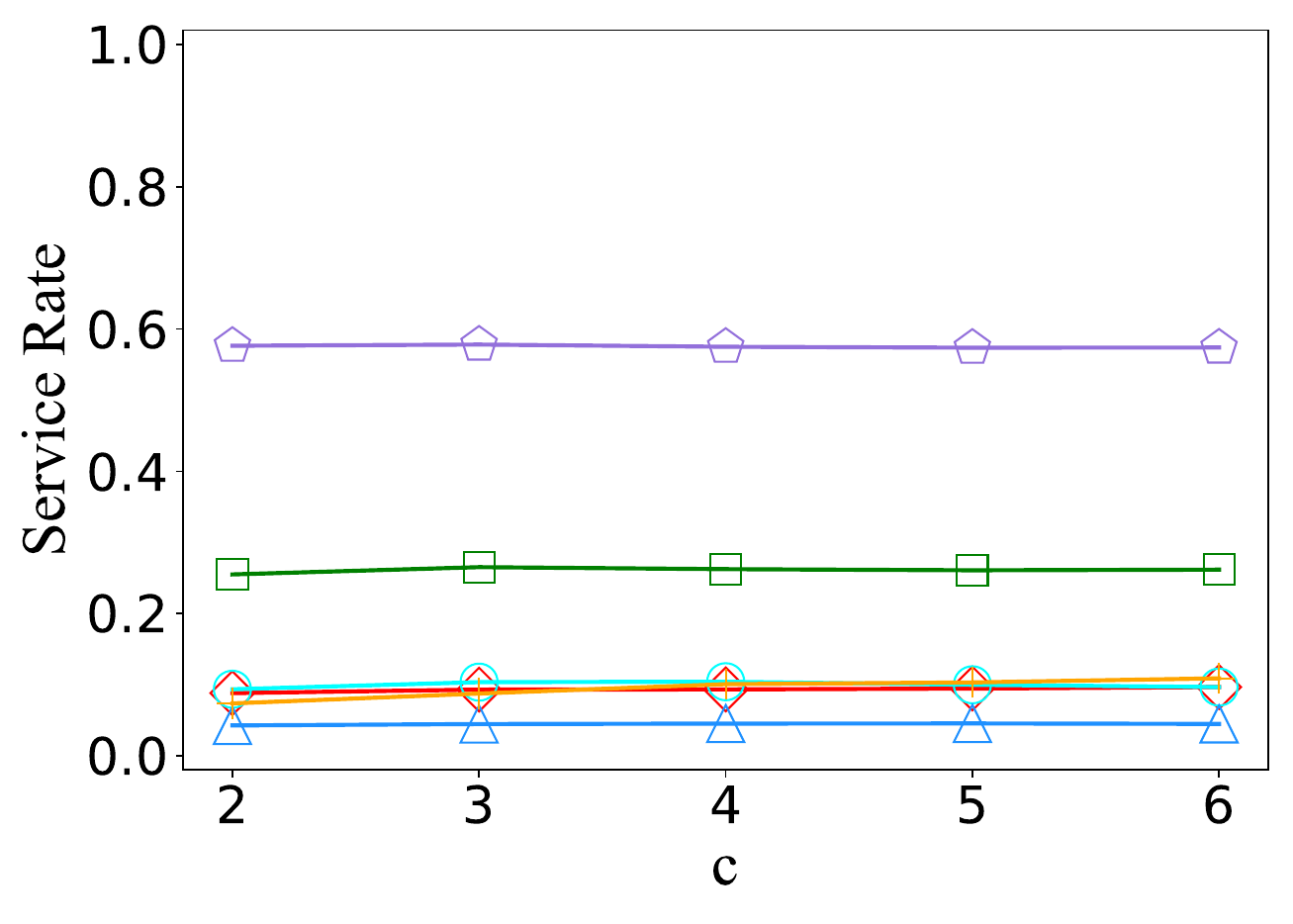}}}
				\label{subfig:sr_varing_cap_nyc}}\\[-2ex]
			\subfigure[][{\scriptsize Running Time (\textit{CHD})}]{
				\raisebox{-1ex}{\scalebox{0.19}[0.17]{\includegraphics{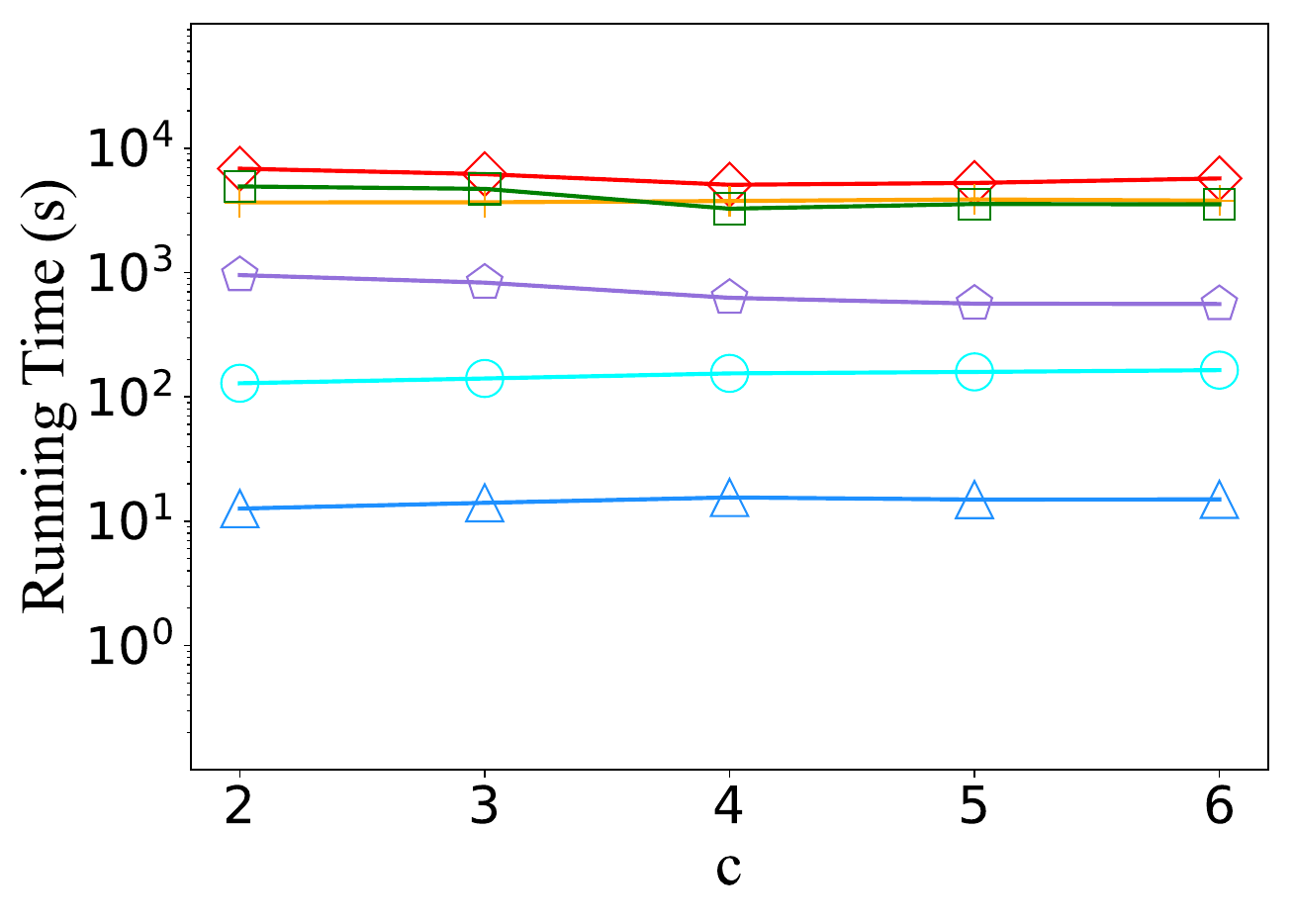}}}
				\label{subfig:tm_varing_cap_cd}}\hspace{-2ex}
			\subfigure[][{\scriptsize Running Time (\textit{NYC})}]{
				\raisebox{-1ex}{\scalebox{0.19}[0.17]{\includegraphics{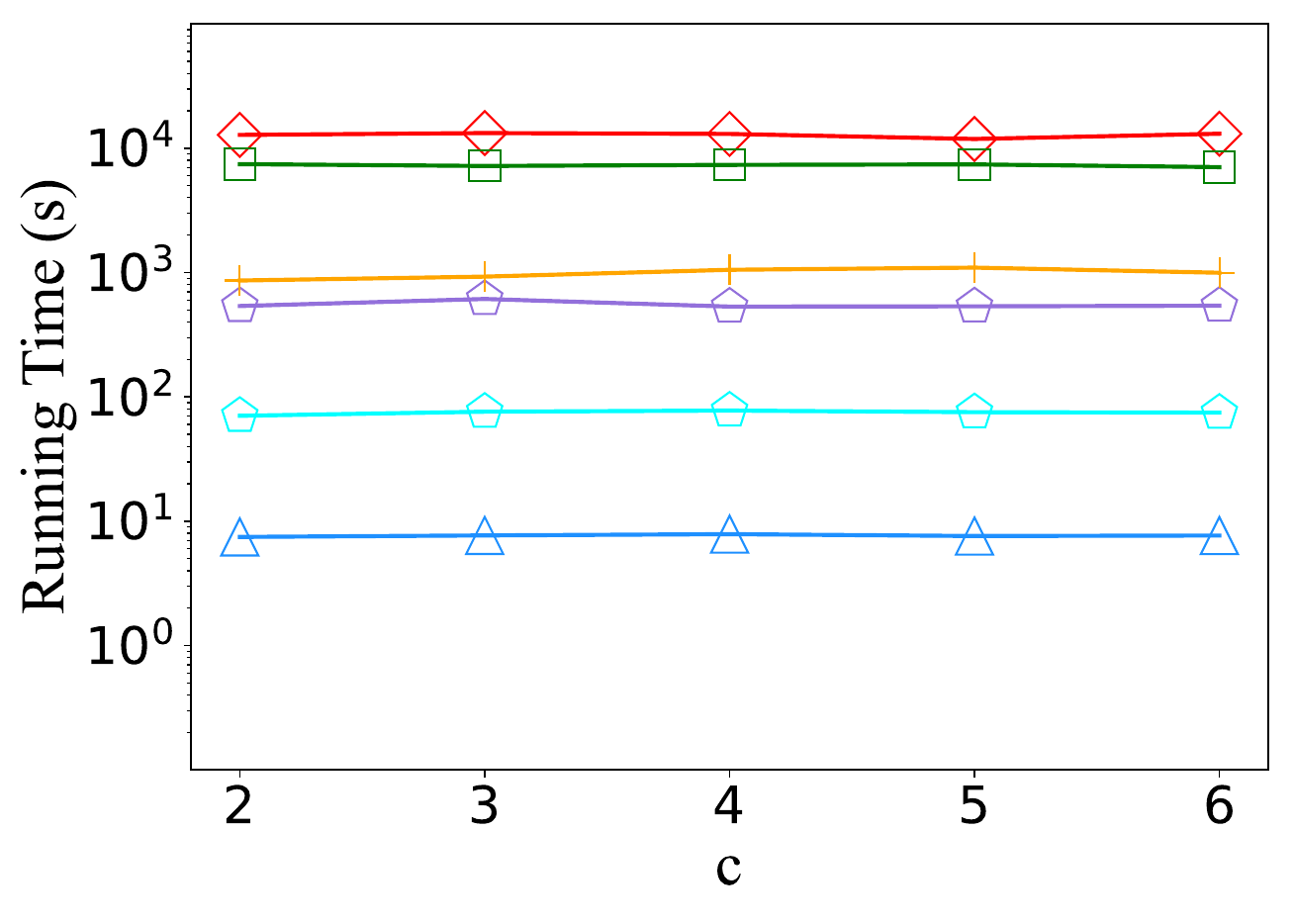}}}
				\label{subfig:tm_varing_cap_nyc}}
			\caption{ Performance of Varying $c$.}
			\label{fig:vary_cap}
		\end{minipage}
	\end{tabularx}
\end{figure*}

\subsection{Experimental Results}
\label{subsec:exp_res}
\noindent\textbf{Effect of the number of vehicles.}
Figure~\ref{fig:vary_worker} represents the results for 1K to 5K vehicles.
For the unified cost, {SARD} and {GAS} are leading the other methods, and the margin between these two algorithms and the others are widening gradually.
In the \textit{CHD} dataset, the results of the two algorithms are very close, 
but {SARD} achieves an improvement of $2.09\%\sim 59.22\%$ compared with other methods in the \textit{NYC} dataset.
As for the service rate, since the penalty part of the unified cost caused by the rejected requests decreases when the algorithm achieves a higher overall service rate, the gap between the algorithms is consistent with unified cost trends.
\revision{However, {DARM+DPRS} is an exception to this trend, as its scheduling of more idle vehicles increases travel time.
	{TicketAssign+} achieves better service rates through simultaneous decision-making, thus escaping the local optima of the greedy algorithm.}
Our method improves service rates by up to $30.99\%$ and $52.99\%$ on two datasets.
{pruneGDP} and \revision{{TicketAssign+}} excel in running time due to the efficiency of linear insertion\revision{, while {TicketAssign+} experiences reduced speed as a result of vehicle resource contention}.
RTV, GAS, and SARD are slower than pruneGDP.
Specifically, {SARD} achieves a speedup ratio up to $7.12\times$ and $23.5\times$ than RTV and GAS on two datasets, respectively.
In addition, the running time of {SARD} even decreases with the increase of vehicle number, which is mainly due to the decrease in the number of propose-acceptance rounds.  
With more vehicles, fewer riders will conflict, and thus fewer rounds of the proposal are needed. 
This decreases the average proposed riders per vehicle, shortening the acceptance stage.
\revision{DARM+DPRS showing more competitive speed on \textit{NYC} but approaching batch-based algorithm runtime on \textit{CHD} due to its larger state space.}

\begin{figure*}[t!]\centering
	\begin{tabularx}{\textwidth}{XX}
		\begin{minipage}[t]{.5\textwidth}
			\subfigure{
				\scalebox{0.45}[0.45]{\includegraphics{legend-eps-converted-to.pdf}}}\hfill
			\addtocounter{subfigure}{-1}\\[-3ex]
			\subfigure[][{\scriptsize Unified Cost (\textit{CHD})}]{
				\raisebox{-1ex}{\scalebox{0.19}[0.17]{\includegraphics{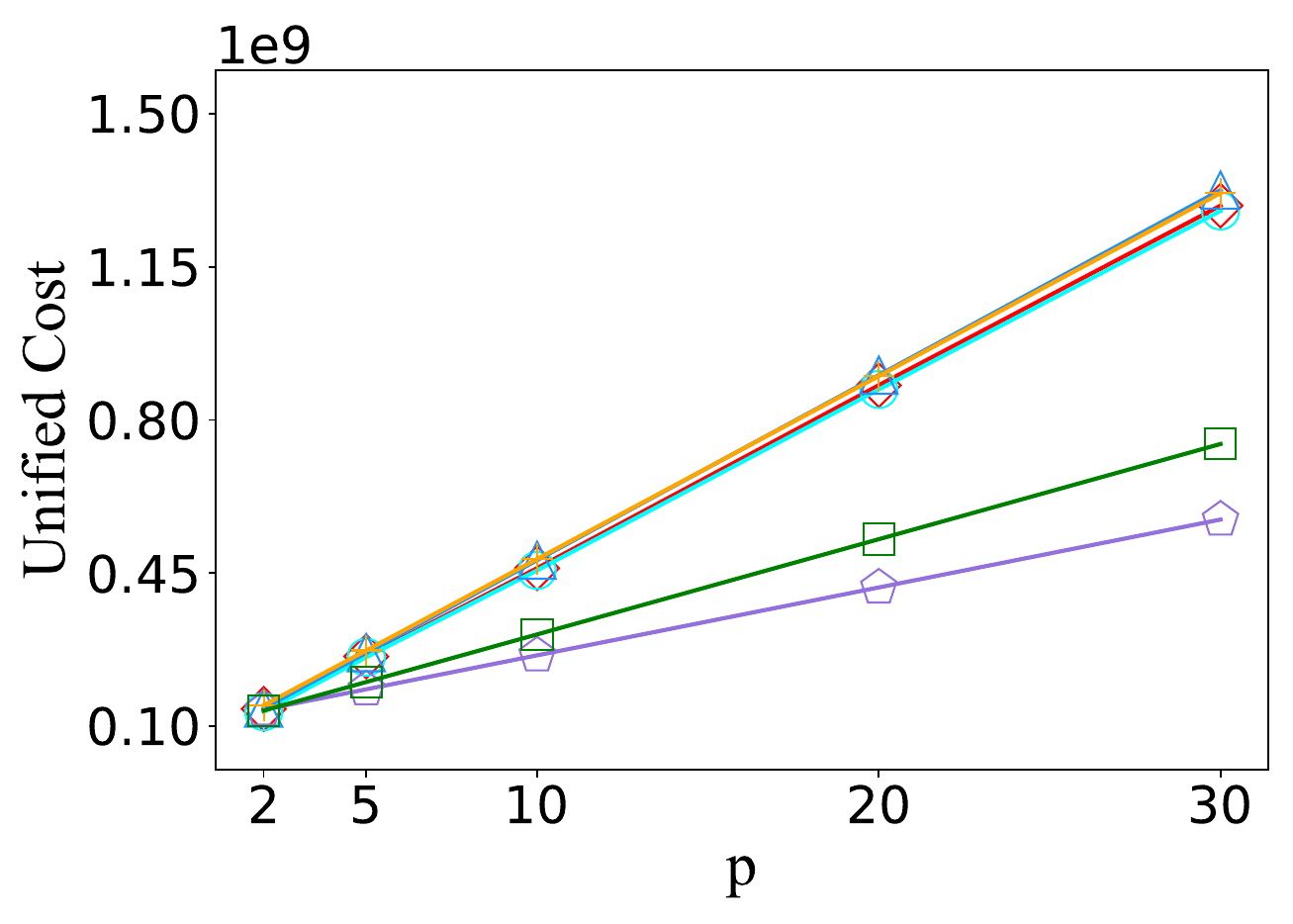}}}
				\label{subfig:uc_varing_pen_cd}}\hspace{-2ex}
			\subfigure[][{\scriptsize Unified Cost (\textit{NYC})}]{
				\raisebox{-1ex}{\scalebox{0.19}[0.17]{\includegraphics{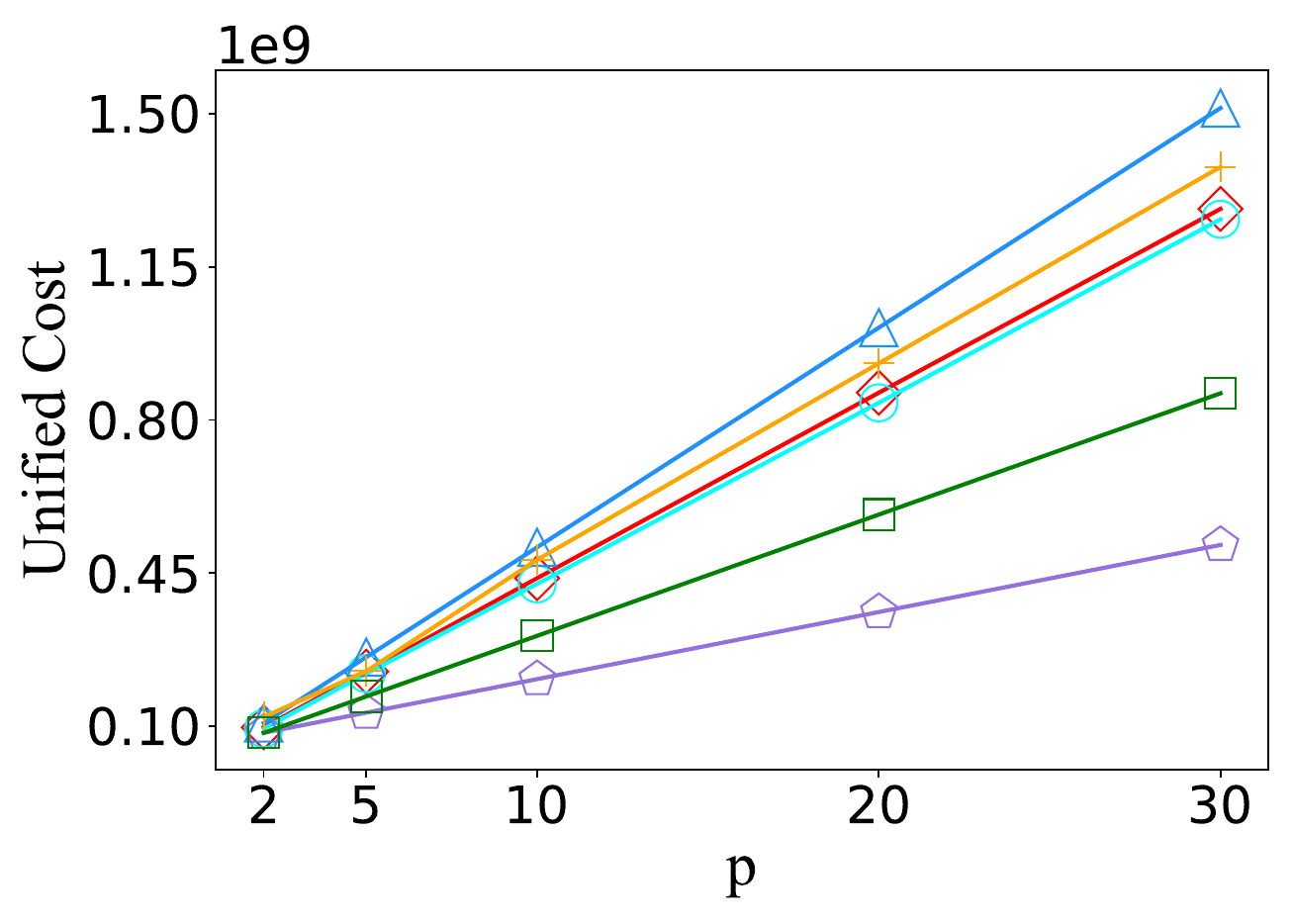}}}
				\label{subfig:uc_varing_pen_nyc}}\\[-2ex]
			\subfigure[][{\scriptsize Service Rate (\textit{CHD})}]{
				\raisebox{-1ex}{\scalebox{0.19}[0.17]{\includegraphics{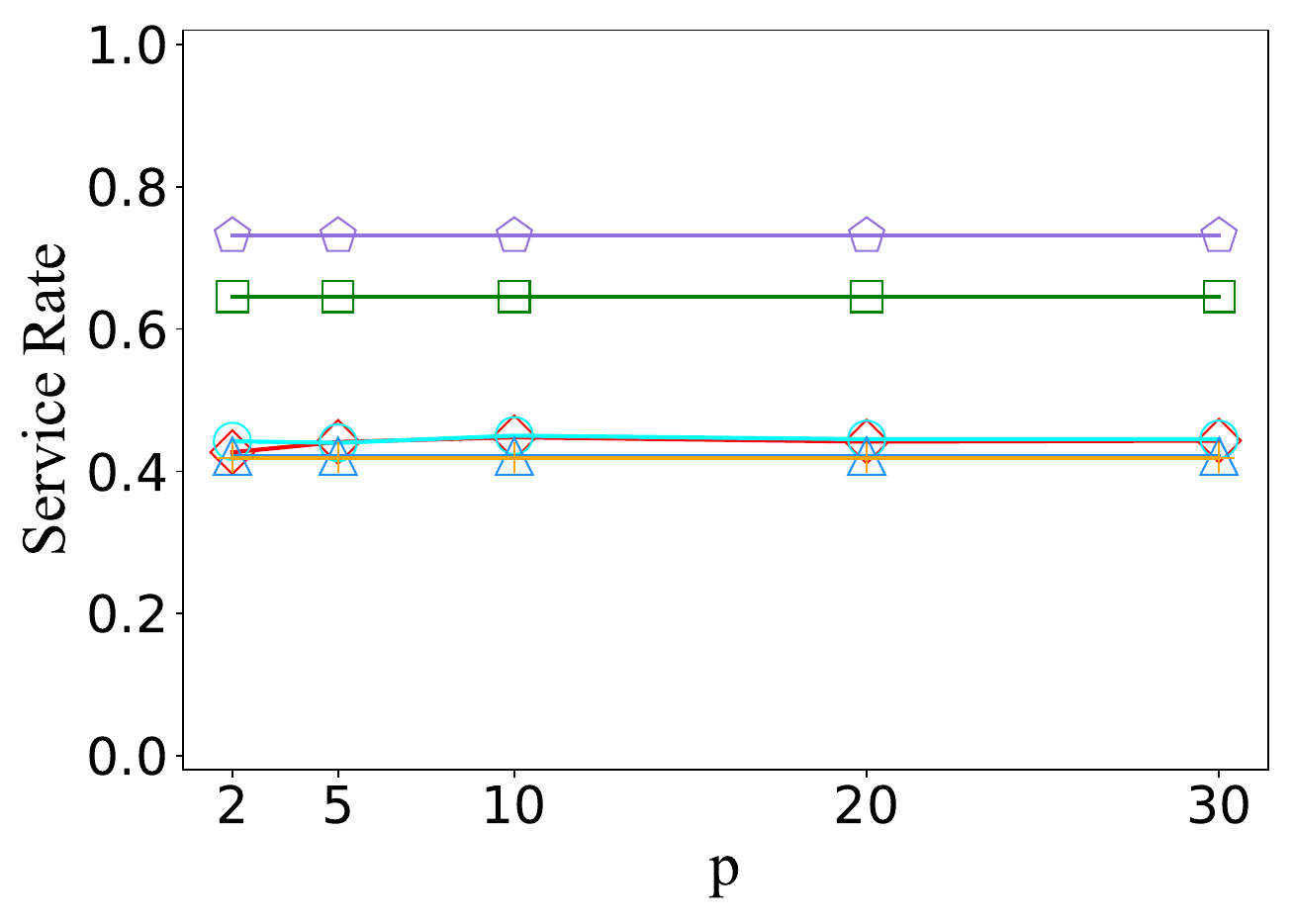}}}
				\label{subfig:sr_varing_pen_cd}}\hspace{-2ex}
			\subfigure[][{\scriptsize Service Rate (\textit{NYC})}]{
				\raisebox{-1ex}{\scalebox{0.19}[0.17]{\includegraphics{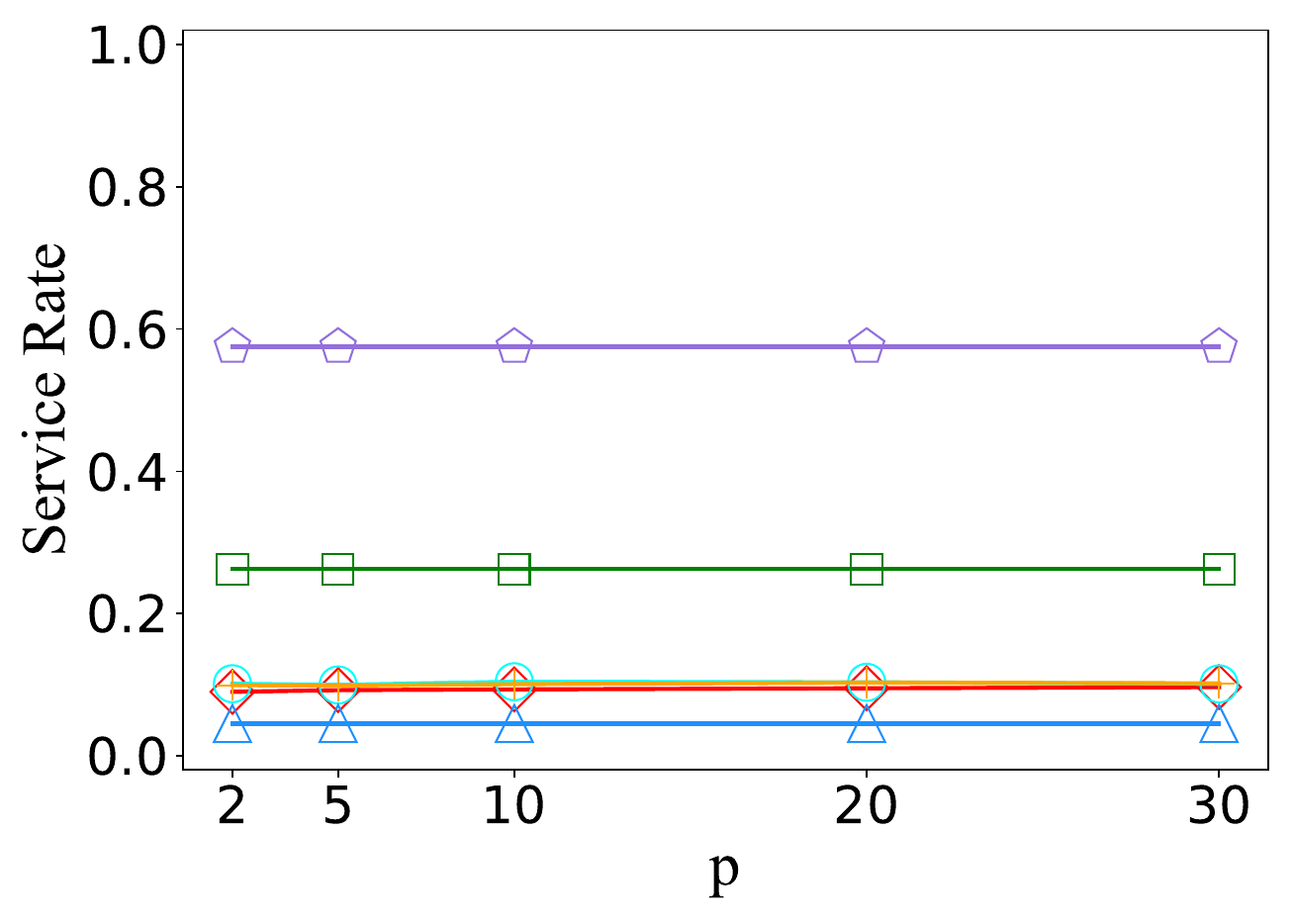}}}
				\label{subfig:sr_varing_pen_nyc}}\\[-2ex]
			\subfigure[][{\scriptsize Running Time (\textit{CHD})}]{
				\raisebox{-1ex}{\scalebox{0.19}[0.17]{\includegraphics{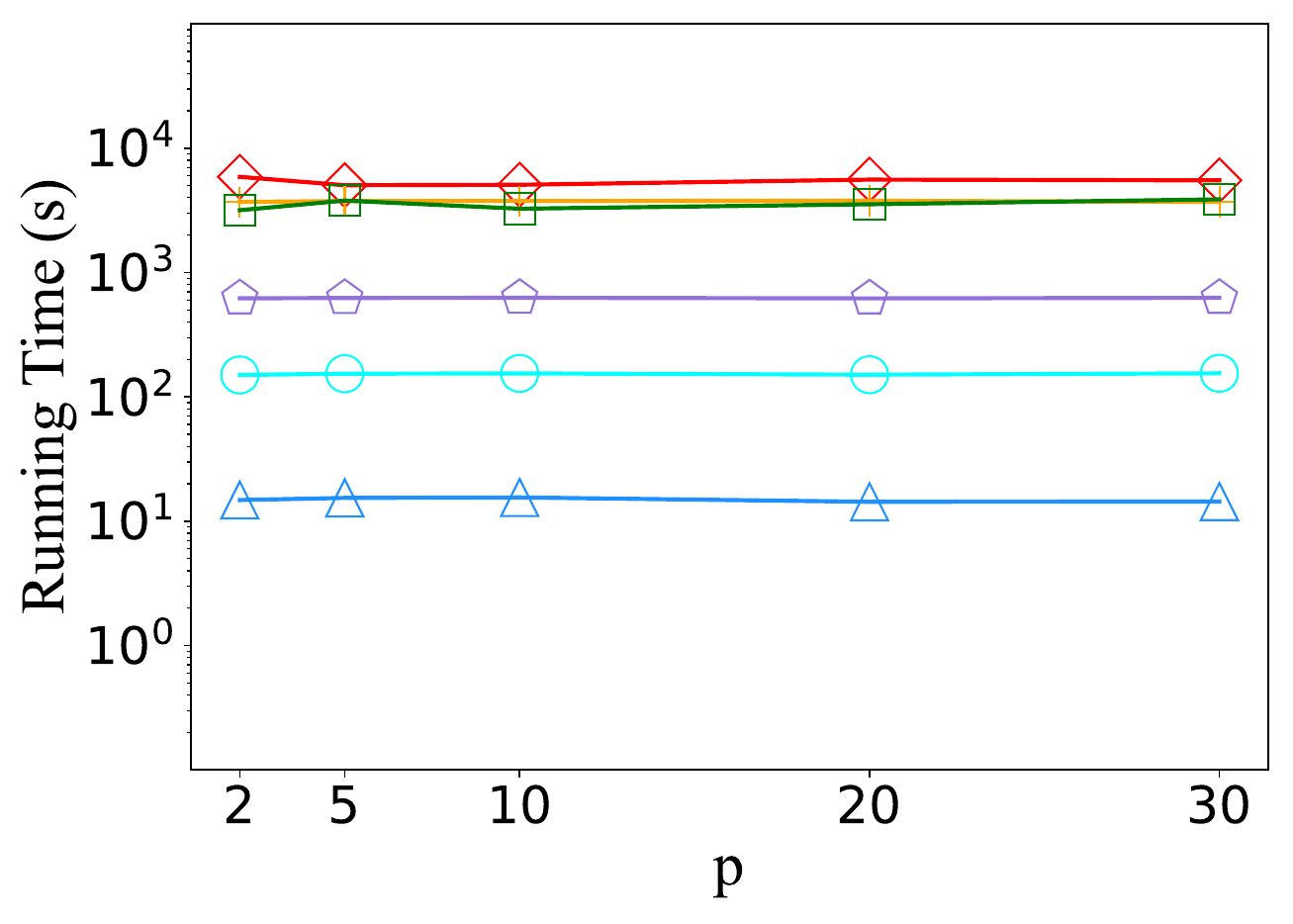}}}
				\label{subfig:tm_varing_pen_cd}}\hspace{-2ex}
			\subfigure[][{\scriptsize Running Time (\textit{NYC})}]{
				\raisebox{-1ex}{\scalebox{0.19}[0.17]{\includegraphics{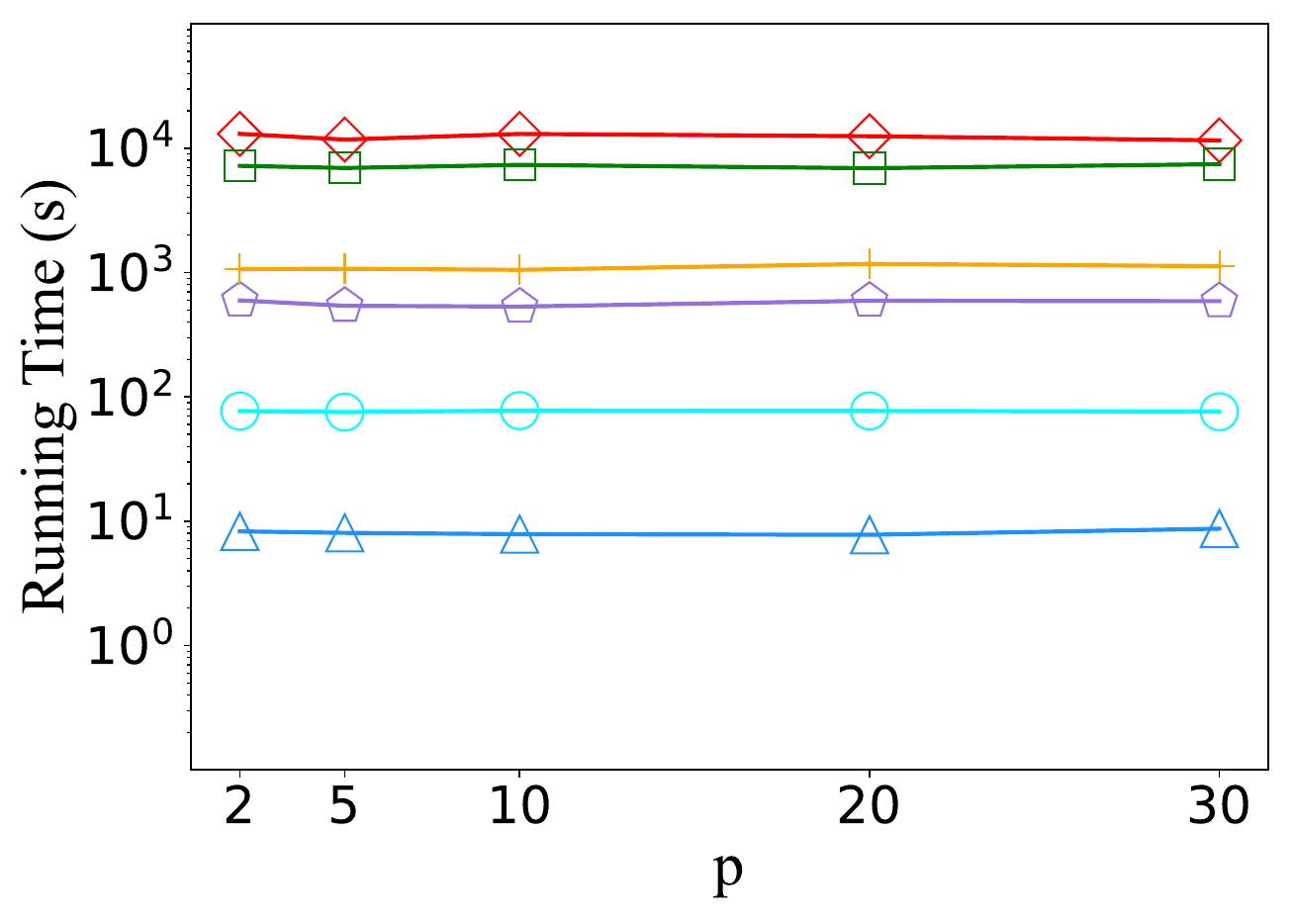}}}
				\label{subfig:tm_varing_pen_nyc}}
			\caption{ Performance of Varying $p_r$.}
			\label{fig:vary_pen}
		\end{minipage} &
		\begin{minipage}[t]{.5\textwidth}
			\subfigure{
				\scalebox{0.5}[0.5]{\includegraphics{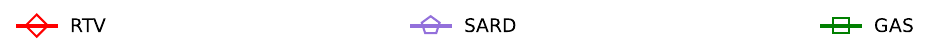}}}\hfill
			\addtocounter{subfigure}{-1}\\[-3ex]
			\subfigure[][{\scriptsize Unified Cost (\textit{CHD})}]{
				\raisebox{-1ex}{\scalebox{0.19}[0.17]{\includegraphics{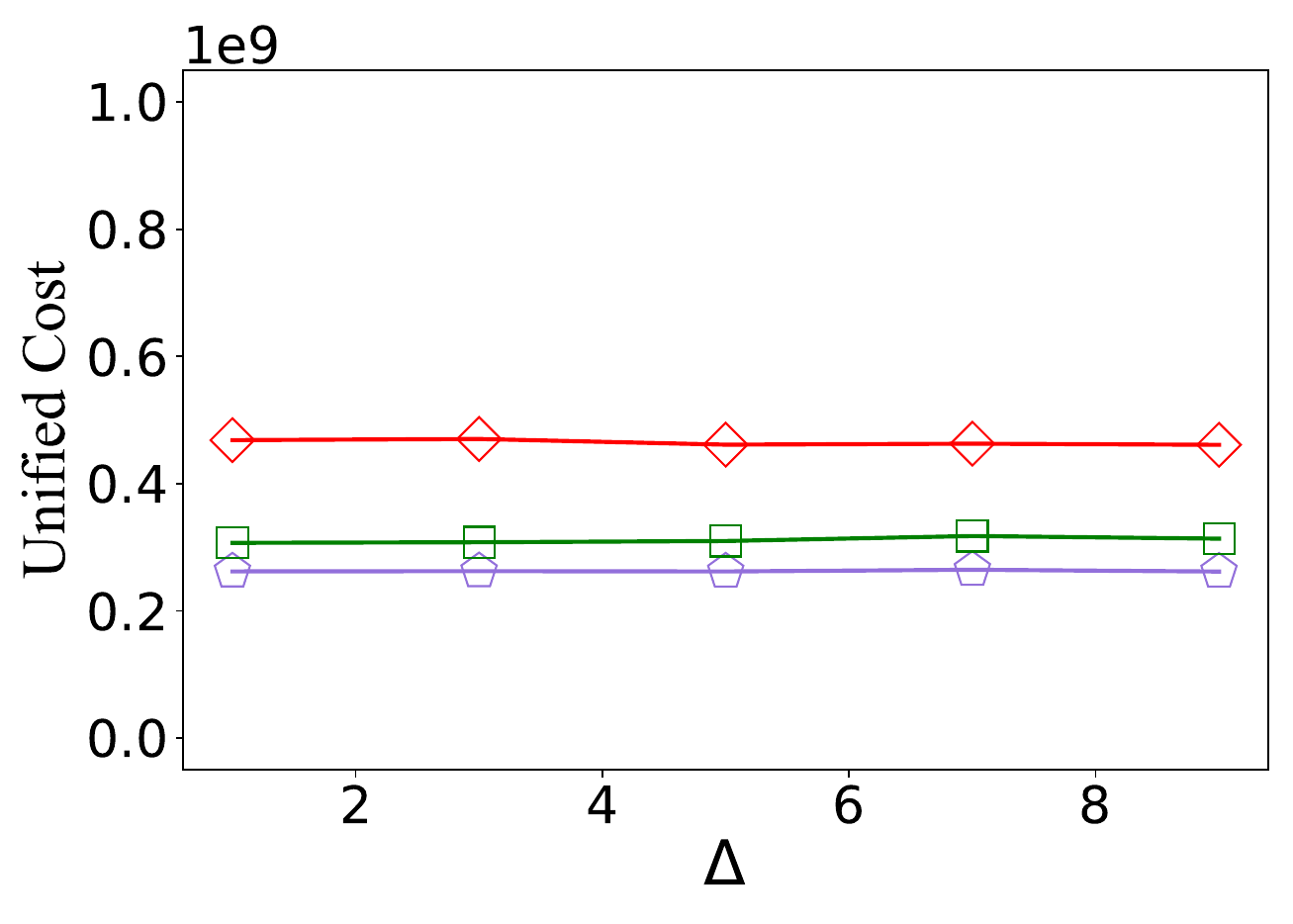}}}
				\label{subfig:uc_varing_batch_cd}}\hspace{-2ex}
			\subfigure[][{\scriptsize Unified Cost (\textit{NYC})}]{
				\raisebox{-1ex}{\scalebox{0.19}[0.17]{\includegraphics{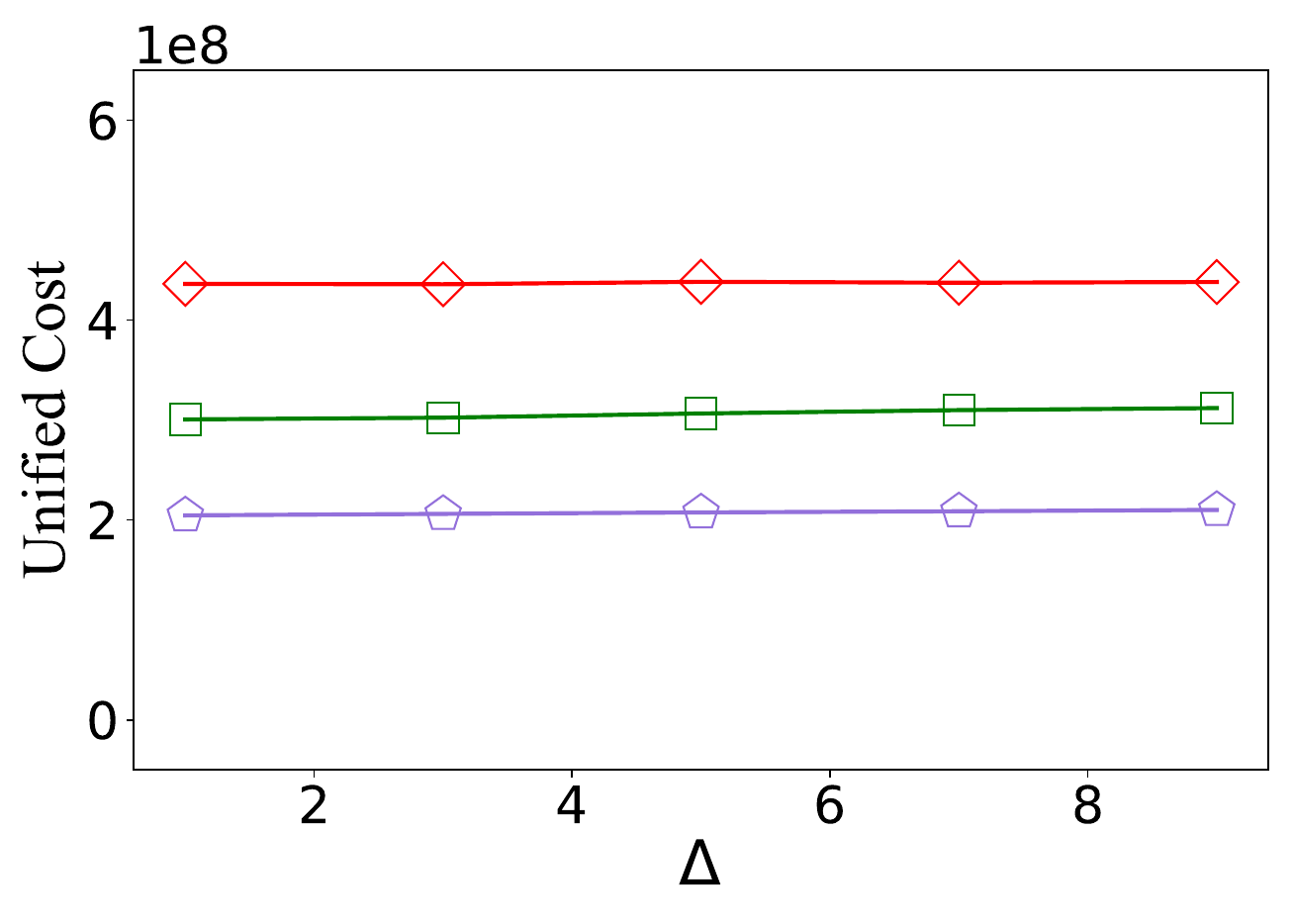}}}
				\label{subfig:uc_varing_batch_nyc}}\\[-2ex]
			\subfigure[][{\scriptsize Service Rate (\textit{CHD})}]{
				\raisebox{-1ex}{\scalebox{0.19}[0.17]{\includegraphics{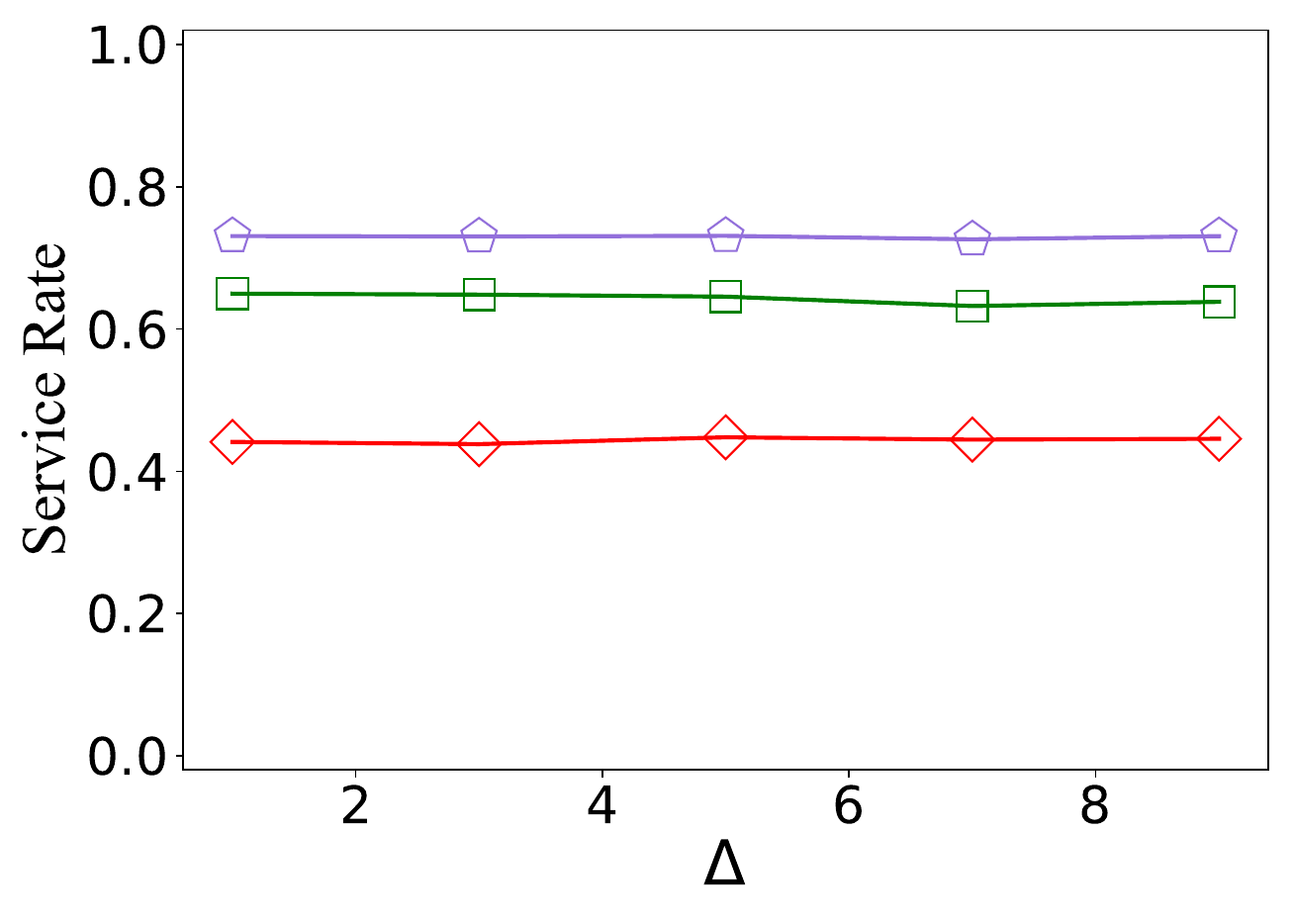}}}
				\label{subfig:sr_varing_batch_cd}}\hspace{-2ex}
			\subfigure[][{\scriptsize Service Rate (\textit{NYC})}]{
				\raisebox{-1ex}{\scalebox{0.19}[0.17]{\includegraphics{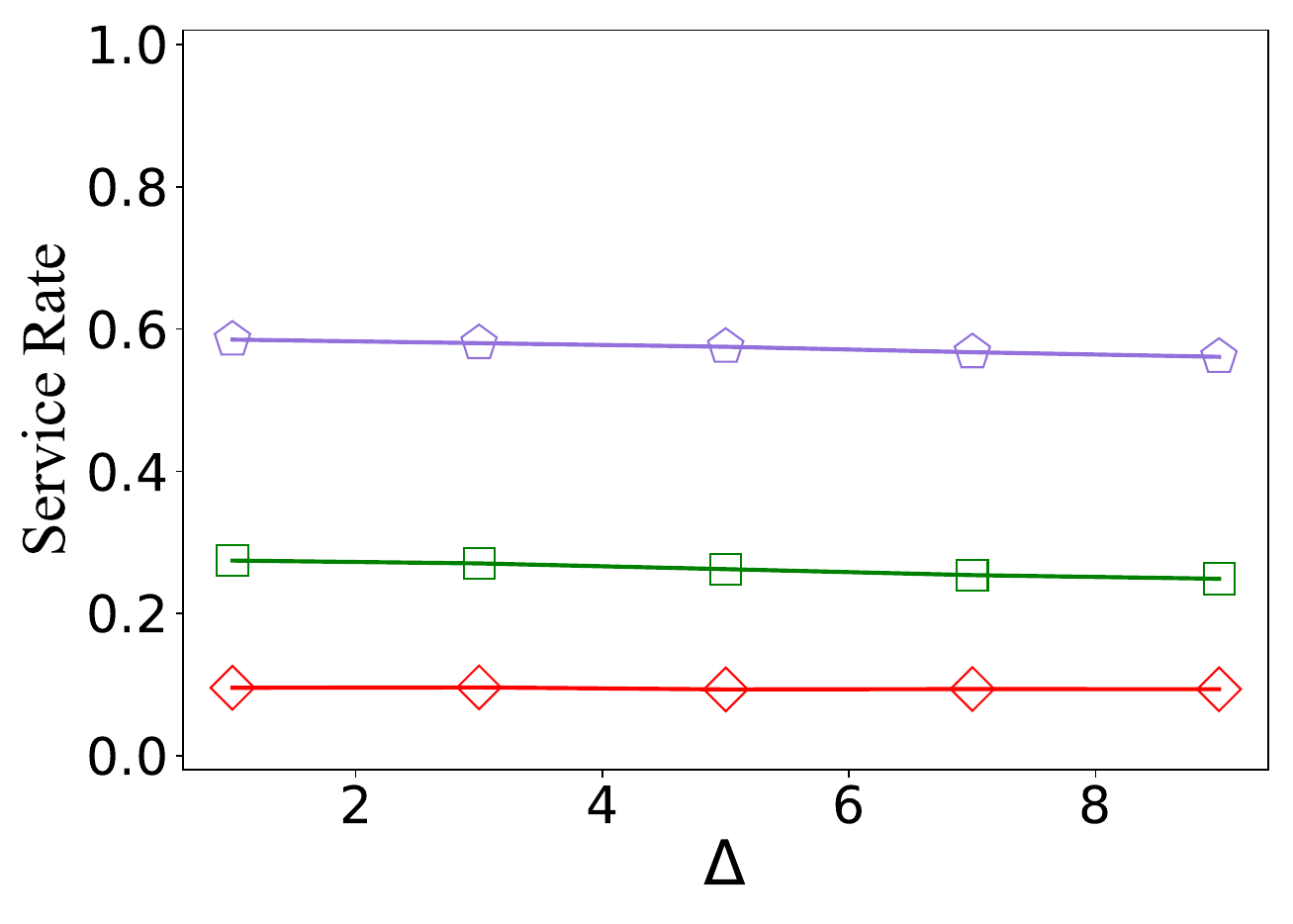}}}
				\label{subfig:sr_varing_batch_nyc}}\\[-2ex]
			\subfigure[][{\scriptsize Running Time (\textit{CHD})}]{
				\raisebox{-1ex}{\scalebox{0.19}[0.17]{\includegraphics{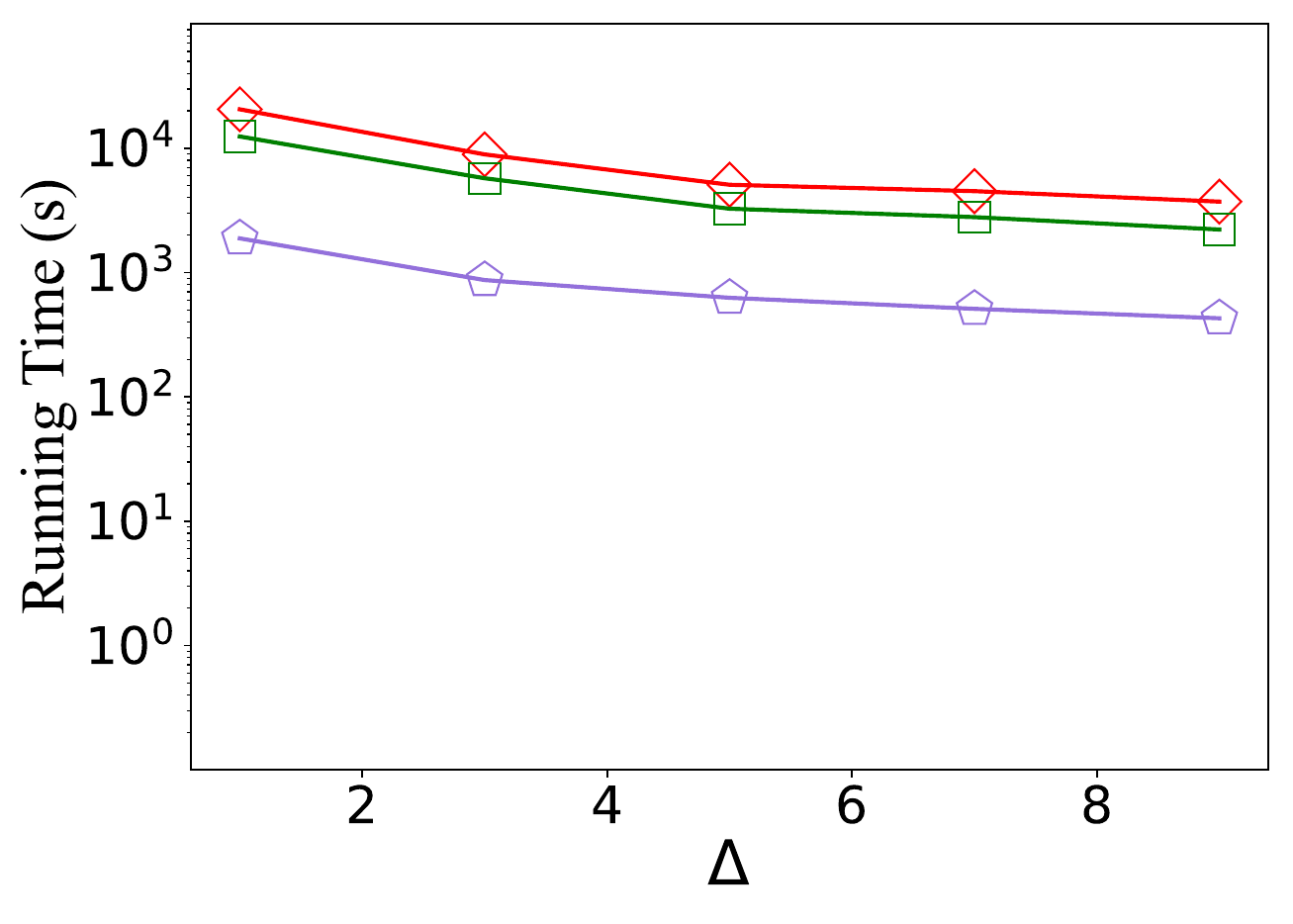}}}
				\label{subfig:tm_varing_batch_cd}}\hspace{-2ex}
			\subfigure[][{\scriptsize Running Time (\textit{NYC})}]{
				\raisebox{-1ex}{\scalebox{0.19}[0.17]{\includegraphics{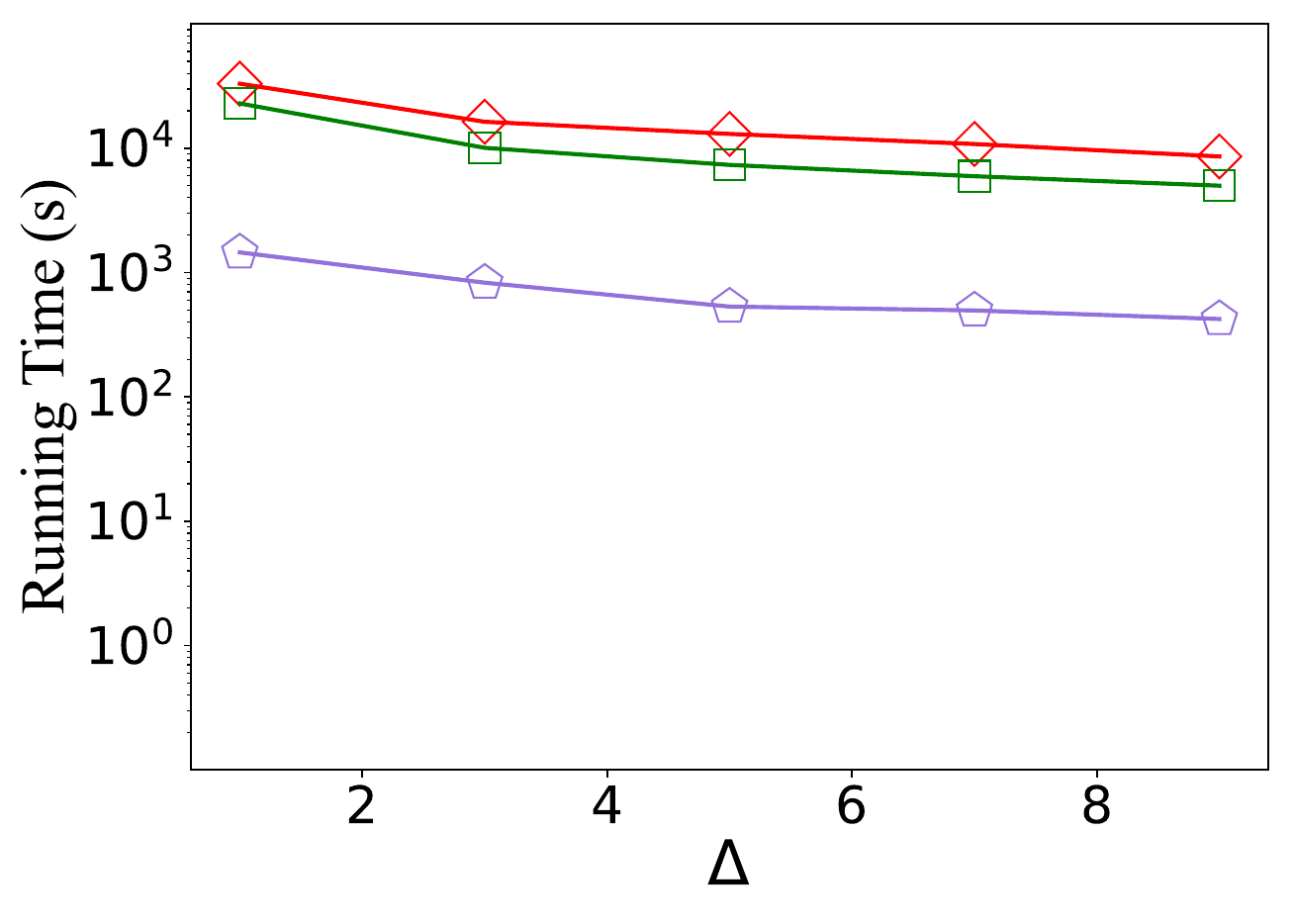}}}
				\label{subfig:tm_varing_batch_nyc}}
			\caption{ Performance of Varying $\Delta$.}
			\label{fig:vary_batch}
		\end{minipage}
	\end{tabularx}
\end{figure*}

\noindent\textbf{Effect of the number of requests.}
Figure~\ref{fig:vary_req} shows that as requests increase from 10K to 250K, the unified costs of all algorithms grow.
Meanwhile, {SARD} and {GAS} achieve smaller unified costs than other methods when the number of requests becomes larger.
For service rate, {SARD} improves the service rate by $41.97\%\sim52.99\%$ at most compared with {pruneGDP}.
Besides, compared with the state-of-art batch-based method {GAS}, {SARD} achieves up to $12.09\%$ and $31.27\%$ higher service rates on the two datasets, respectively.
\revision{{DARM+DPRS} demonstrates the advantages of predictive scheduling when the request volume is very small, but fails to handle a larger number of requests.
	All traditional algorithms initially exhibit a decline in service rate due to insufficient vehicles as requests increase. 
	Subsequently, the service rate grows as the increased number of requests improves the probability of sharing rides with other requests, then more requests can be served.}
In terms of running time, insertion-based methods are faster. 
Among batch-based methods (SARD, RTV and GAS), {SARD} is $4.21\times\sim 38.6\times$ faster than {GAS} and {RTV}.

\noindent\textbf{Effect of deadline.}
Figure~\ref{fig:vary_ddl} presents the impact of varying request deadlines by changing the deadline parameter $\gamma$ from $1.2$ to $2.0$.
With a strict deadline of $\gamma=1.2$, {SARD}'s service rate is similar to existing algorithms. 
This is due to a significant reduction in the number of candidate vehicles for each request, making it challenging to achieve noticeable performance improvements through grouping strategies in batch mode.
However, as the deadline increases, the superiority of {SARD} and {GAS} gradually becomes apparent. 
\revision{TicketAssign+ increases contention with higher deadlines, boosting service rate while increasing the gap in runtime compared to pruneGDP on the \textit{NYC} dataset}
When the deadline is $1.8\times$, {SARD}'s service rate exceeds $90\%$, up to $37.42\%$ higher than that of other algorithms on the \textit{CHD} dataset.
The superiority of {SARD} is more evident in the \textit{NYC} dataset, where its service rate is up to $83.98\%$ higher than that of {pruneGDP}.
{GAS} operates less efficiently on the \textit{NYC} dataset with $\gamma=2.0$, mainly because there are more request combinations as the deadline is relaxed. 
Additionally, {GAS} considers all combinations of requests and schedules for almost every vehicle. 
In {SARD}, the requests are proposed to different vehicles more decentralized in each round and only go to the next round of enumeration after being discarded.
Thus, the time cost for the combination enumeration of each vehicle is significantly reduced in SARD, benefiting from our ``proposal-acceptance'' execution strategy. 
{RTV} results for $\gamma\geq 1.8$ on the \textit{NYC} dataset are not presented because the constraints of {RTV-Graph} exceed the limit of \textit{glpk}~\cite{GLPK}. 
{SARD} achieves the best unified cost, saving up to $48.03\%$ and $73.16\%$ compared to others on two datasets.
For the running time, SARD is $1.83\times$ to $72.68\times$ faster than {RTV} and {GAS}.

\noindent\textbf{Effect of vehicle's capacity constraint.}
Figure~\ref{fig:vary_cap} illustrates the results of varying vehicle capacity from 2 to 6.
{SARD} and {GAS} save at least $30.87\%$ in unified cost compared to other algorithms.
For service rate, {SARD} is the best of all tested algorithms ($7.58\%\sim 53.39\%$ higher service rate than other tested algorithms).
As for running time, {SARD} is also the fastest among batch methods (RTV, GAS, SARD), being $5.17\times\sim 17.52\times$ faster than {RTV} and {GAS}.

\noindent\textbf{Effect of penalty.}
Figure~\ref{fig:vary_pen} represents the effect of varying the penalty coefficient from 2 to 30.
Most methods' service rates are unaffected by the penalty change.
Since {pruneGDP}, \revision{{TicketAssign+}, {DARM+DPRS},} {GAS} and {SARD} take distance, group profit, and shareability loss as indicators in greedy strategies in the assignment phase, the penalty coefficient only affect their scores in the unified cost. 
However, {RTV} utilized the penalty coefficient in its linear programming (LP) constraint matrix, affecting LP results only when the penalty is small. 
For large penalties, the effect is negligible.
Unified cost in all algorithms is proportional to the penalty coefficient.
{SARD} leads in both datasets, with a service rate increase of $8.55\%\sim 52.99\%$.
Similarly, since the penalty coefficient does not affect the assignment phase, the execution time remains stable across datasets.

\noindent\textbf{Effect of the batching period.}
Figure~\ref{fig:vary_batch} represents the results of varying the batching period from 1 to 9 seconds for batch-based methods.
The varying batch time affects the batch-based algorithms in terms of unified cost and service rate.
{SARD} performs the best with the varying the batching period, which saves up to $44.18\%$ and $53.08\%$ in unified cost and increases the service rate by up to $29.2\%$ and $48.96\%$ in two datasets, respectively. 
For the running time, since the number of execution rounds will decrease when the period of each batch increases, the running times of these methods decrease. 
{SARD} is $5.17\times\sim 24.47\times$ faster than {RTV} and {GAS}.

\noindent\textbf{Discussion.}
(1) \textbf{Running Time}. Online methods pruneGDP and \revision{TicketAssign+} have a complexity of $O(n^2)$, making them the fastest in the experiments. 
\revision{The running time of the reinforcement learning method DARM+DPRS hinges primarily on the size of its scheduled state space and the complexity of its predictive model structure.}
RTV, GAS, and our SARD are batch-based. 
The time complexity of RTV is $O((m|G|)^c)$, where $m$ is the number of vehicles, $|G|$ is the number of possible combinations of $n$ tasks in the batch, notated as $n^c$. 
GAS's complexity is $O(T_s+mT_c)$, with $T_s$ for building an additive index (as $O(n^c)$), and $T_c$ for greedily searching the additive index (as O($n^c$)). 
Our SARD needs $O(cn^2+mn^c)$, where $n$ is the numbers of requests. 
Simplified, their complexities are $O(m^c n^2c)$, $O(mn^c)$, and $O(mn^c)$, respectively. 
GAS and SARD outperform RTV. 
Unlike GAS, our SARD enumerates combinations only among proposed requests per vehicle, resulting in smaller $n$ than GAS. 
Despite multiple propose-acceptance rounds, SARD's small $n$ makes it much faster than GAS.
(2) \textbf{DataSet}. In the \textit{NYC} dataset, requests per unit time are about double those of \textit{CHD}. 
Hence, in the same batch time, algorithms using combination enumeration (i.e., SARD and GAS) can get more candidate schedules and perform better in \textit{NYC}. 
In addition, \textit{NYC}'s more compact road network (with only half the nodes) and concentrated requests increase sharing opportunities.
Thus, SARD performs better in \textit{NYC} than \textit{CHD} in most experiments.
\revision{DARM+DPRS faces a greater challenge on the \textit{CHD} dataset due to its larger state space.}
(3) \textbf{Scalability}. For SARD's scalability, multi-threading can speed up the Shareability Graph building and acceptance stage as each vehicle decides independently. 
In practice, SARD can be implemented with streaming distribution computation engines (e.g., Apache Flink \cite{flink}) and apply a tumbling window strategy to the distribution environment.

\noindent\textbf{Summary of the experimental study:} 
\begin{itemize}[leftmargin=*]
	\item The batch-based methods (i.e., RTV, GAS and SARD) can get less unified costs and higher service rates than the online-based methods (i.e., pruneGDP, \revision{DARM+DPRS and TicketAssign+}).
	For instance, the service rate of {SARD} can be $50\%$ higher than that of other tested methods.
	\item {SARD} runs up to 72.68 times faster than the other batch-based methods.
	For example, in Figure \ref{subfig:tm_varing_ddl_nyc}, {SARD} handles \textit{NYC} requests in $8$ minutes, while {GAS} takes $9$ hours.
\end{itemize}

\section{Related Work}
\label{sec:related}
Ridesharing research dates back to the dial-a-ride problem~\cite{darp_03}. 
Ridesharing services plan routes for each vehicle to fulfill requests while optimizing various objectives like minimizing travel time~\cite{pnas,T-share,Kinetic_tree}, maximizing served requests~\cite{Online-ridesharing,flexible-realtime}, or both~\cite{luo2020}.
Liu \textit{et al.}~\cite{liu2020} proposed probabilistic routing for taxis to encounter suitable offline passengers.
Existing works are either online-based~\cite{TongZZCYX18,2019Last,insertion,cheng2017utility,demand_aware} or batch-based~\cite{simple_better,zheng2019auction,zheng2018order,pnas} depending on whether the requests are known in advance.

The \textit{insertion} operation is a commonly used core operation~\cite{TongZZCYX18, 2019Last, insertion, demand_aware, Kinetic_tree} in online-based methods. 
It inserts the origin-destination pair of a new request into the current route.
Chen \textit{et al.}~\cite{chen2018} utilized grid indexing to filter vehicles and applied deadline constraints pruning strategy to reduce shortest path distance calculations when inserting requests into the Kinetic Tree~\cite{Kinetic_tree}.
Tong \textit{et al.}~\cite{TongZZCYX18} minimized travel time with a dynamic programming method and proposed a linear-time insertion, applied in~\cite{2019Last}.
Xu \textit{et al.}~\cite{insertion} sped up insertion to linear time by leveraging a segment-based DP algorithm and Fenwick tree.
Wang \textit{et al.}~\cite{demand_aware} enhanced the results by considering insertion effects and demand prediction.

In batch-based methods, requests are grouped by similarity before assignment.
Zheng \textit{et al.}~\cite{zheng2018order} used bipartite matching to group similar riders and drivers.
Zeng \textit{et al.}~\cite{simple_better} developed an index called the \textit{additive tree} to simplify the process and used randomization techniques to improve the approximation ratio.
Alonso-Mora \textit{et al.}~\cite{pnas} explored cliques in the pairwise shareability graph and incrementally computed the optimal assignment using an integer linear program (ILP).
However, none of these methods prioritize batches by taking advantage of the shareable relationships between requests.

\section{Conclusion}
\label{sec:conclusion}
This paper studies the dynamic ridesharing problem with a batch-based processing strategy.
We introduce a graph-based shareability graph to reveal the sharing relationship between requests in each batch. 
With the structure information of the shareability graph, we measure the shareability loss of each request group.
We also propose a heuristic algorithm, \textit{SARD}, with a two-phase ``proposal-acceptance'' strategy. 
Experiments show that our method provides a better service rate, lower cost, and shorter running time compared to state-of-the-art batch-based methods~\cite{simple_better,pnas} on real datasets.

\section*{Acknowledgment}
Peng Cheng's work is partially supported by the National Natural Science Foundation of China under Grant No. 62102149. Lei Chen’s work is partially supported by National Key Research and Development Program of China Grant No. 2023YFF0725100, National Science Foundation of China (NSFC) under Grant No. U22B2060, the Hong Kong RGC GRF Project 16213620, RIF Project R6020-19, AOE Project AoE/E-603/18, Theme-based project TRS T41-603/20R, CRF Project C2004-21G, Guangdong Province Science and Technology Plan Project 2023A0505030011, Hong Kong ITC ITF grants MHX/078/21 and PRP/004/22FX, Zhujiang scholar program 2021JC02X170, Microsoft Research Asia Collaborative Research Grant, HKUST-Webank joint research lab and HKUST(GZ)-Chuanglin Graph Data Joint Lab.
Xuemin Lin is supported by NSFC U2241211 and U20B2046. Corresponding Author: Peng Cheng.

\balance

\bibliographystyle{ieeetr}

\bibliography{add}

\appendix

\subsection{Memory consumption}
Figure~\ref{fig:memory_usage} shows the memory usage of tested traditional algorithms under the default parameters. 
The online mode algorithms follow the first-come-first-serve mode and use less memory, while parallelization brings additional overhead of maintaining a lock for each worker.
In contrast, the batch mode algorithms need additional storage to store the combinations of requests in each batch (e.g., RTV-Graph for RTV, the additive index for GAS, shareability graph for SARD). 
Since RTV rely on an integer linear program, the memory usage in RTV is more than twice that in GAS and SARD. 
Moreover, SARD and GAS have similar costs for storing undirected, unweighted shareability graphs.

\begin{figure}[h!]\centering
	\subfigure{
		\scalebox{0.65}[0.65]{\includegraphics{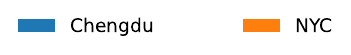}}}\hfill
	\addtocounter{subfigure}{-1}\\[-2ex]
	\subfigure{
		\scalebox{0.32}[0.32]{\includegraphics{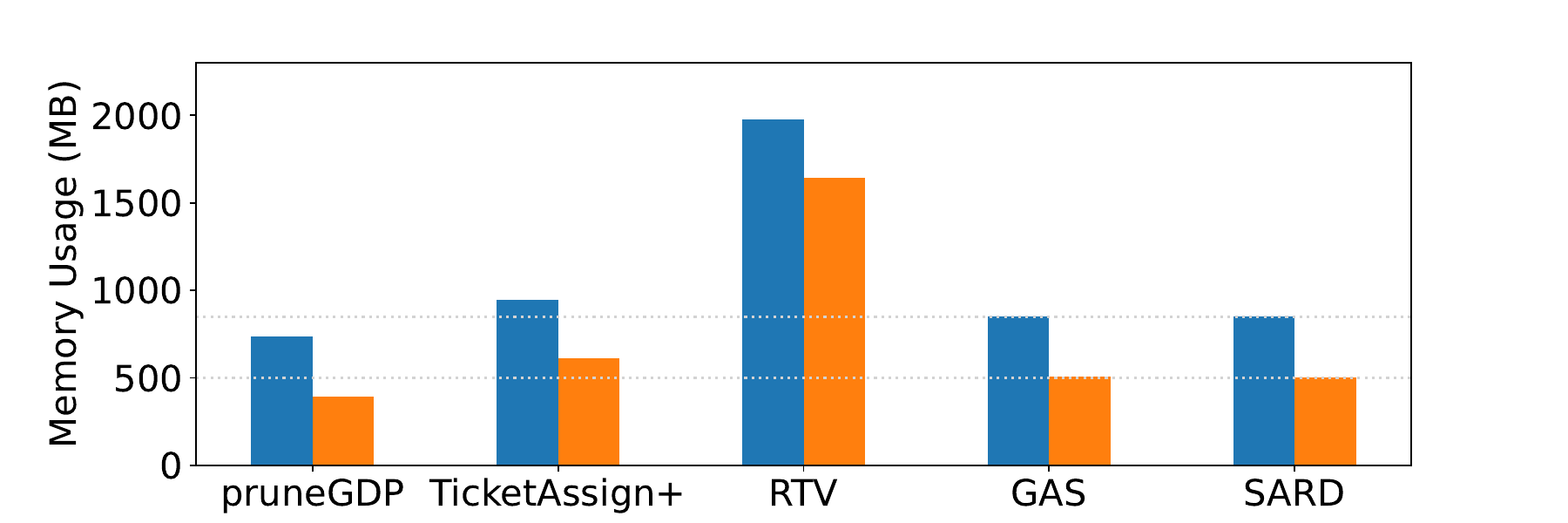}}}
	\caption{ Memory Consumption.}
	\label{fig:memory_usage}
\end{figure}

\subsection{Experimental Study on \textit{Cainiao} Dataset}
We conduct comprehensive experiments using the Cainiao~\cite{cainiao} dataset, a real-world dataset from a prominent delivery platform in Shanghai, China, published by LaDe~\cite{cainiaodataset}. 
This dataset encompasses millions of packages and provides detailed daily information on tasks and workers from the Cainiao platform throughout 2021. 
Unlike taxi datasets, the delivery dataset exhibits a more dispersed request distribution and features more generous deadlines, allowing for greater routing flexibility. 
Due to insufficient training data, we only report the results of traditional algorithms here.
Table~\ref{tab:cainiao_settings} presents the parameter settings, with default values in bold.

\begin{table}[h!]
	\begin{center}
		{\scriptsize 
			\caption{ Experimental Settings.} 
			\label{tab:cainiao_settings}
			\begin{tabular}{l|l}
				{\bf \qquad \qquad \quad Parameters} & {\bf \qquad \qquad \qquad Values} \\ \hline \hline
				the number, $n$, of requests  & 50K, 75K, \textbf{100K}, 125K, 150K\\
				the number, $m$, of vehicles  & 3K, 3.5K, 4K, 4.5K, \textbf{5K} \\
				the capacity of vehicles $c$ & 2, 3, \textbf{4}, 5, 6\\
				the deadline parameter $\gamma$ & 1.8, 1.9, \textbf{2.0}, 2.1, 2.2 \\
				the penalty coefficient $p_r$ & 2, 5, \textbf{10}, 20, 30 \\
				the batching time $\Delta$ (s) & 1, 3, \textbf{5}, 7, 9 \\
				the variance $\sigma$ &\textbf{0.0}, 0.5, 1.0, 1.5, 2.0 \\
				\hline
			\end{tabular}
		}
	\end{center}
\end{table}

\begin{figure*}[t!]\centering
	\subfigure{
		\scalebox{0.5}[0.5]{\includegraphics{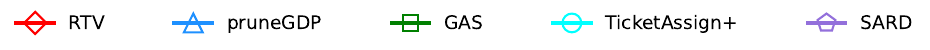}}
	}\hfill\\\vspace{-2ex}
	\addtocounter{subfigure}{-1}
	\subfigure[][{\scriptsize Unified Cost (\textit{Cainiao})}]{
		\scalebox{0.15}[0.14]{\includegraphics{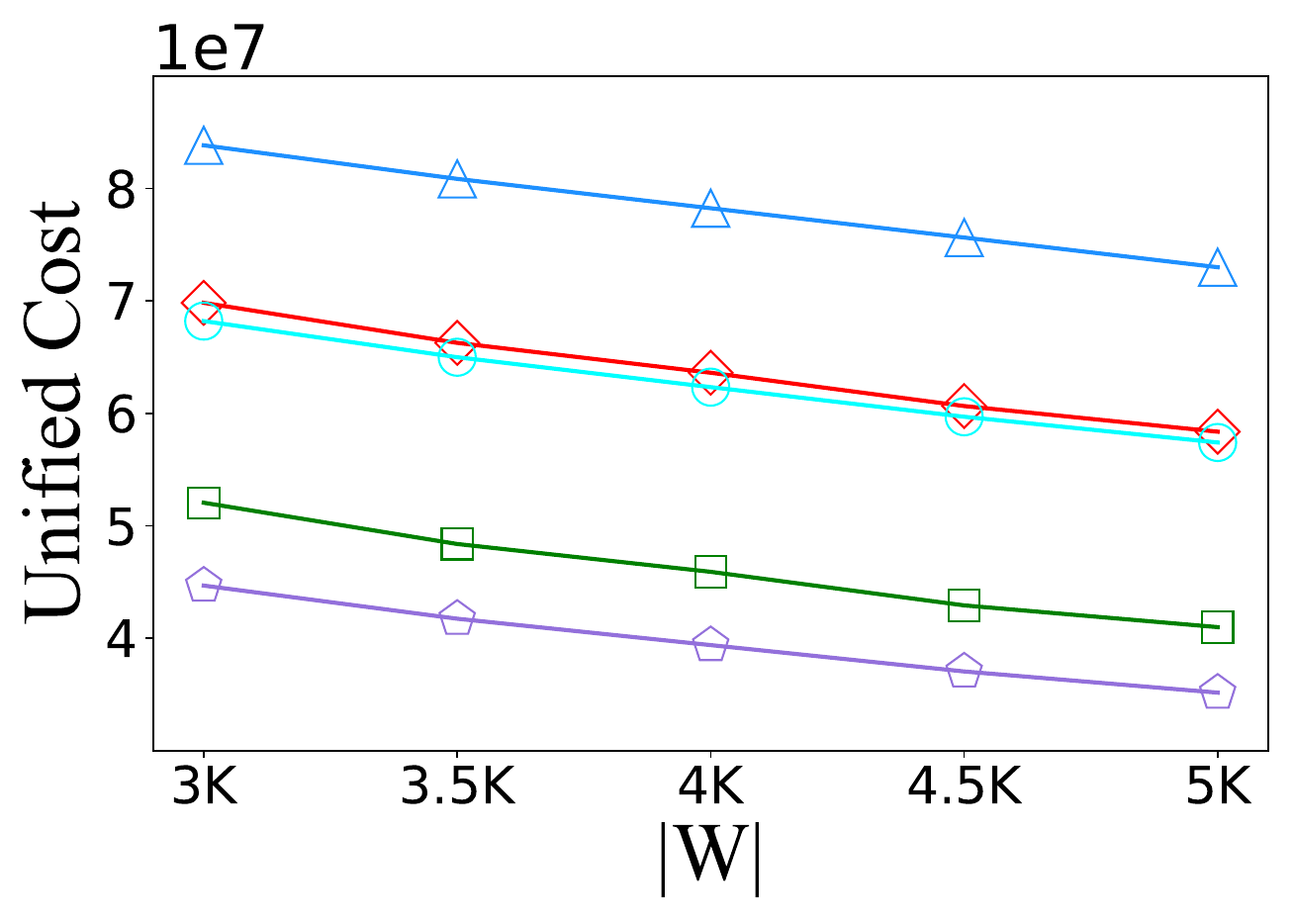}}
		\label{subfig:uc_varing_worker_sh}
	}\hspace{-2ex}
	\subfigure[][{\scriptsize Unified Cost (\textit{Cainiao})}]{
		\scalebox{0.16}[0.14]{\includegraphics{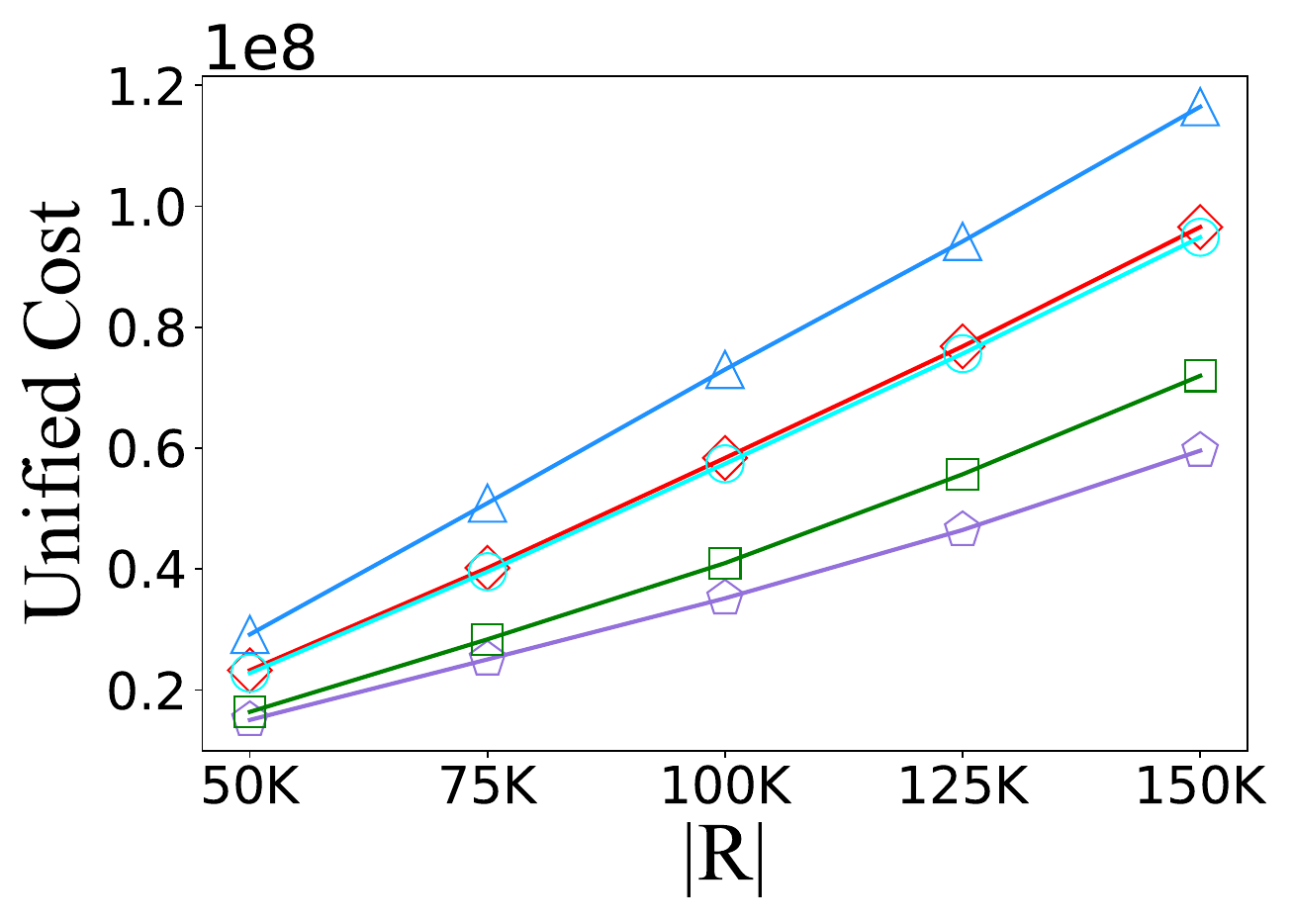}}
		\label{subfig:uc_varing_req_sh}
	}\hspace{-2ex}
	\subfigure[][{\scriptsize Unified Cost (\textit{Cainiao})}]{
		\scalebox{0.16}[0.14]{\includegraphics{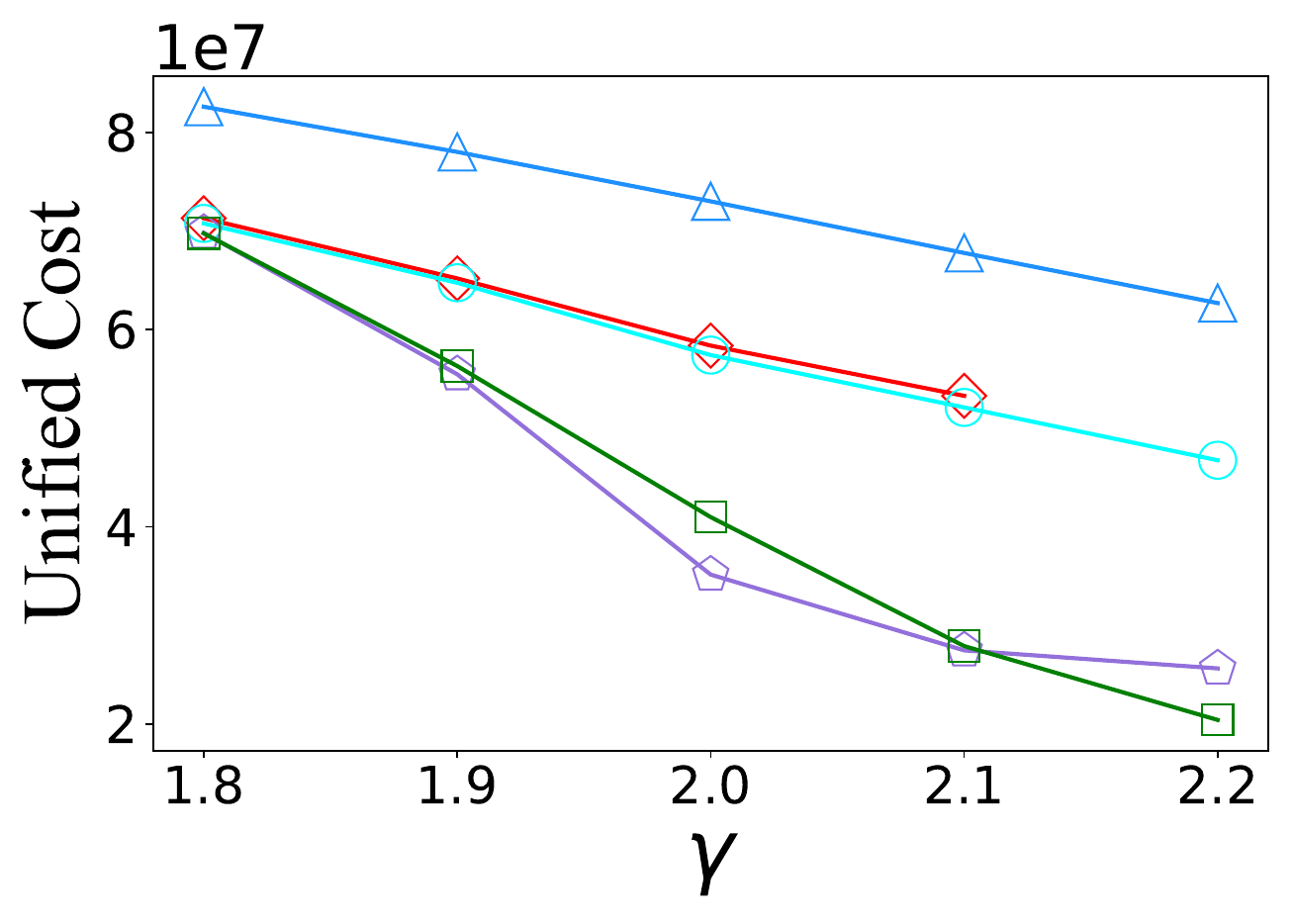}}
		\label{subfig:uc_varing_ddl_sh}
	}\hspace{-2ex}
	\subfigure[][{\scriptsize Unified Cost (\textit{Cainiao})}]{
		\scalebox{0.16}[0.14]{\includegraphics{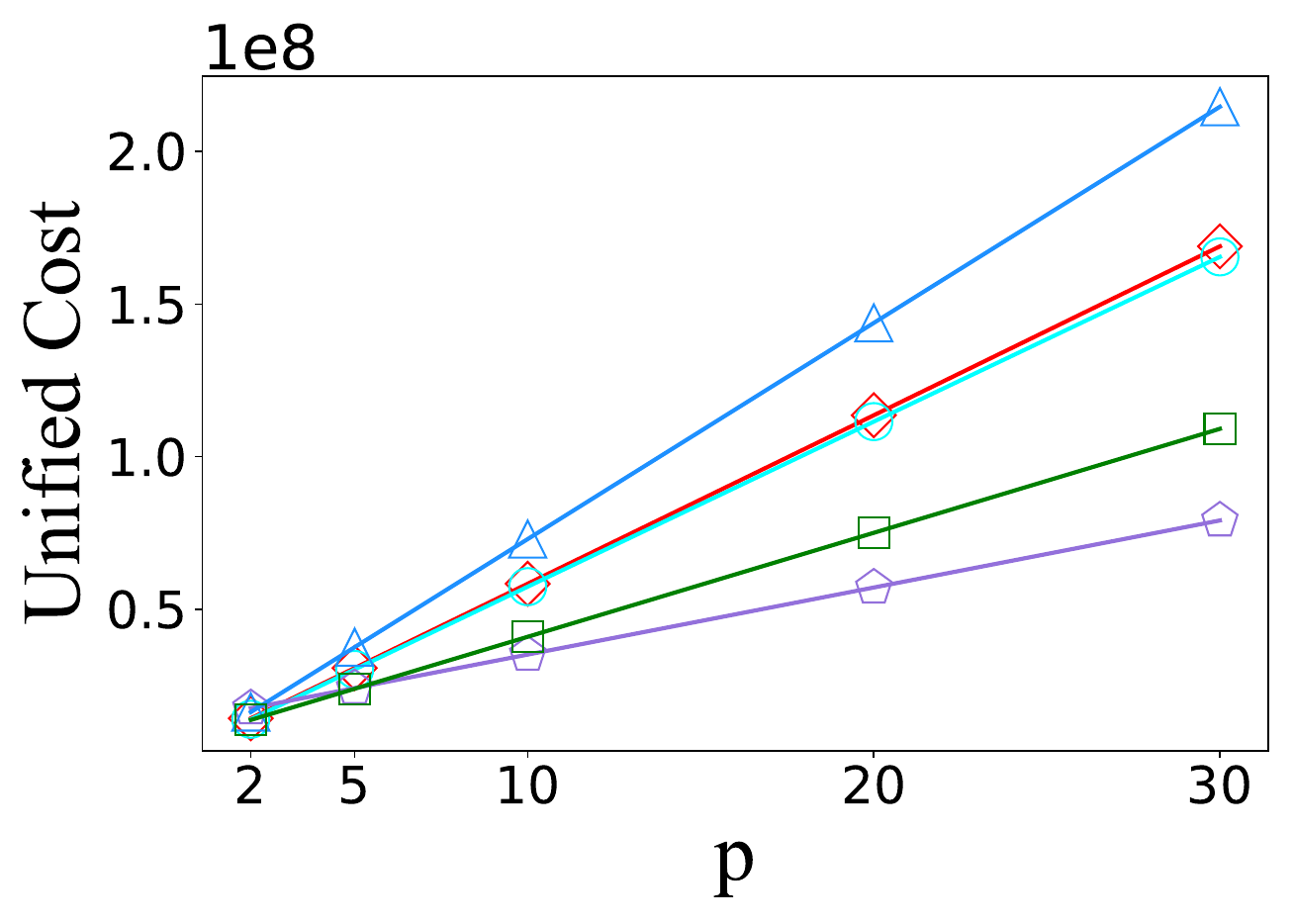}}
		\label{subfig:uc_varing_pen_sh}
	}\hspace{-2ex}
	\subfigure[][{\scriptsize Unified Cost (\textit{Cainiao})}]{
		\scalebox{0.16}[0.14]{\includegraphics{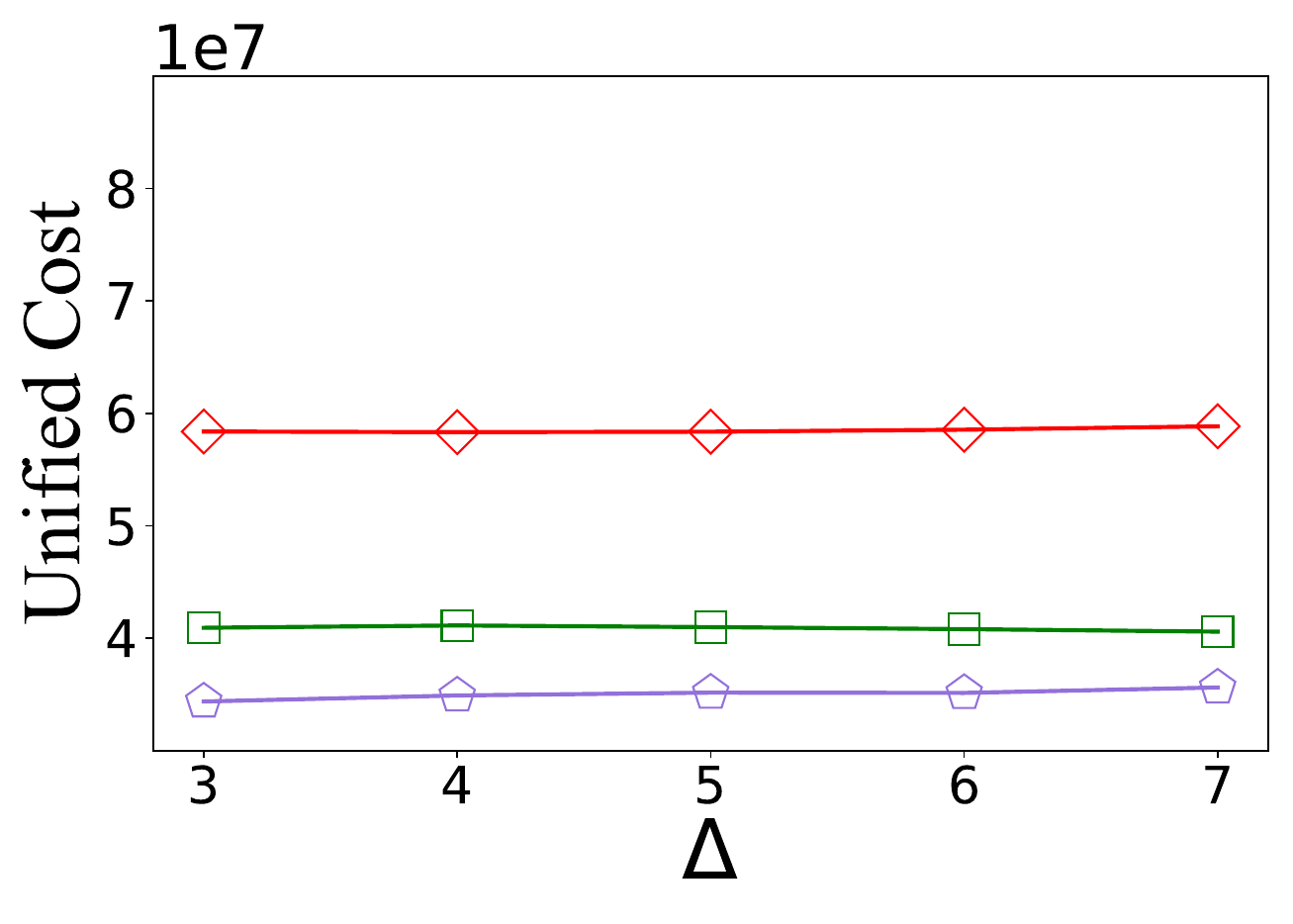}}
		\label{subfig:uc_varing_batch_sh}
	}\\\vspace{-2ex}
	
	\subfigure[][{\scriptsize Service Rate (\textit{Cainiao})}]{
		\scalebox{0.158}[0.14]{\includegraphics{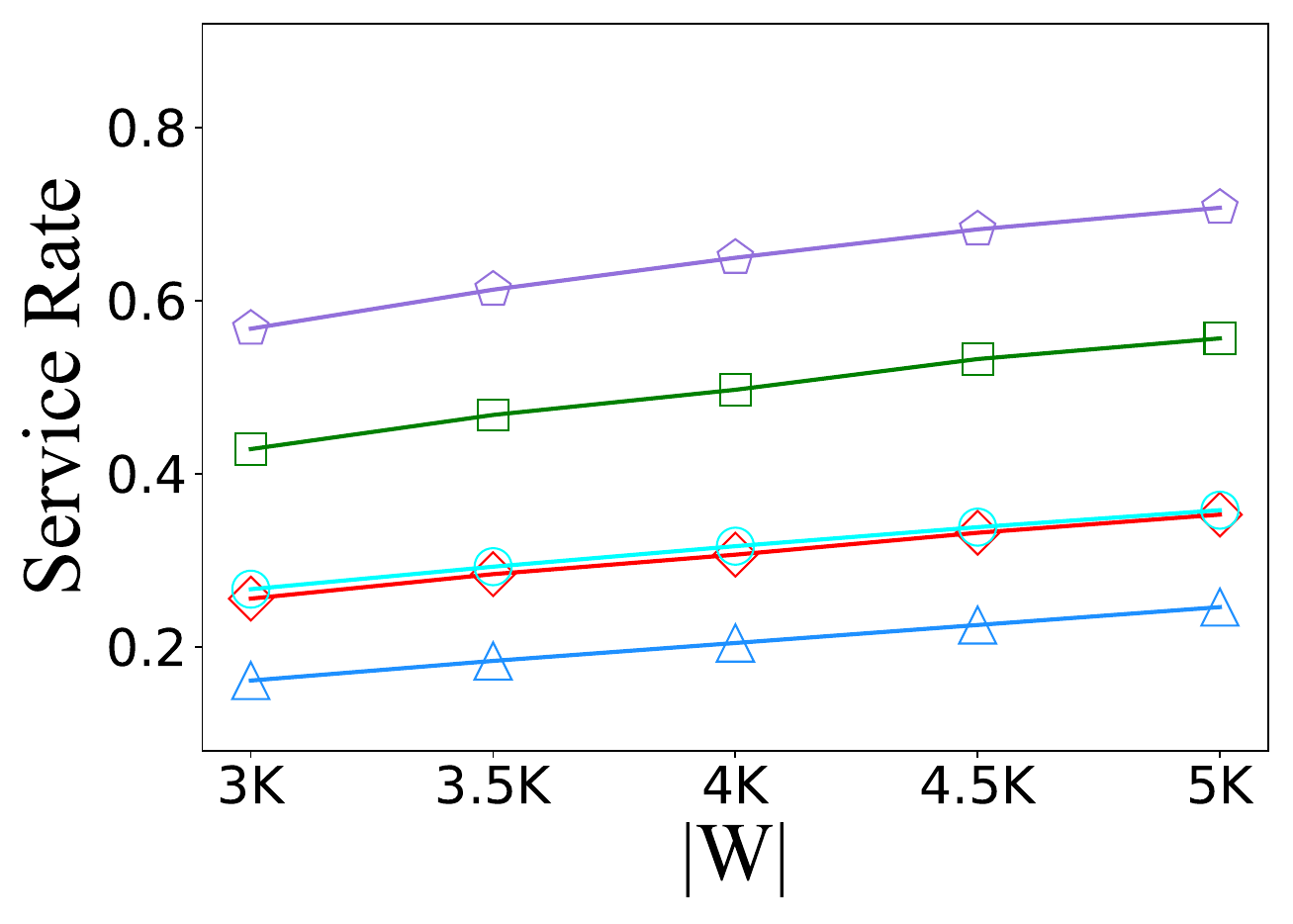}}
		\label{subfig:sr_varing_worker_sh}
	}\hspace{-2ex}
	\subfigure[][{\scriptsize Service Rate (\textit{Cainiao})}]{
		\scalebox{0.158}[0.14]{\includegraphics{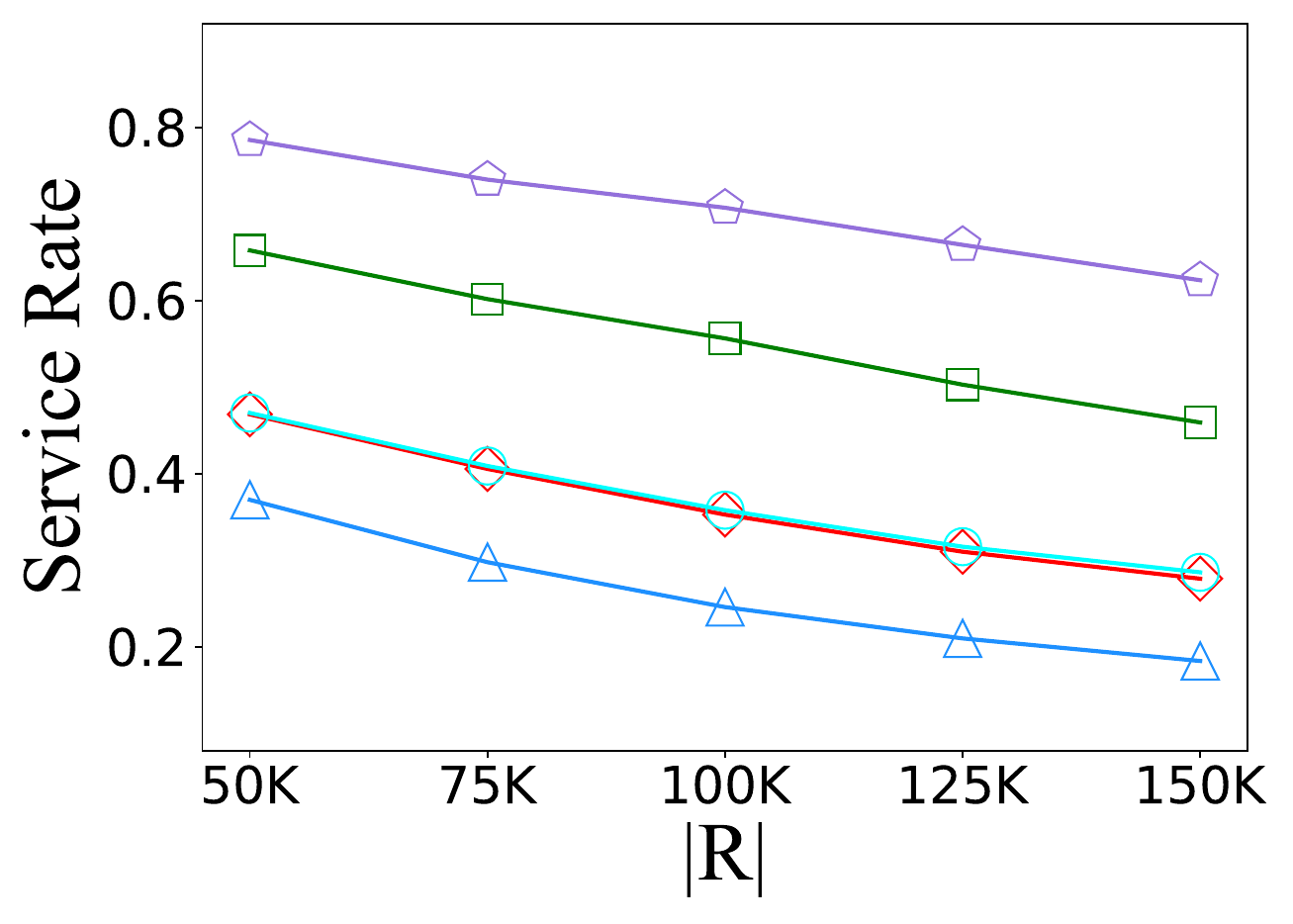}}
		\label{subfig:sr_varing_req_sh}
	} \hspace{-2ex}
	\subfigure[][{\scriptsize Service Rate (\textit{Cainiao})}]{
		\scalebox{0.158}[0.14]{\includegraphics{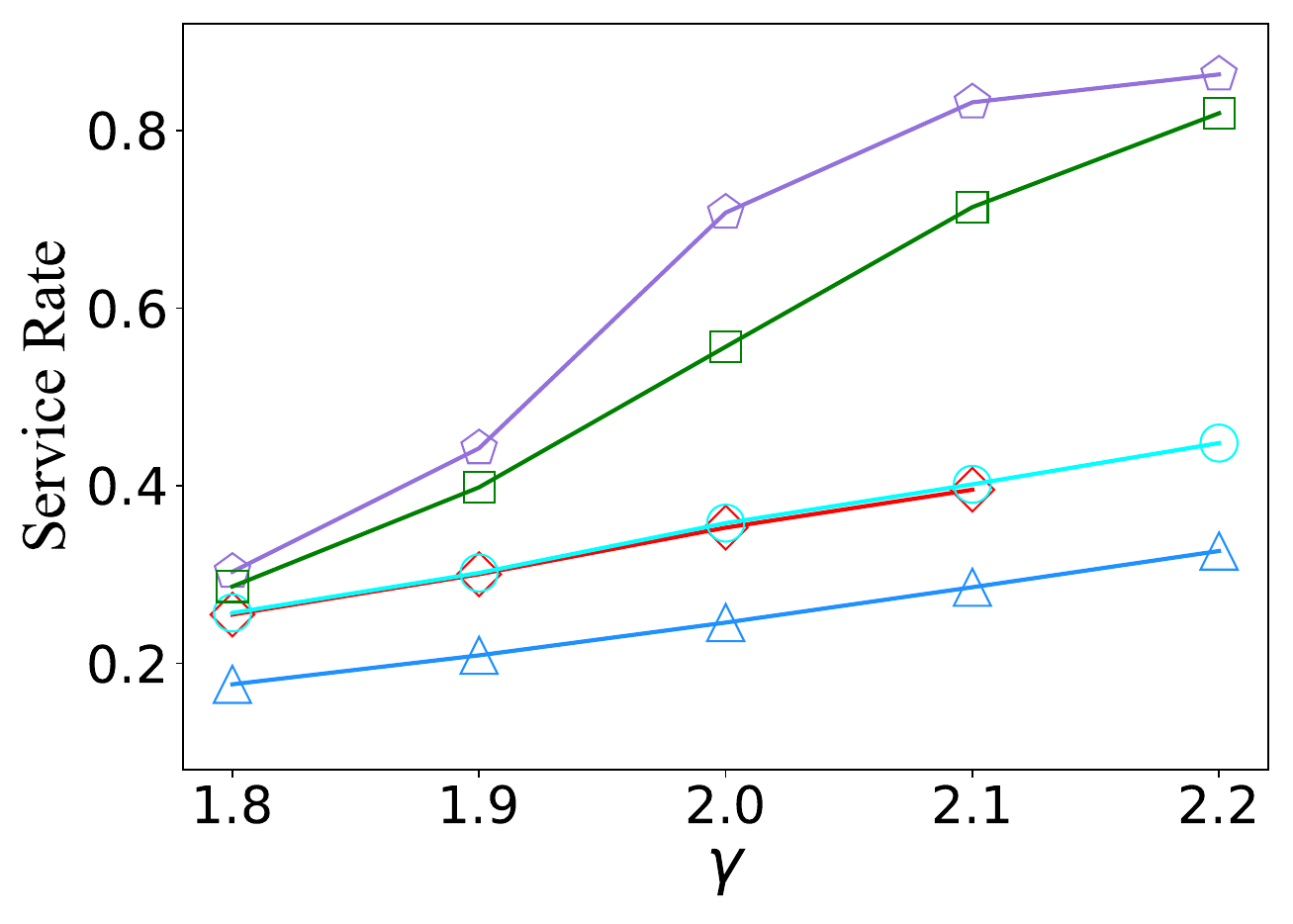}}
		\label{subfig:sr_varing_ddl_sh}
	}\hspace{-2ex}
	\subfigure[][{\scriptsize Service Rate (\textit{Cainiao})}]{
		\scalebox{0.158}[0.14]{\includegraphics{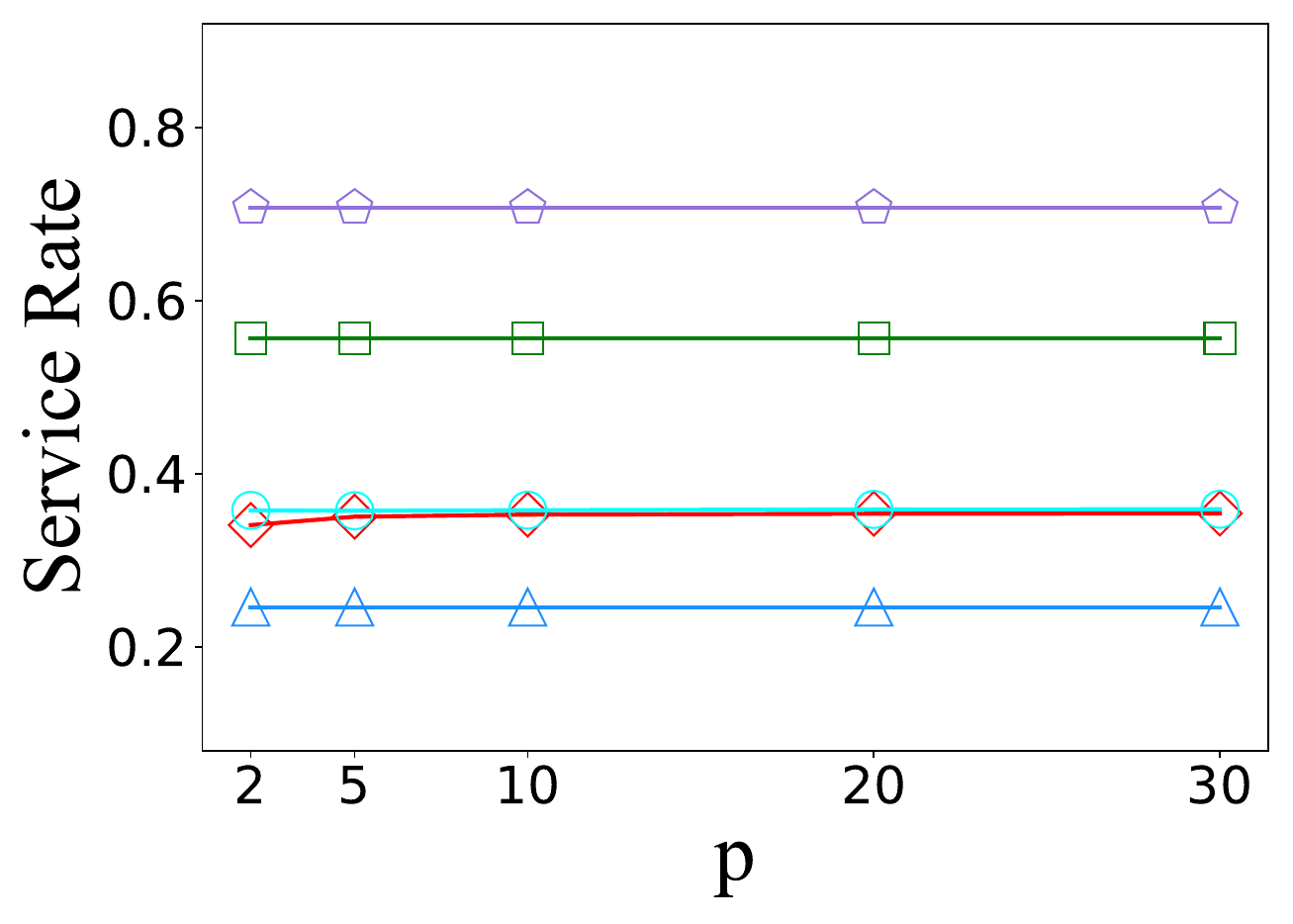}}
		\label{subfig:sr_varing_pen_sh}
	}\hspace{-2ex}
	\subfigure[][{\scriptsize Service Rate (\textit{Cainiao})}]{
		\scalebox{0.158}[0.14]{\includegraphics{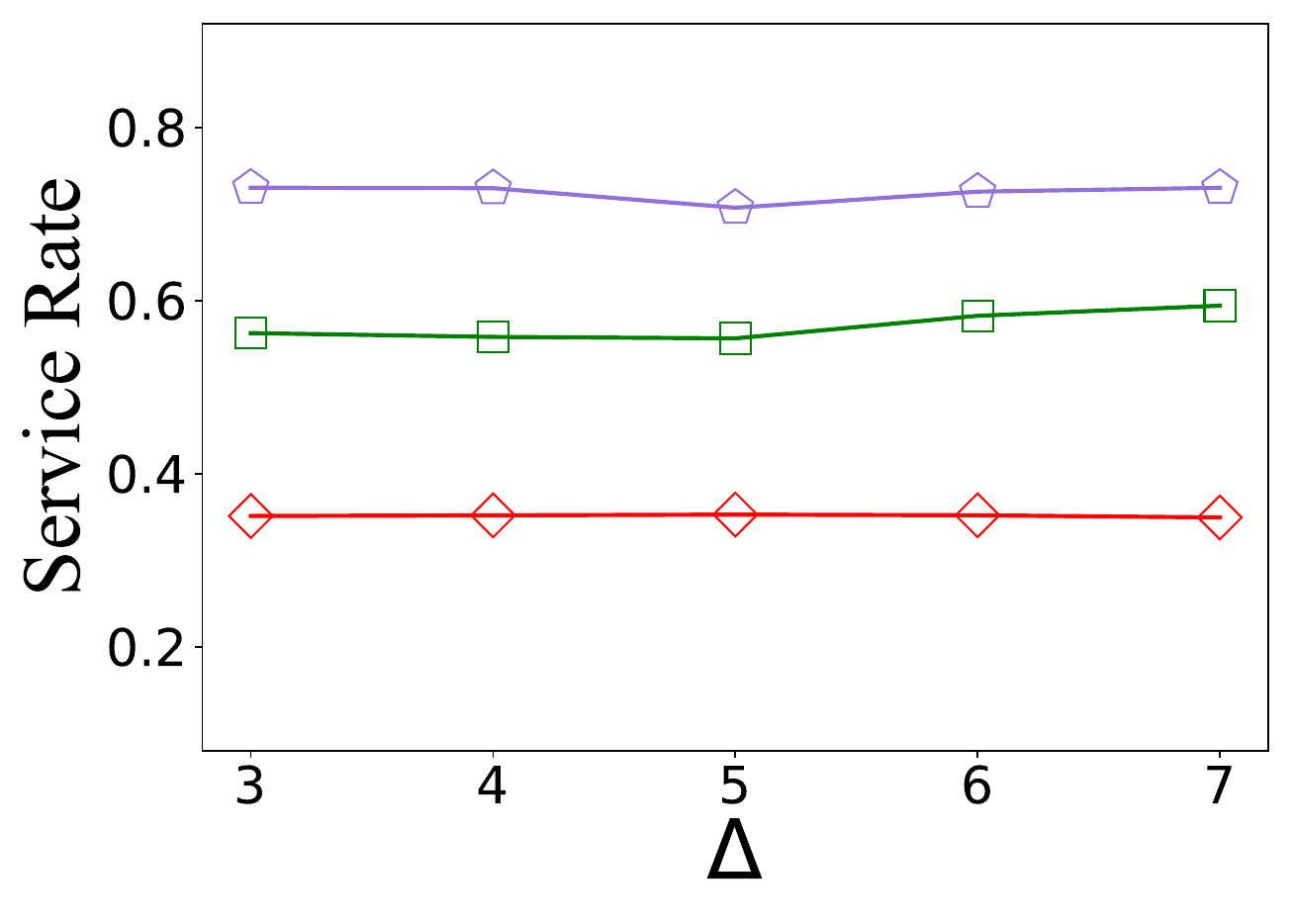}}
		\label{subfig:sr_varing_batch_sh}
	}\vspace{-2ex}
	
	\subfigure[][{\scriptsize Running Time (\textit{Cainiao})}]{
		\scalebox{0.158}[0.14]{\includegraphics{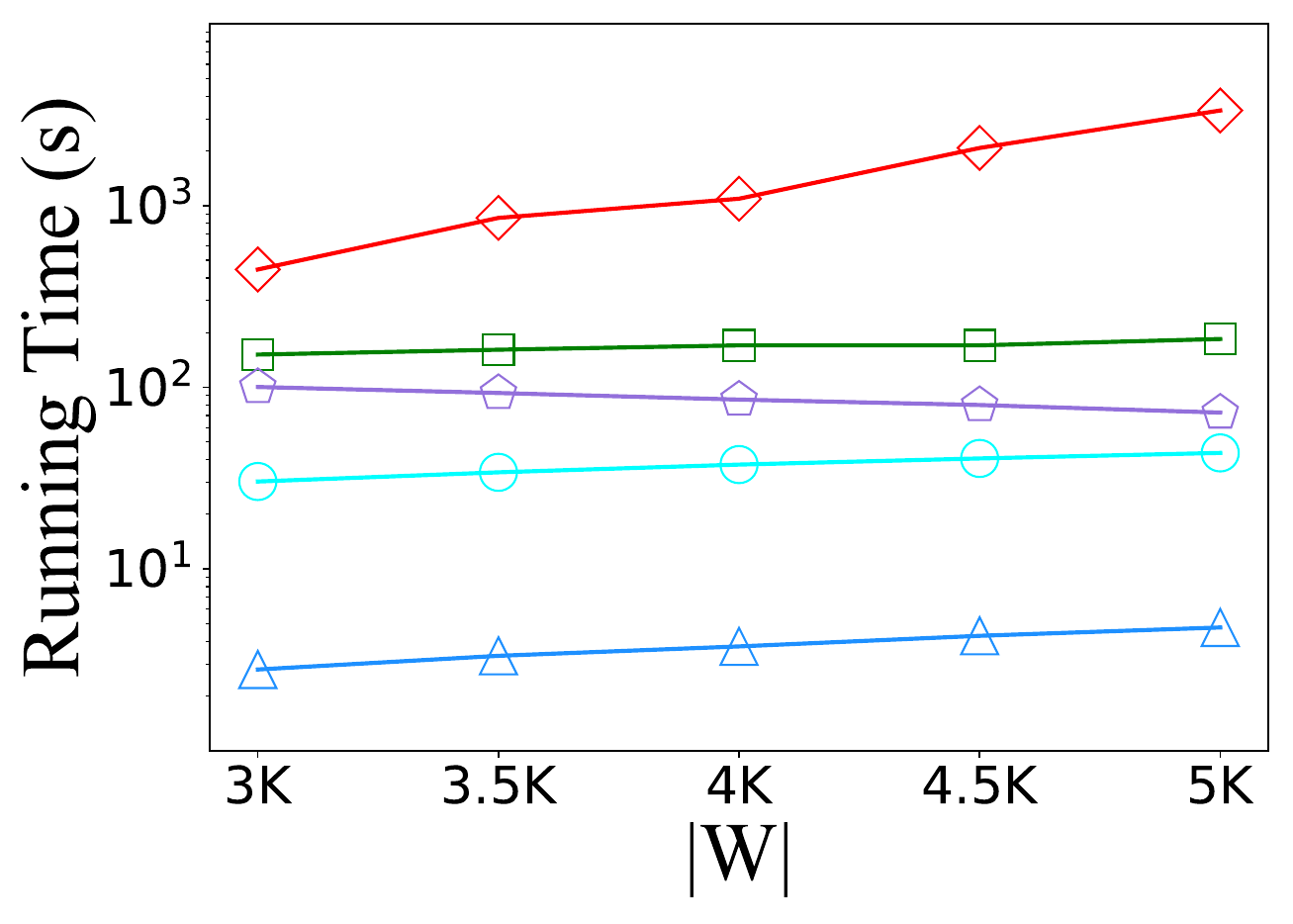}}
		\label{subfig:tm_varing_worker_sh}
	}\hspace{-2ex}
	\subfigure[][{\scriptsize Running Time (\textit{Cainiao})}]{
		\scalebox{0.158}[0.14]{\includegraphics{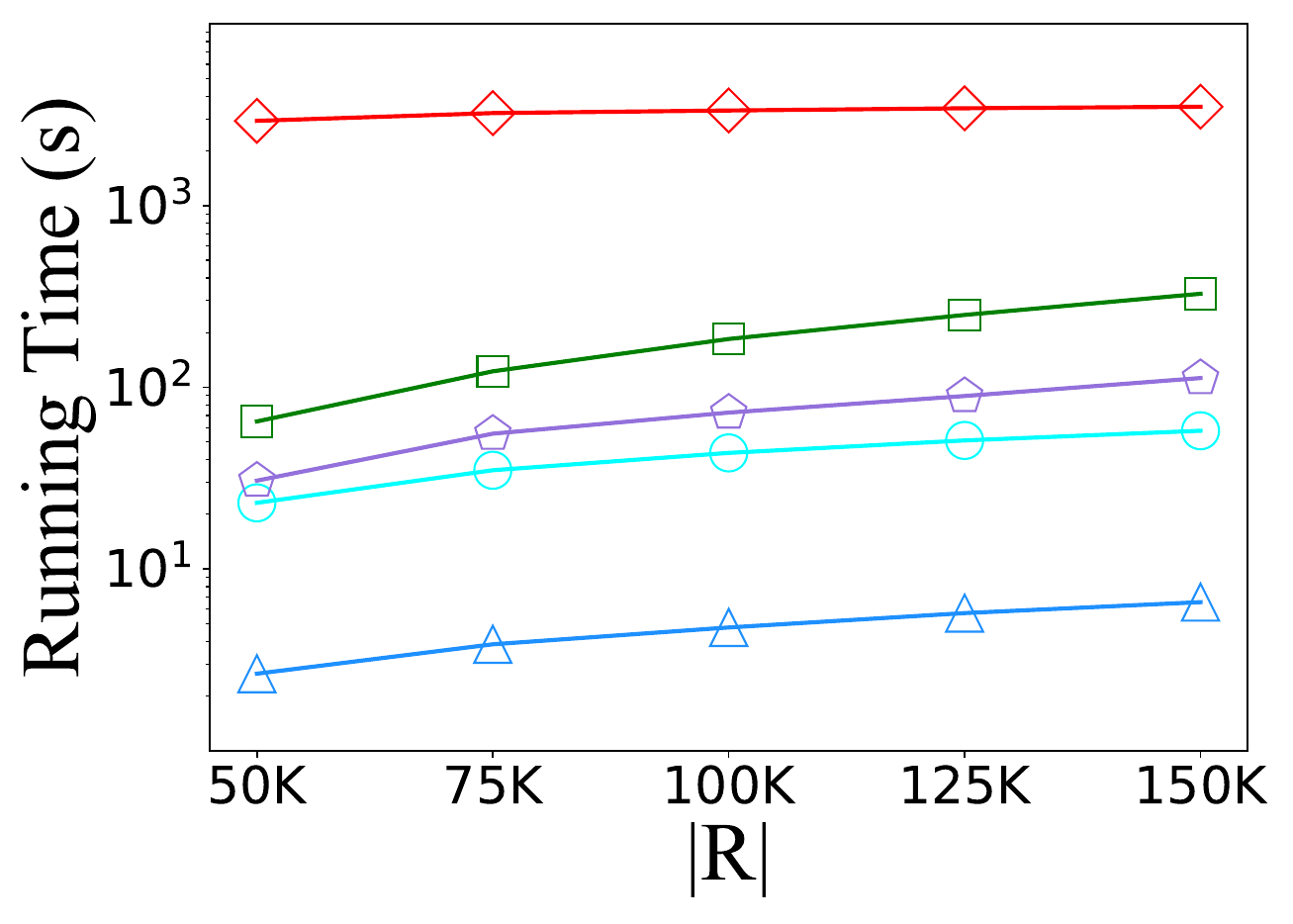}}
		\label{subfig:tm_varing_req_sh}
	} \hspace{-2ex}
	\subfigure[][{\scriptsize Running Time (\textit{Cainiao})}]{
		\scalebox{0.158}[0.14]{\includegraphics{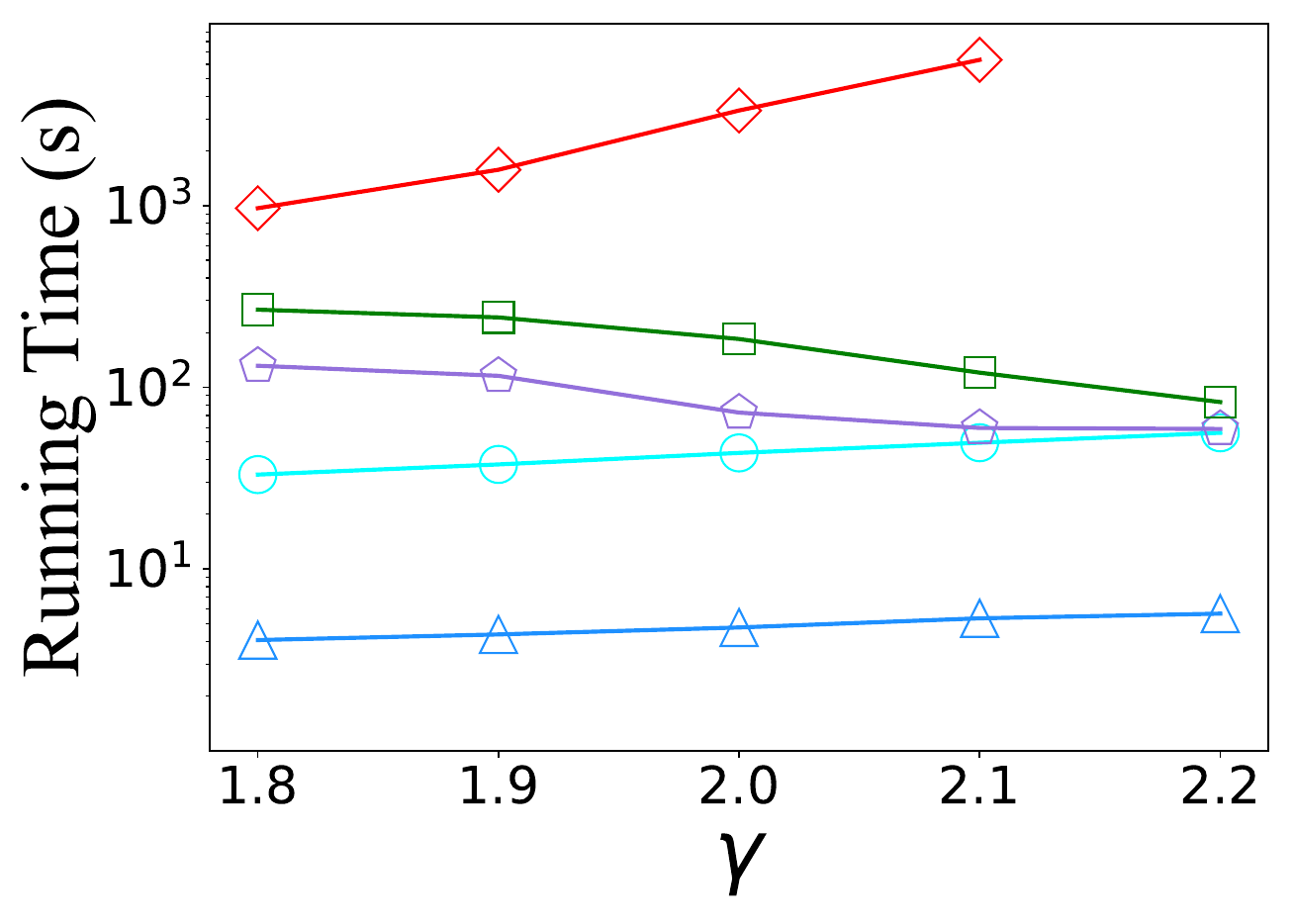}}
		\label{subfig:tm_varing_ddl_sh}
	}\hspace{-2ex}
	\subfigure[][{\scriptsize Running Time (\textit{Cainiao})}]{
		\scalebox{0.158}[0.14]{\includegraphics{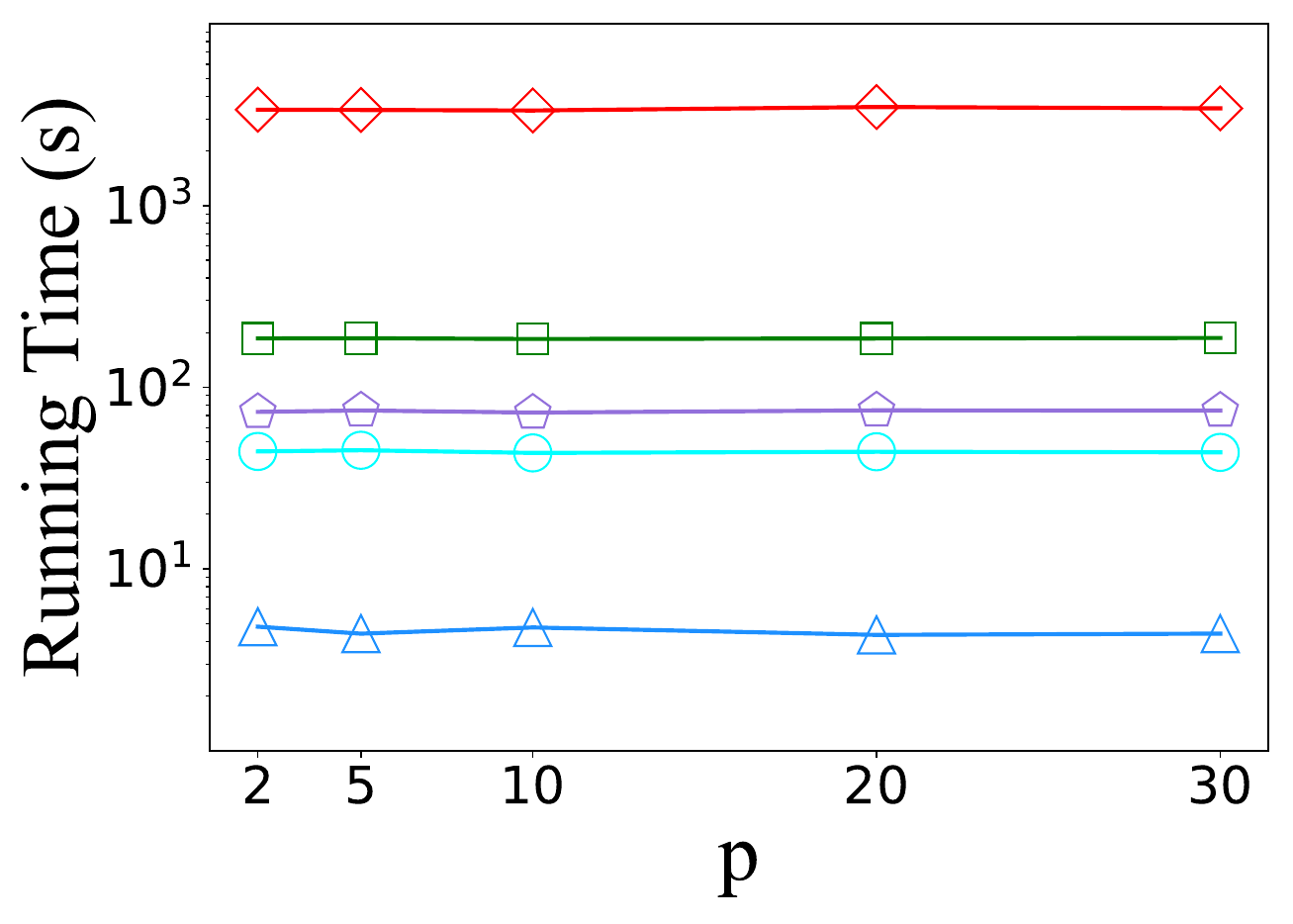}}
		\label{subfig:tm_varing_pen_sh}
	}\hspace{-2ex}
	\subfigure[][{\scriptsize Running Time (\textit{Cainiao})}]{
		\scalebox{0.158}[0.14]{\includegraphics{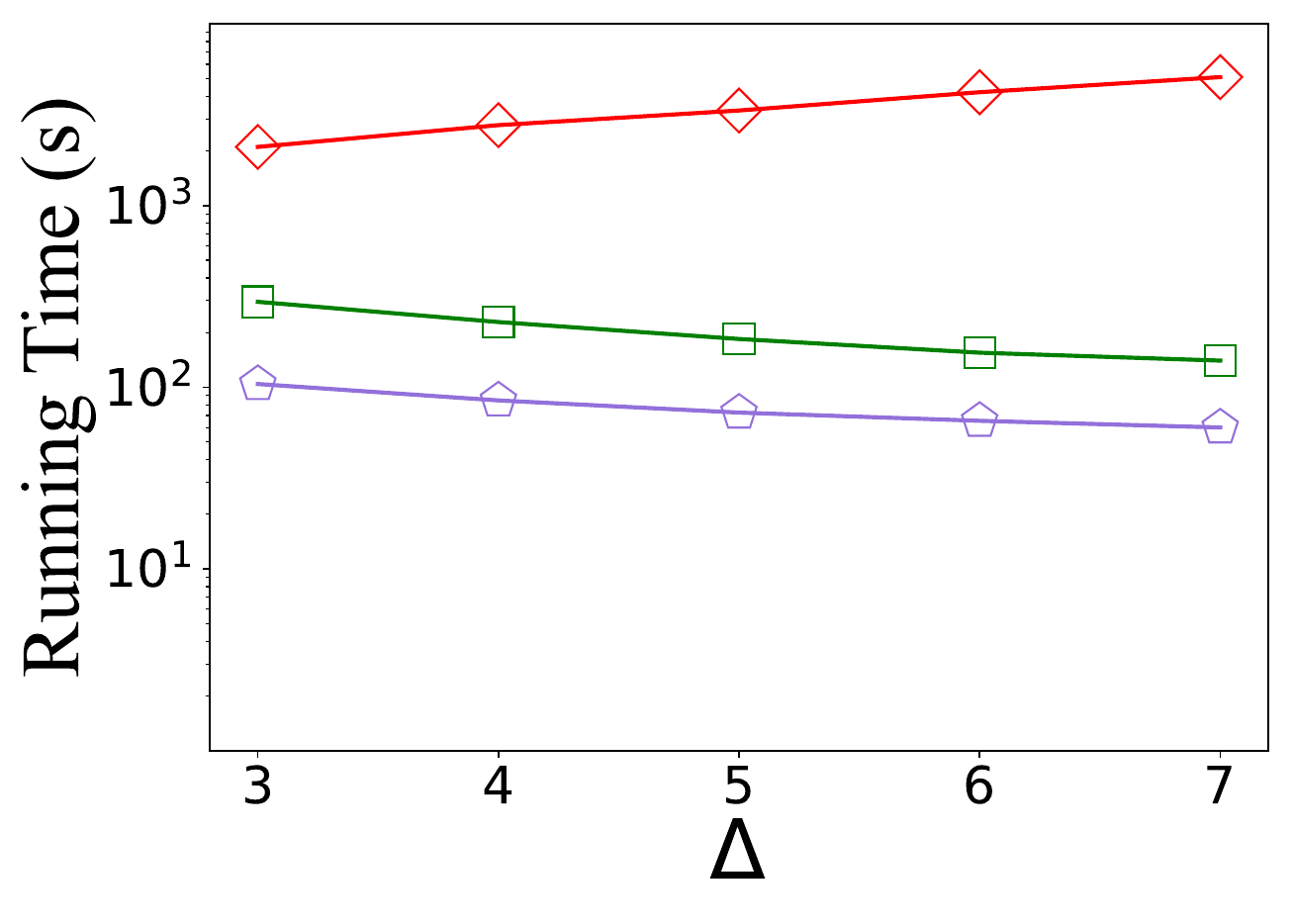}}
		\label{subfig:tm_varing_batch_sh}
	}
	
	\caption{Performance of Cainiao Dataset when Varying $|W|$, $|R|$, $\gamma$, $p_r$ and $\Delta$.}
	\label{fig:sh_result}
\end{figure*}

\noindent\textbf{Effect of the number of vehicles.}
The first column in Figure~\ref{fig:sh_result} shows results as the number of vehicles varies from 3K to 5K.
For unified cost, {SARD} and {GAS} outperform other methods, with {SARD} achieving a $14.16\%\sim 51.85\%$ improvement over other methods on the \textit{Cainiao} dataset.
Our method also demonstrates up to a $46.16\%$ improvement in service rate on the \textit{Cainiao} dataset compared to existing methods.
Compared to pruneGDP, TicketAssign+ enhances service rates through concurrent decision-making processes, effectively mitigating the local optima challenges typically associated with greedy algorithms.
All algorithms exhibit trends in unified cost that correspond consistently with their service rates.
In terms of running time, {pruneGDP} excels as a result of its efficient linear insertion, while {TicketAssign+} experiences increased running time because of vehicle contention.
RTV, GAS, and SARD are slower than pruneGDP.
Notably, {SARD} achieves up to a $46.01\times$ speedup compared to RTV and GAS on the \textit{Cainiao} dataset.
{SARD}'s running time exhibits an inverse correlation with the number of vehicles.
This enhancement in efficiency can be attributed to a reduction in propose-acceptance iterations.
As the number of available vehicles ($m$) increases, competition among riders for individual vehicles decreases, resulting in fewer necessary proposal rounds. 

\noindent\textbf{Effect of the number of requests.}
The second column in Figure~\ref{fig:sh_result} illustrates that as requests increase from 50K to 150K, the unified costs of all algorithms grow. 
{SARD} and {GAS} achieve lower unified costs compared to other methods as the number of requests increases. Regarding service rate, {SARD} outperforms {pruneGDP} by up to $41.59\%$. 
Moreover, {SARD} achieves up to $16.44\%$ higher service rates than the state-of-the-art batch-based method {GAS} on the \textit{Cainiao} dataset. 
In terms of running time, insertion-based methods prove faster. 
Among batch-based methods (SARD, RTV, and GAS), {SARD} demonstrates $2.91\times\sim 31.15\times$ faster performance than {GAS} and {RTV}.

\noindent\textbf{Effect of deadline.}
The third column in Figure~\ref{fig:sh_result} illustrates the impact of varying request deadlines by adjusting the deadline parameter $\gamma$ from $1.8$ to $2.2$.
With a strict deadline of $\gamma=1.8$, {SARD}'s service rate is comparable to existing algorithms. 
This stems from a significant reduction in candidate vehicles for each request, making it difficult to achieve notable performance improvements through batch mode grouping strategies.
However, as the deadline extends, the superiority of {SARD} and {GAS} becomes increasingly evident. 
At a deadline of $2.1\times$, {SARD}'s service rate surpasses $80\%$, up to $54.60\%$ higher than other algorithms on the \textit{Cainiao} dataset.
{RTV} results for $\gamma=2.2$ on the \textit{Cainiao} dataset are omitted due to {RTV-Graph} constraints exceeding the \textit{glpk}~\cite{GLPK} limit. 
{SARD} achieves the best unified cost, saving up to $59.12\%$ compared to others on the \textit{Cainiao} dataset.
Regarding running time, SARD is $1.41\times$ to $106.19\times$ faster than {RTV} and {GAS}.

\noindent\textbf{Effect of penalty.}
The fourth column in Figure~\ref{fig:sh_result} illustrates the effect of varying the penalty coefficient from 2 to 30.
Most methods' service rates remain unaffected by this change.
{pruneGDP}, {TicketAssign+}, {GAS}, and {SARD} incorporate distance, group profit, and shareability loss as indicators in their greedy assignment strategies. Consequently, the penalty coefficient only influences their unified cost scores.
In contrast, {RTV} integrates the penalty coefficient into its linear programming (LP) constraint matrix. This affects LP results solely when the penalty is small; for larger penalties, the impact becomes negligible.
The unified cost for all algorithms shows a proportional relationship to the penalty coefficient.
{SARD} outperforms in \textit{Cainiao} datasets, achieving a service rate increase of $15.10\%\sim 46.16\%$.
As the penalty coefficient doesn't influence the assignment phase, execution times remain consistent across datasets.

\begin{figure*}[t!]\centering
	\begin{tabularx}{\textwidth}{XX}
		\begin{minipage}[t]{.5\textwidth}
			\subfigure{
				\scalebox{0.5}[0.5]{\includegraphics{legend_cainiao-eps-converted-to.pdf}}}\hfill
			\addtocounter{subfigure}{-1}\\[-3ex]
			\subfigure[][{\scriptsize Unified Cost (\textit{Cainiao})}]{
				\raisebox{-1ex}{\scalebox{0.19}[0.17]{\includegraphics{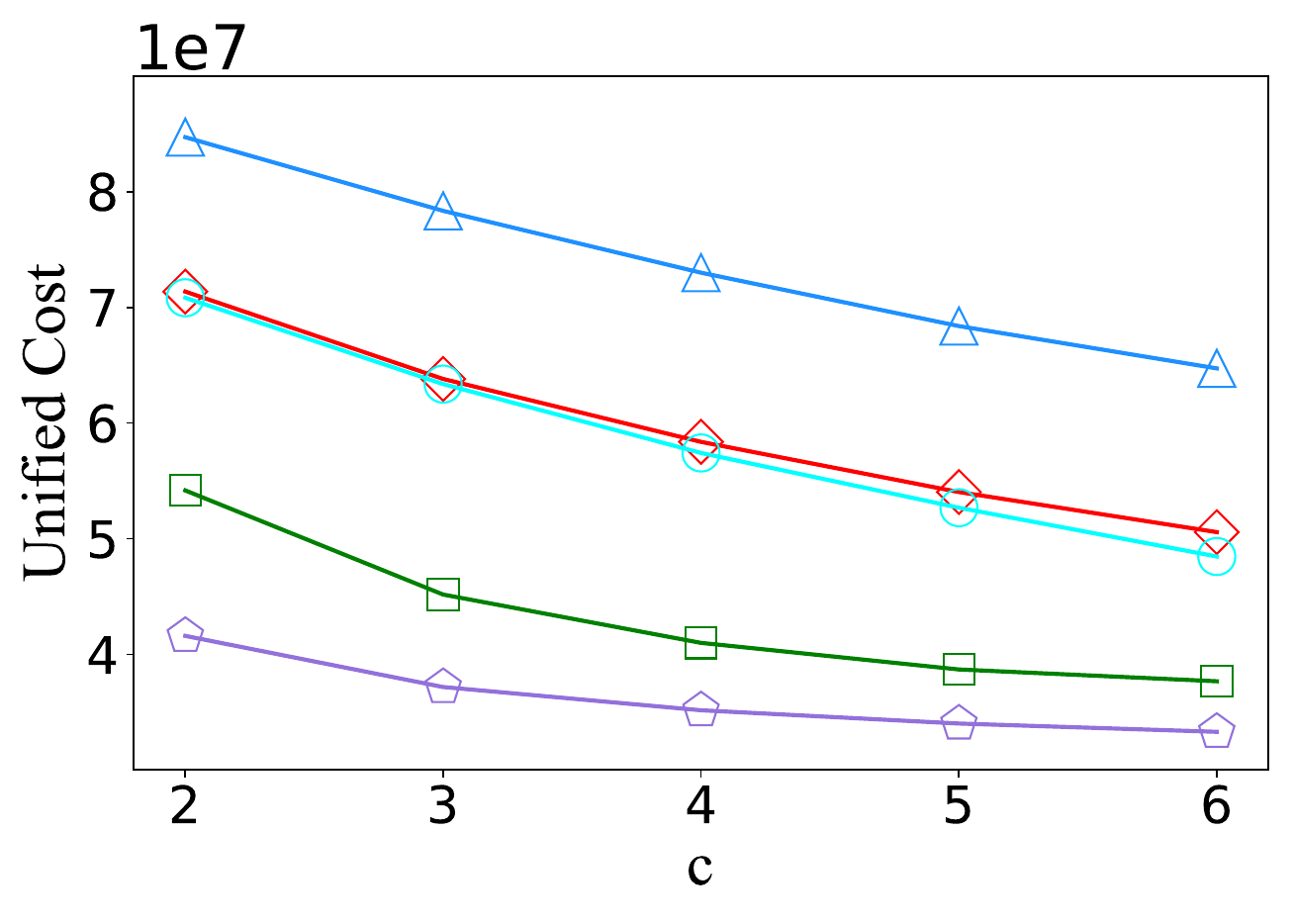}}}
				\label{subfig:uc_varing_cap_sh}}\hspace{-2ex}
			\subfigure[][{\scriptsize Unified Cost (\textit{Cainiao})}]{
				\raisebox{-1ex}{\scalebox{0.19}[0.17]{\includegraphics{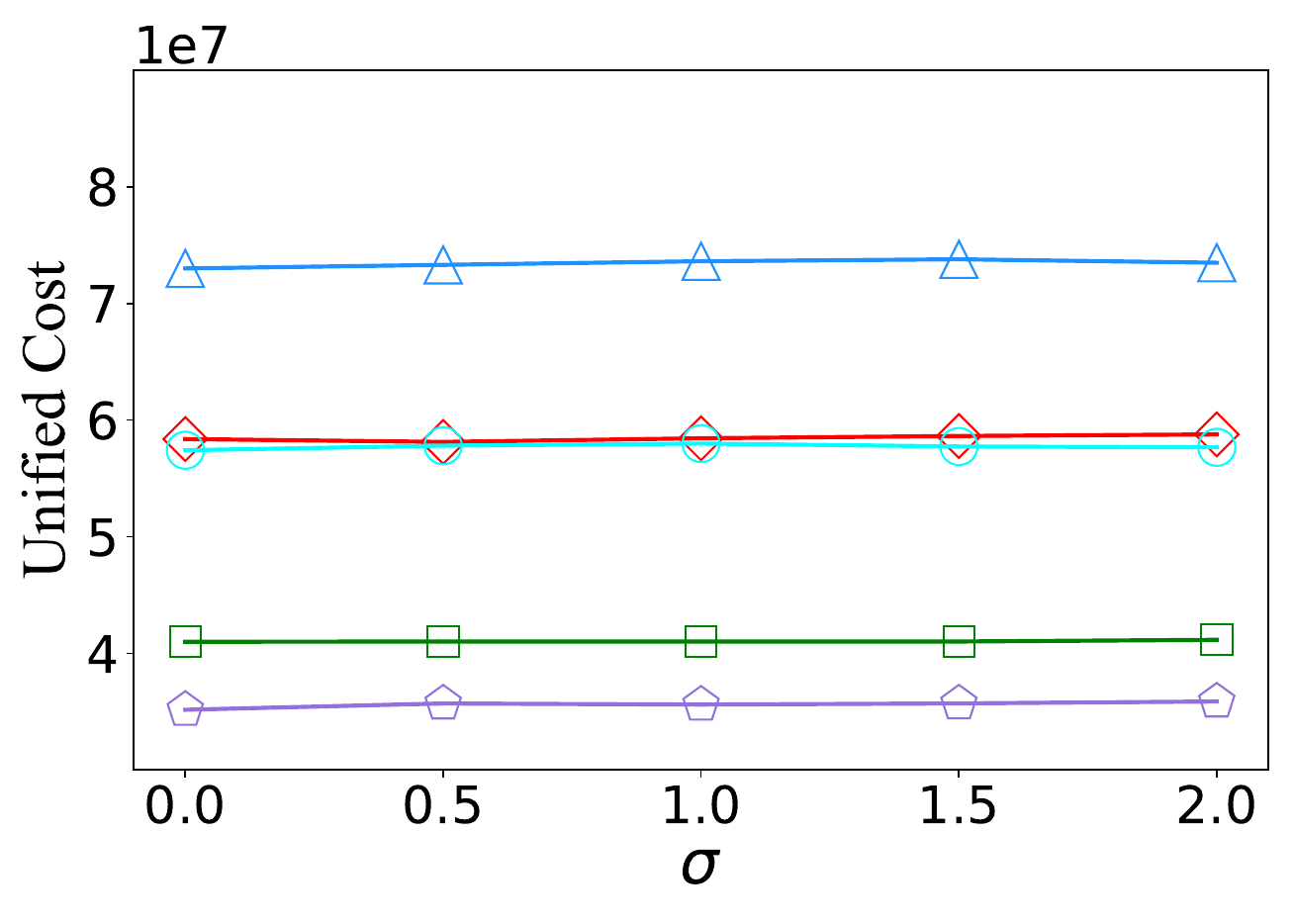}}}
				\label{subfig:uc_varing_var_sh}}\\[-2ex]
			\subfigure[][{\scriptsize Service Rate (\textit{Cainiao})}]{
				\raisebox{-1ex}{\scalebox{0.19}[0.17]{\includegraphics{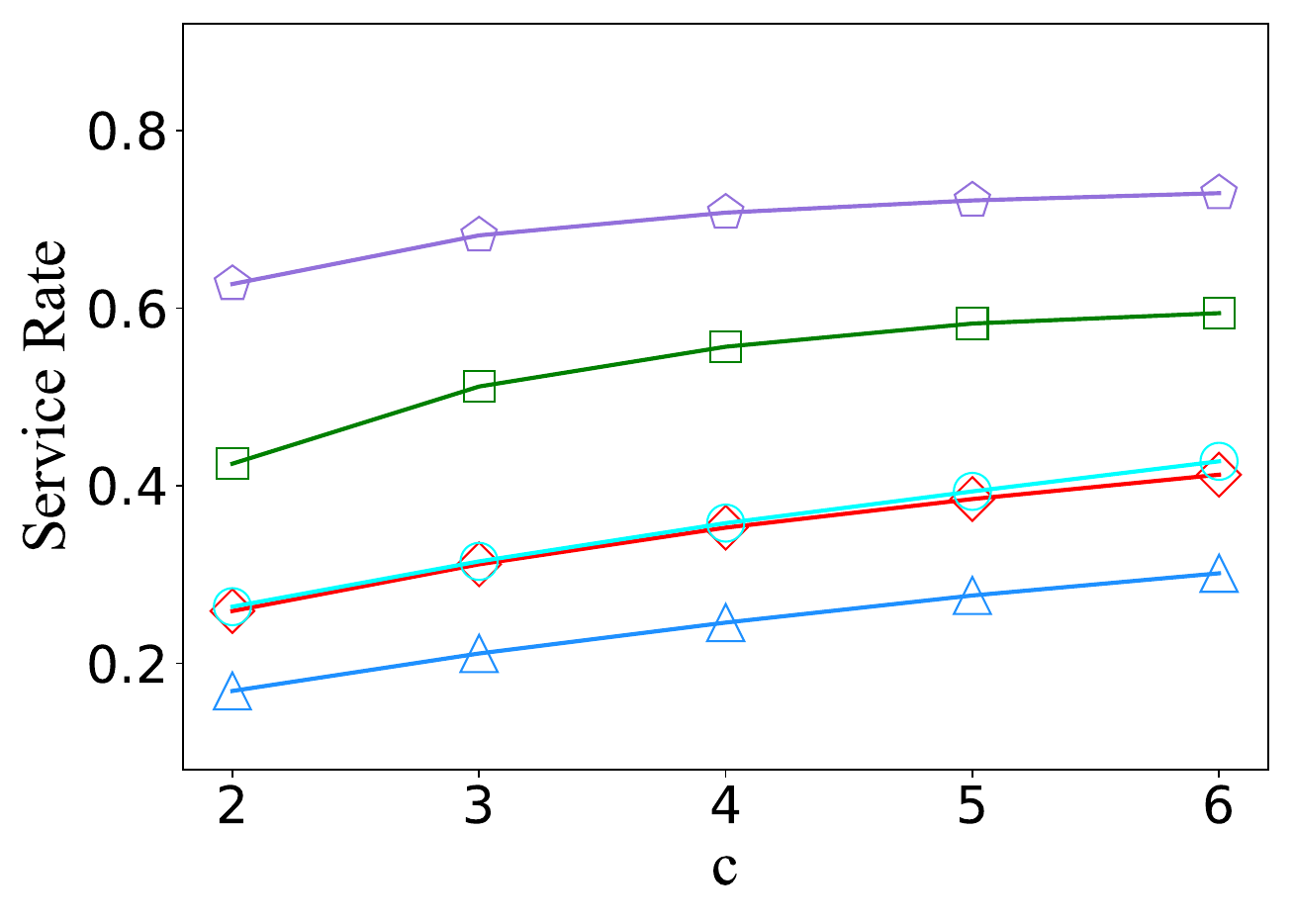}}}
				\label{subfig:sr_varing_cap_sh}}\hspace{-2ex}
			\subfigure[][{\scriptsize Service Rate (\textit{Cainiao})}]{
				\raisebox{-1ex}{\scalebox{0.19}[0.17]{\includegraphics{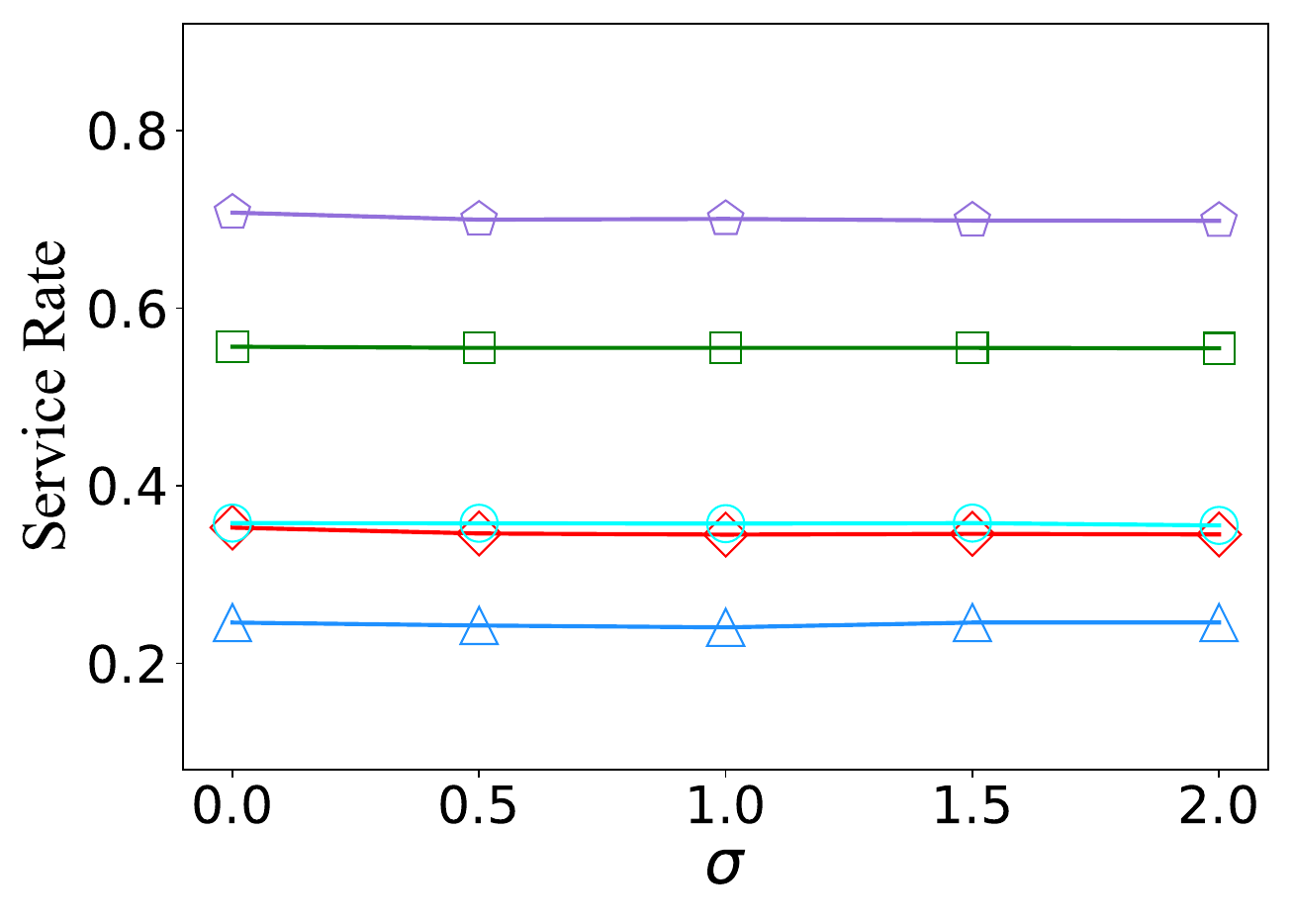}}}
				\label{subfig:sr_varing_var_sh}}\\[-2ex]
			\subfigure[][{\scriptsize Running Time (\textit{Cainiao})}]{
				\raisebox{-1ex}{\scalebox{0.19}[0.17]{\includegraphics{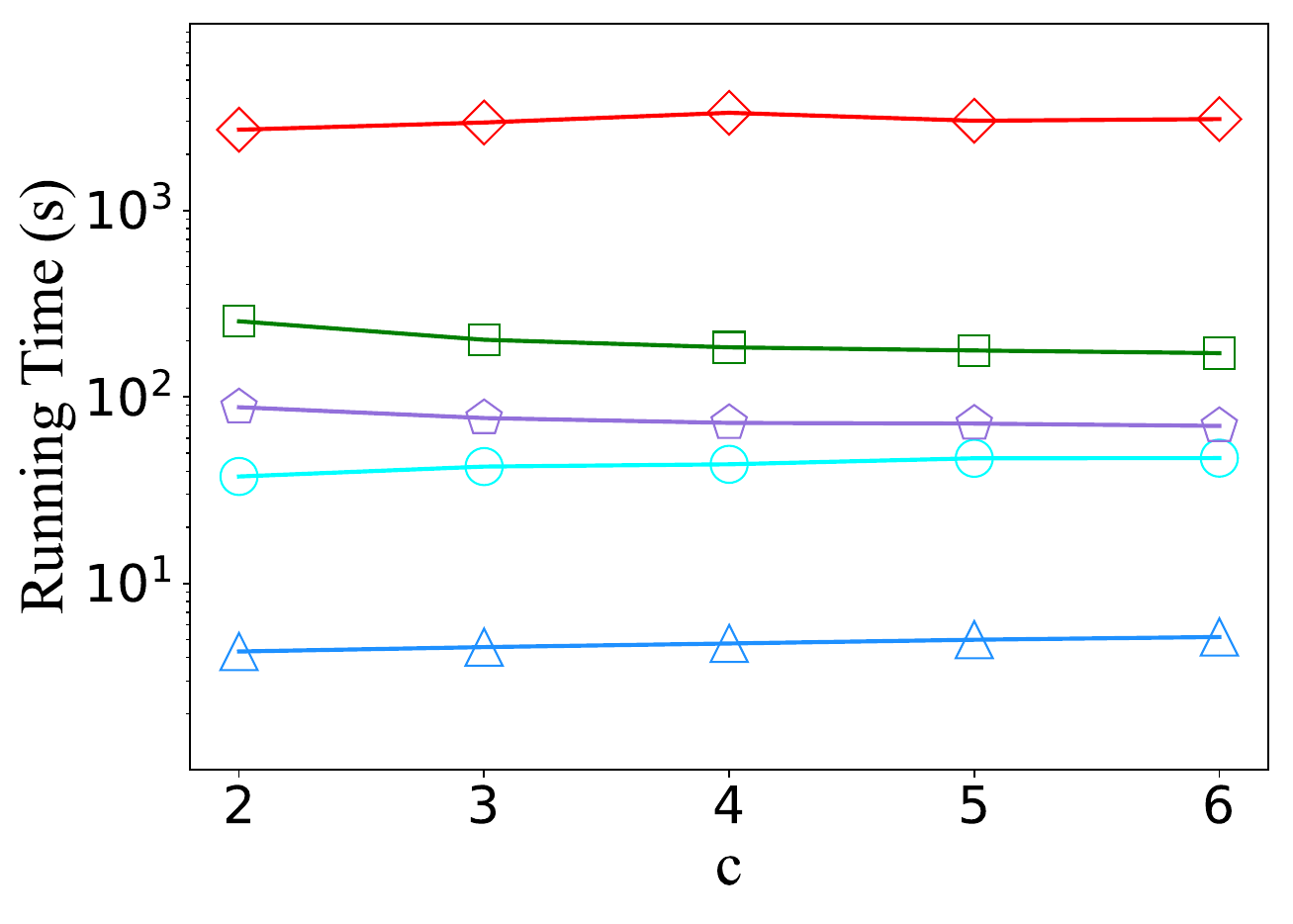}}}
				\label{subfig:tm_varing_cap_sh}}\hspace{-2ex}
			\subfigure[][{\scriptsize Running Time (\textit{Cainiao})}]{
				\raisebox{-1ex}{\scalebox{0.19}[0.17]{\includegraphics{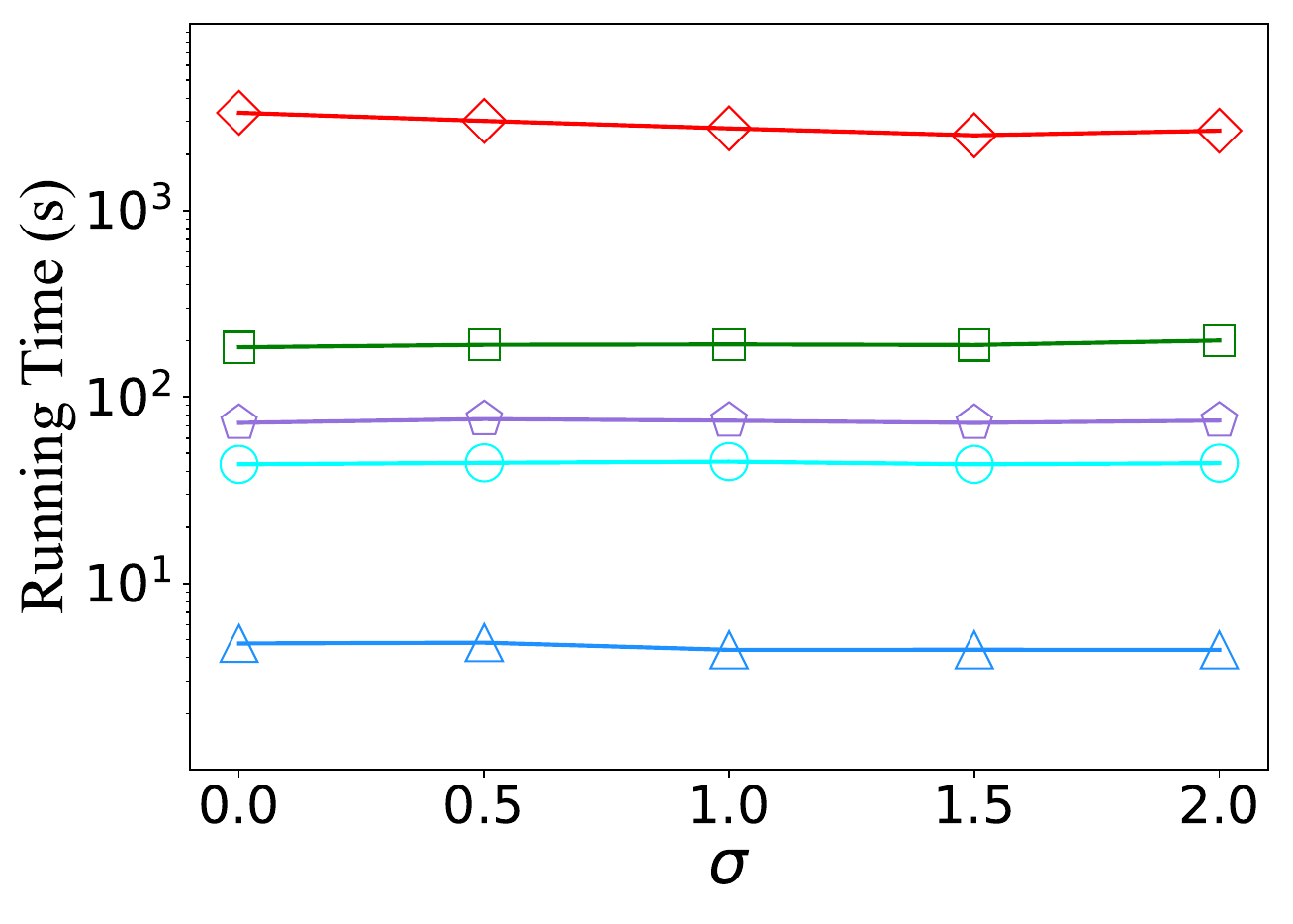}}}
				\label{subfig:tm_varing_var_sh}}
			\caption{ Performance of Varying $c$ and $\sigma$.}
			\label{fig:vary_cap}
		\end{minipage} &
		\begin{minipage}[t]{.5\textwidth}
			\subfigure{
				\scalebox{0.45}[0.45]{\includegraphics{legend-eps-converted-to.pdf}}}\hfill
			\addtocounter{subfigure}{-1}\\[-3ex]
			\subfigure[][{\scriptsize Unified Cost (\textit{CHD})}]{
				\raisebox{-1ex}{\scalebox{0.19}[0.17]{\includegraphics{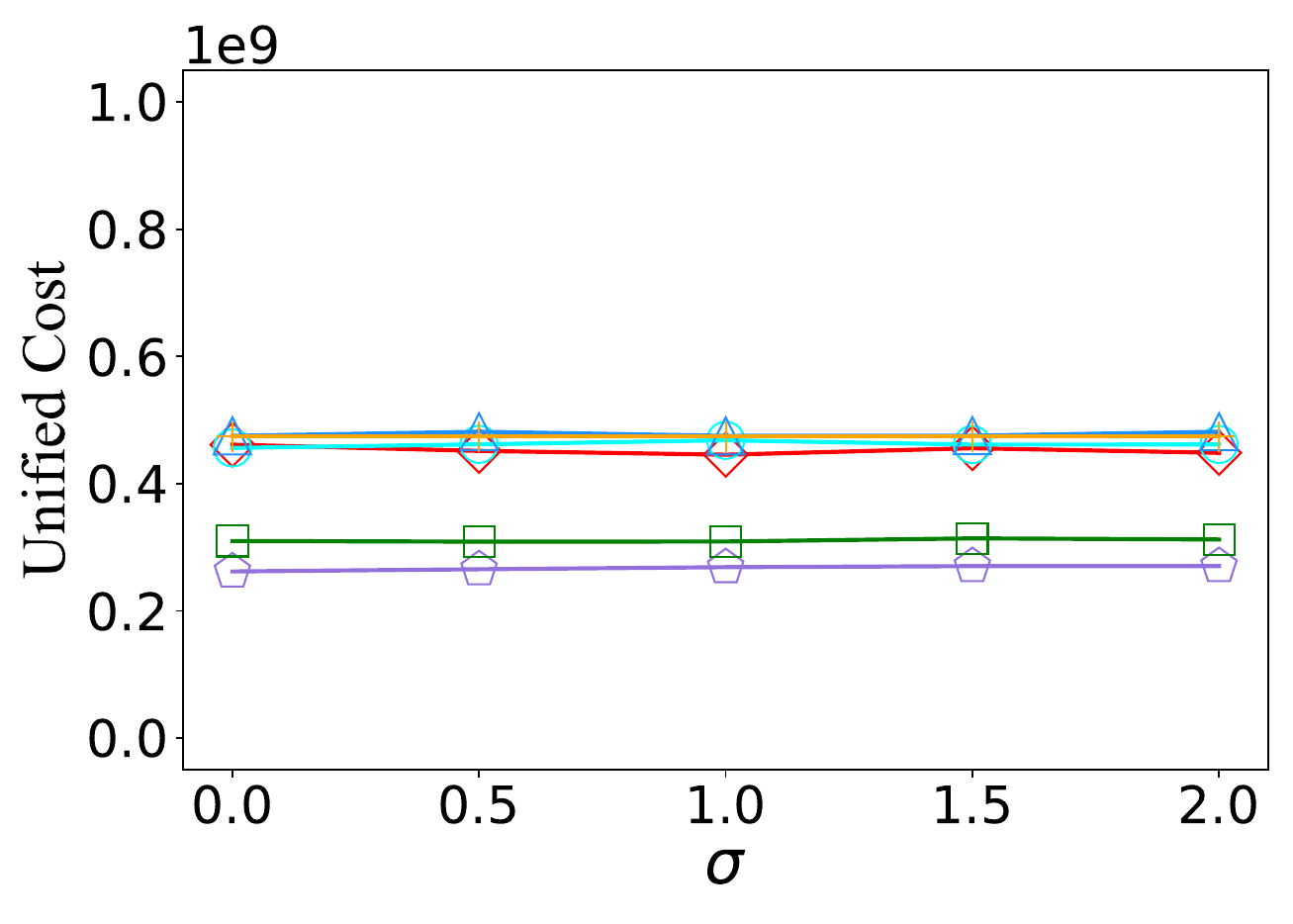}}}
				\label{subfig:uc_varing_var_cd}}\hspace{-2ex}
			\subfigure[][{\scriptsize Unified Cost (\textit{NYC})}]{
				\raisebox{-1ex}{\scalebox{0.19}[0.17]{\includegraphics{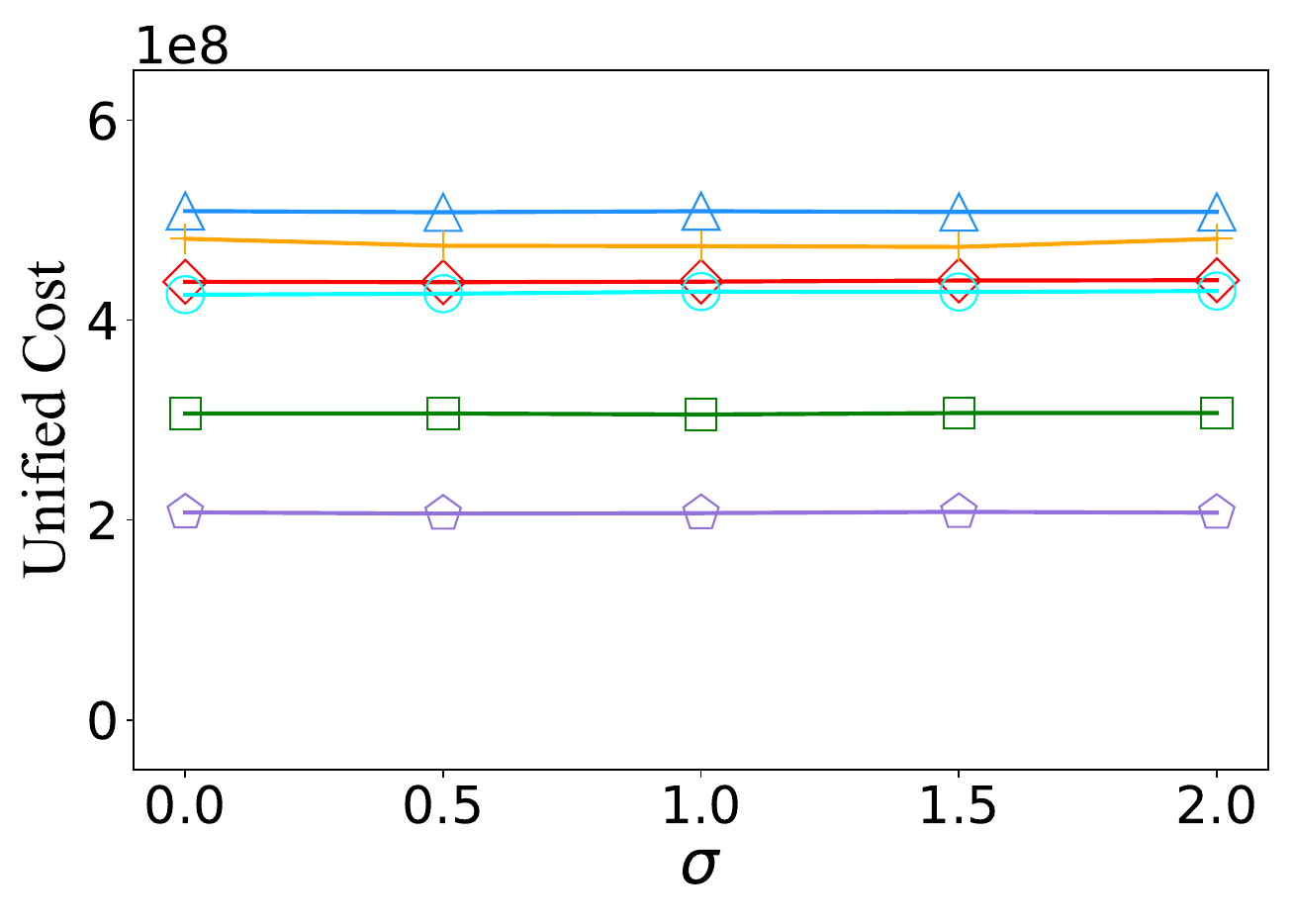}}}
				\label{subfig:uc_varing_var_nyc}}\\[-2ex]
			\subfigure[][{\scriptsize Service Rate (\textit{CHD})}]{
				\raisebox{-1ex}{\scalebox{0.19}[0.17]{\includegraphics{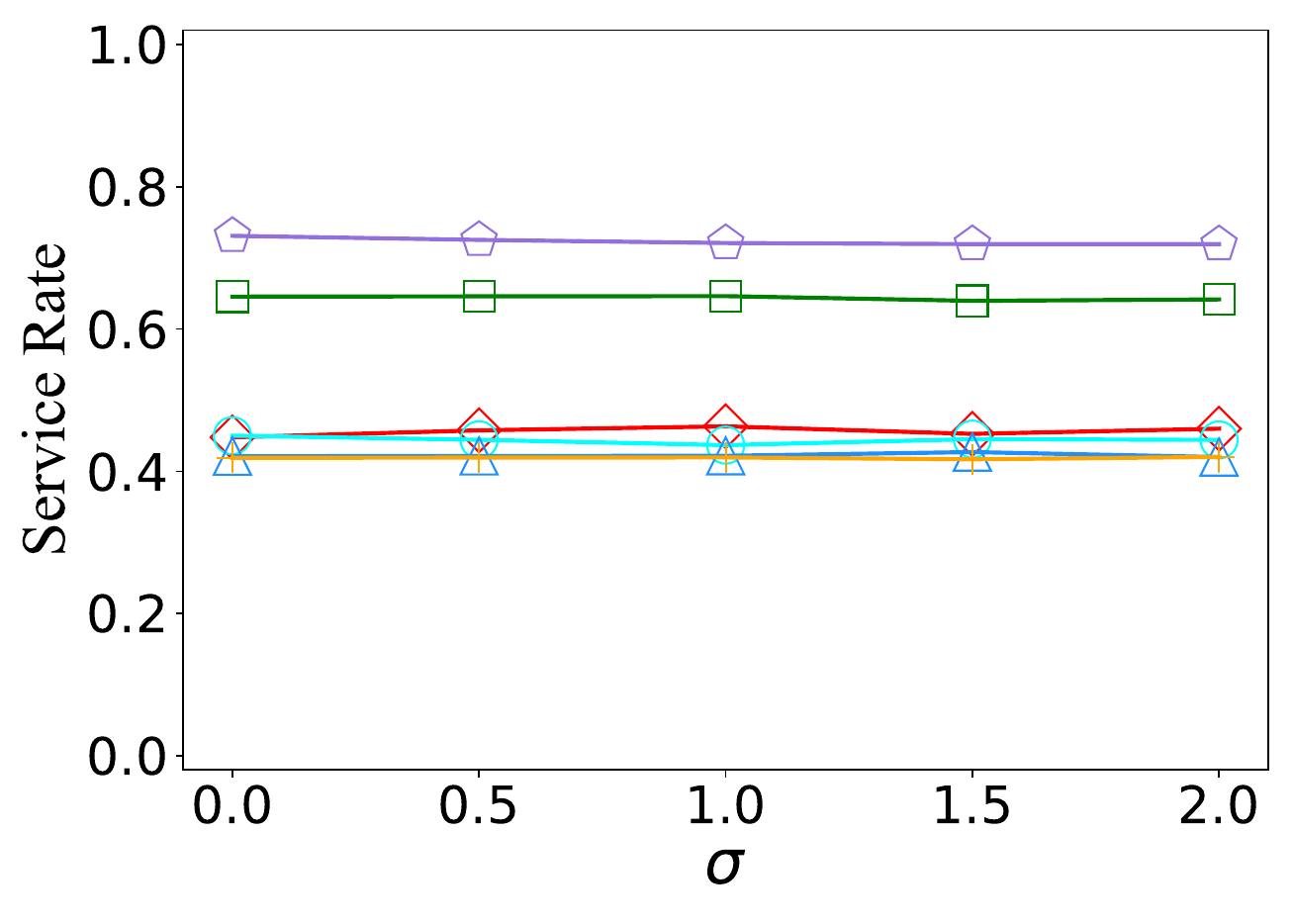}}}
				\label{subfig:sr_varing_var_cd}}\hspace{-2ex}
			\subfigure[][{\scriptsize Service Rate (\textit{NYC})}]{
				\raisebox{-1ex}{\scalebox{0.19}[0.17]{\includegraphics{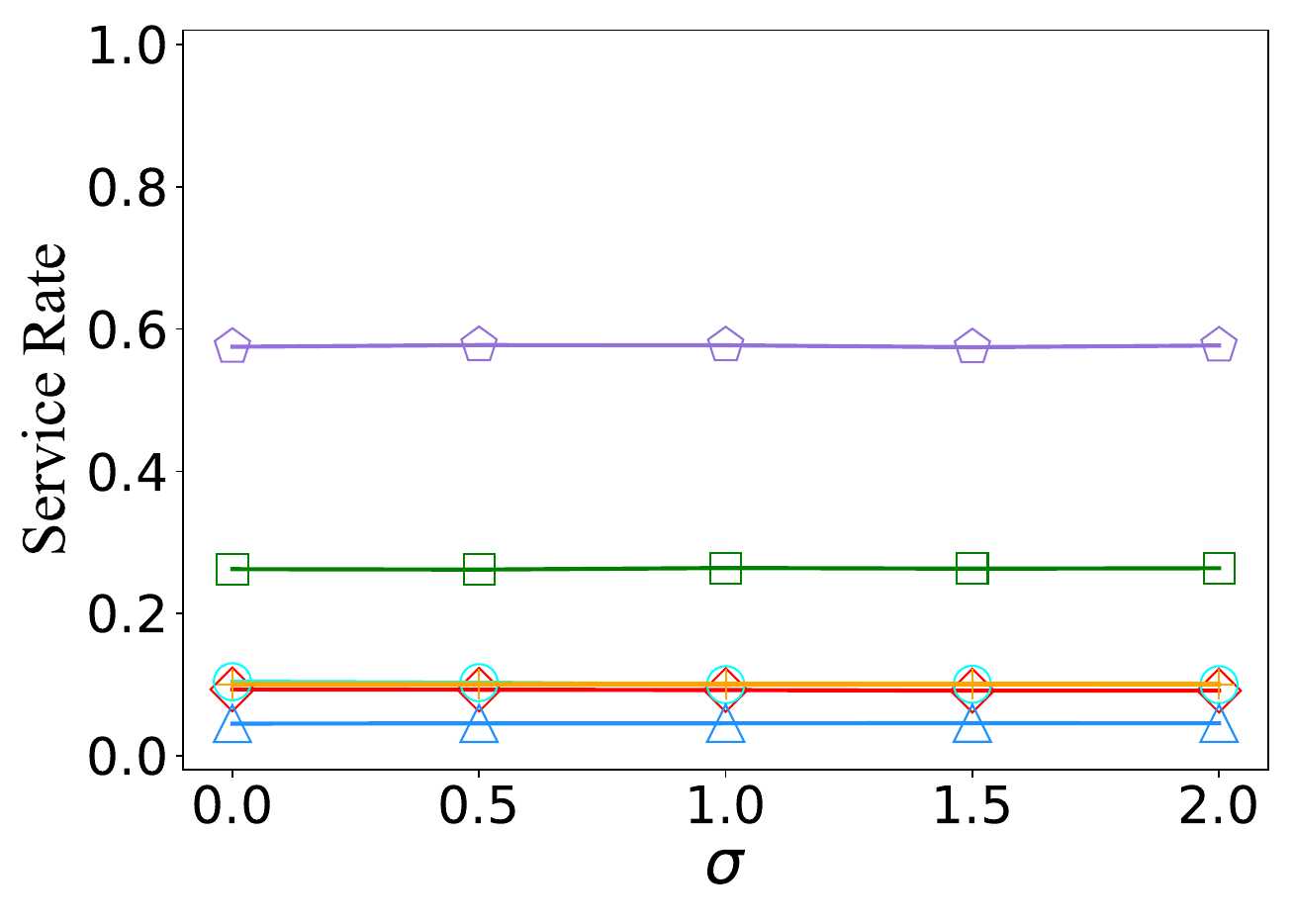}}}
				\label{subfig:sr_varing_var_nyc}}\\[-2ex]
			\subfigure[][{\scriptsize Running Time (\textit{CHD})}]{
				\raisebox{-1ex}{\scalebox{0.19}[0.17]{\includegraphics{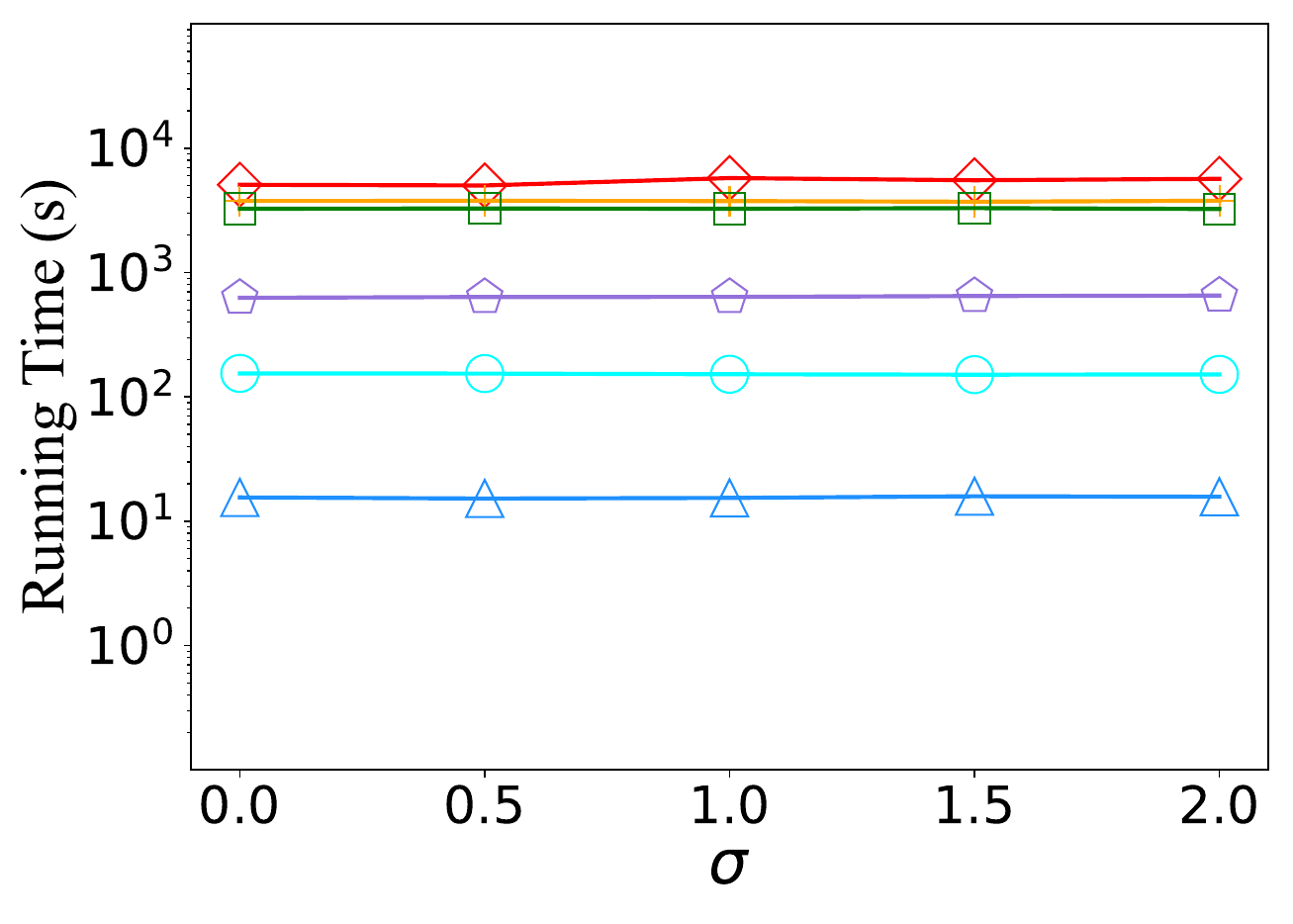}}}
				\label{subfig:tm_varing_var_cd}}\hspace{-2ex}
			\subfigure[][{\scriptsize Running Time (\textit{NYC})}]{
				\raisebox{-1ex}{\scalebox{0.19}[0.17]{\includegraphics{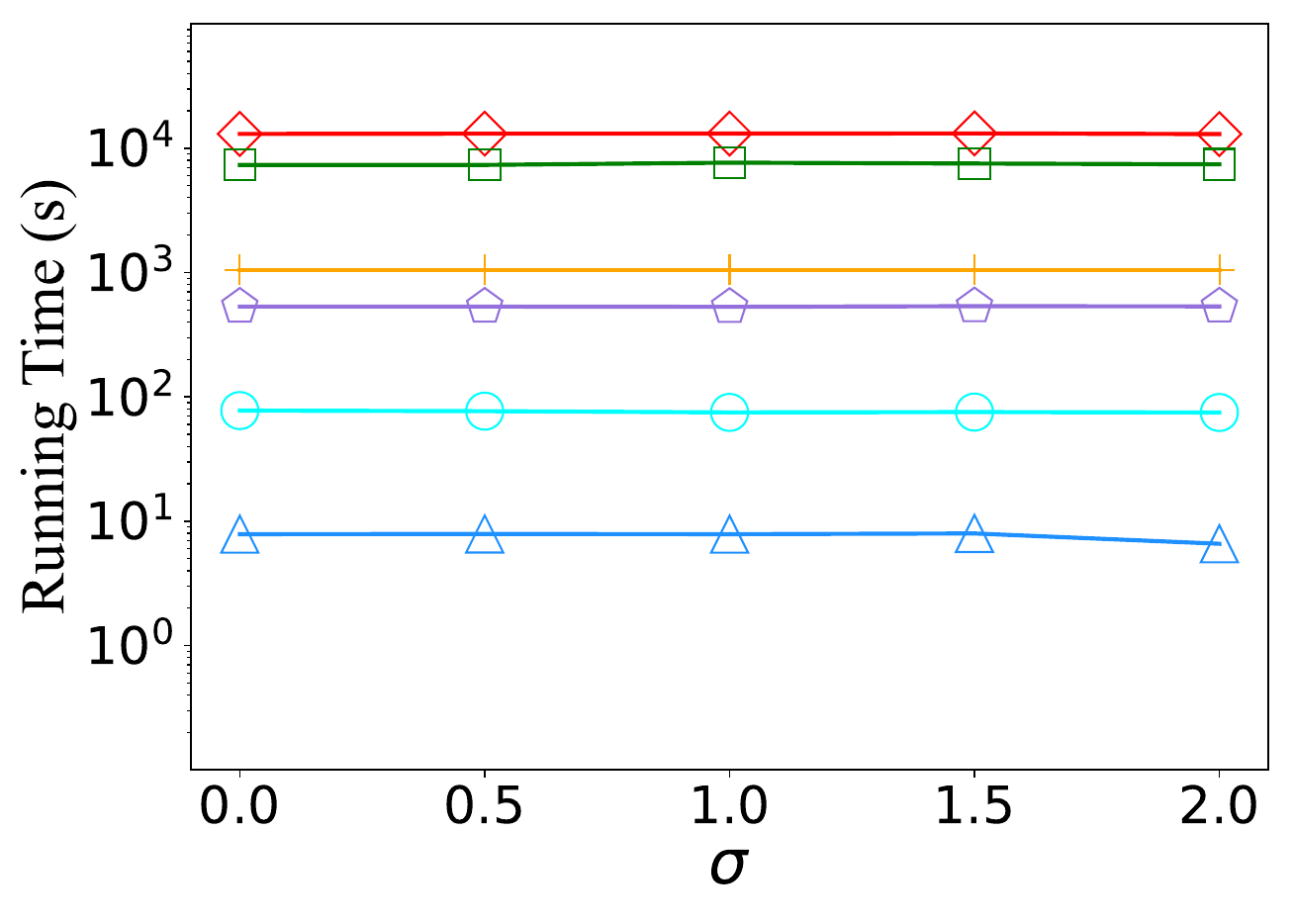}}}
				\label{subfig:tm_varing_var_nyc}}
			\caption{ Performance of Varying $\sigma$.}
			\label{fig:vary_var}
		\end{minipage}
	\end{tabularx}
\end{figure*}

\noindent\textbf{Effect of the batching period.}
The final column in Figure~\ref{fig:sh_result} illustrates the impact of varying the batching period from 3 to 7 seconds for batch-based methods.
This variation in batch time influences the unified cost and service rate of batch-based algorithms.
{SARD} demonstrates superior performance across varying batching periods, achieving up to $39.48\%$ reduction in unified cost and a $38.1\%$ increase in service rate on the \textit{Cainiao} dataset. 
Regarding execution time, as the duration of each batch increases, the number of execution rounds decreases, resulting in reduced running times for most methods. 
However, {RTV} experiences increased time costs due to a higher number of matches. 
Notably, {SARD} exhibits $2.33\times\sim 84.82\times$ faster performance compared to {RTV} and {GAS}.

\noindent\textbf{Effects of Angle pruning strategy.}
Table~\ref{tab:sssp_pruning_sh} shows the effect of the Angle pruning strategies proposed in Section~\ref{subsec:sn_gen} (parameters are in default values in Table \ref{tab:cainiao_settings}).
We note the method without pruning strategies as \textit{SARD}, with the Angle pruning as \textit{SARD-O}. 
SARD-O saves up to 41.9\% of the shortest path queries and saves 33.9\% of the total running time on the \textit{Cainiao} dataset compared with {SARD}. 
Besides, it has almost no harm on the service rate and unified cost.
\begin{table}[h!]
	\begin{center}
		{\scriptsize
			\caption{ Performance of Pruning Strategies.} 
			\label{tab:sssp_pruning_sh}
			\begin{tabular}{l|l|l|l|l|l}
				{\bf City} & {\bf Method}& \makecell{\bf Unified \\ \bf Cost}& \makecell{\bf Service \\ \bf Rate}  & \makecell{\bf \#Shortest \\ \bf Path Queries} & \makecell{\bf Time  \\ (s) } \\ \hline \hline
				\multirow{3}{*}{Cainiao} & \textit{SARD} & 351,22K & 70.84\% & 782,827K & 109.870  \\
				& \textit{SARD-O} & 351,58K & 70.75\%  & 454,887K & 72.543
				\\\hline
			\end{tabular}
		}
	\end{center}
\end{table}

\noindent\textbf{Memory consumption.}
Figure~\ref{fig:memory_usage_cainiao} illustrates the memory usage of tested traditional algorithms under default parameters.
Online-based algorithms, following a first-come-first-serve approach, consume less memory. However, parallelization incurs additional overhead due to the need to maintain a lock for each worker.
In contrast, batch-based algorithms require extra storage to hold request combinations for each batch. For instance, RTV employs an RTV-Graph, GAS uses an additive index, and SARD utilizes a shareability graph.
RTV's reliance on an integer linear program results in memory usage more than 20 times that of GAS and SARD.
The latter two algorithms use similar memory for their shareability graphs.
\begin{figure}[h!]\centering
	\subfigure{
		\scalebox{0.32}[0.32]{\includegraphics{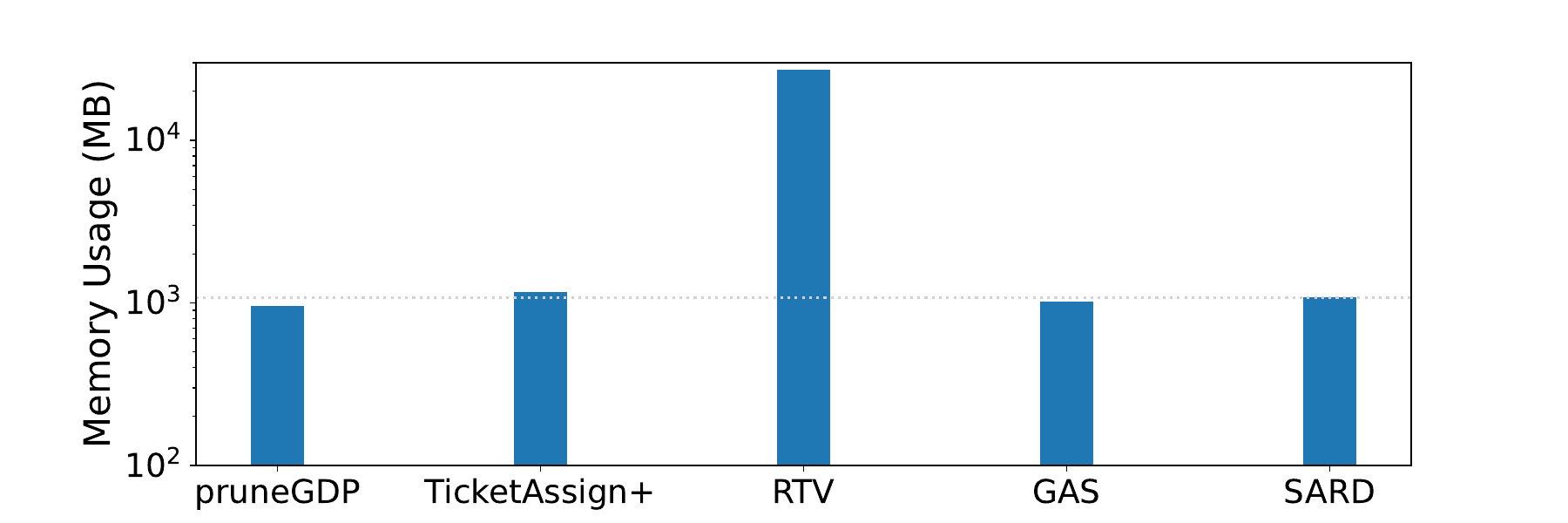}}}
	\caption{\textit{Cainiao} Memory Consumption.}
	\label{fig:memory_usage_cainiao}
\end{figure}

\noindent\textbf{Effect of vehicle's capacity constraint.}
The first column in Figure~\ref{fig:vary_cap} presents the outcomes of adjusting vehicle capacity from 2 to 6.
{SARD} and {GAS} demonstrate superior performance, achieving a minimum of $22.26\%$ reduction in unified cost compared to alternative algorithms.
All evaluated algorithms exhibited enhanced service quality due to expanded sharing opportunities.
In terms of service rate, {SARD} outperforms all other tested algorithms, achieving a $13.51\%\sim 42.80\%$ higher service rate.
Regarding execution time, {SARD} emerges as the most efficient among batch methods (RTV, GAS, SARD), executing $2.45\times\sim 44.19\times$ faster than {RTV} and {GAS}.

\subsection{Experimental Study on Vehicle Capacity Distribution}
To evaluate the efficacy of ride-sharing  under scenarios where different vehicles can have different capacities, we set the variance parameter $\sigma$ as shown in Table~\ref{tab:cainiao_settings}.
This parameter generates diverse vehicle capacity distributions adhering to a normal distribution with a mean of 4 and varying variances. 
Our previous default configuration is considered to have a variance of 0.
We conduct tests on our \textit{Cainiao}, \textit{NYC}, and \textit{CHD} datasets, as shown in Figure~\ref{fig:vary_cap} and \ref{fig:vary_var}. 
All algorithms remain stable across the three metrics, indicating that the distribution of different vehicle capacities had a negligible impact on ride-sharing quality.

\subsection{Effects of Angle pruning strategy}

Table~\ref{tab:sssp_pruning} shows the effect of the Angle pruning strategies proposed in Section~\ref{subsec:sn_gen} (parameters are in default values in Table \ref{tab:settings}).
We note the method without pruning strategies as \textit{SARD}, with the Angle pruning as \textit{SARD-O}. 
SARD-O saves up to 7.3\% of the shortest path queries and saves 5.2\% of the total running time on the two datasets compared with {SARD}. 
Besides, it has almost no harm on the service rate and unified cost.

\begin{table}[h!]
	\begin{center}
		{\scriptsize
			\caption{ Performance of Pruning Strategies.} 
			\label{tab:sssp_pruning}
			\begin{tabular}{l|l|l|l|l|l}
				{\bf City} & {\bf Method}& \makecell{\bf Unified \\ \bf Cost}& \makecell{\bf Service \\ \bf Rate}  & \makecell{\bf \#Shortest \\ \bf Path Queries} & \makecell{\bf Time  \\ (s) } \\ \hline \hline
				\multirow{3}{*}{CHD} & \textit{SARD} & 263,682K & 72.82\% & 2,971,425K & 699.8  \\
				& \textit{SARD-O} & 260,641K & 73.34\%  & 2,839,734K & 662.0
				\\\hline
				\multirow{3}{*}{NYC} & \textit{SARD} & 206,754K & 57.64\% & 1,373,047K  & 770.6  \\
				& \textit{SARD-O} & 206,470K & 57.83\%  & 1,270,116K & 750.6 
				\\\hline
			\end{tabular}
		}
	\end{center}
\end{table}

\end{document}